\documentclass[conference,compsoc]{IEEEtran}
\IEEEoverridecommandlockouts
\ifCLASSINFOpdf

\else

\fi
\usepackage{mathtools}
\usepackage{nicefrac}
\usepackage{siunitx}
\usepackage{array,framed}
\usepackage{booktabs}
\usepackage{
  color,
  float,
  epsfig,
  wrapfig,
  graphics,
  graphicx
}
\usepackage{kantlipsum}
\usepackage{textcomp,amssymb}
\usepackage{setspace}
\usepackage{latexsym,fancyhdr,url}
\usepackage{enumerate}
\usepackage[ruled,linesnumbered,noend]{algorithm2e}
\usepackage{algpseudocode}
\usepackage{graphics}
\usepackage{xparse} 
\usepackage{xspace}
\usepackage[utf8]{inputenc}
\usepackage{multirow}
\usepackage{csvsimple}
\usepackage{balance}
\usepackage[misc]{ifsym}
\usepackage{enumitem}
\usepackage{nopageno}

\DeclareMathOperator*{\argmin}{arg\,min}
\usepackage{
  tikz,
  pgfplots,
  pgfplotstable
}

\usepackage{hyperref}

\newcommand{\Prob}{\mathbb P}
\usetikzlibrary{
  shapes.geometric,
  arrows,
  external,
  pgfplots.groupplots,
  matrix
}

\usepackage{amsmath}
\usepackage{amssymb}
\usepackage{eucal}
\usepackage{comment}
\usepackage{subfigure}
\newcommand{\subhead}[1]{\vspace {0.04in}\noindent{\textbf{#1.}}}
\usepackage{soul}
\usepackage{amsthm}
\usepackage{multirow}
\newtheorem{theorem}{Theorem}[section]
\newtheorem{remark}{Remark}[section]

\newtheorem{lemma}[theorem]{Lemma}
\newtheorem{definition}{Definition}

\pgfplotsset{compat=1.18}

\begin{document}
\fancyhead{}

\hyphenation{op-tical net-works semi-conduc-tor}

\def\thetitle{DPI: Ensuring Strict \underline{\textit{D}}ifferential \underline{ \textit{P}}rivacy for \underline{\textit{I}}nfinite Data Streaming}
\title{\thetitle}

\author{
\IEEEauthorblockN{Shuya Feng\textsuperscript{$\dagger$}\textsuperscript{$*$}, Meisam Mohammady\textsuperscript{$\ddagger$}\textsuperscript{$*$}, Han Wang\textsuperscript{$\S$}, Xiaochen Li\textsuperscript{$\P$}, Zhan Qin\textsuperscript{$\P$}, Yuan Hong\textsuperscript{$\dagger$}\textsuperscript{\Letter}}
\IEEEauthorblockA{
\textsuperscript{$\dagger$}University of Connecticut, \textsuperscript{$\ddagger$}Iowa State University, \textsuperscript{$\S$}University of Kansas, \textsuperscript{$\P$}Zhejiang University \\
\textsuperscript{$\dagger$}\{shuya.feng, yuan.hong\}@uconn.edu, \textsuperscript{$\ddagger$}meisam@iastate.edu, \textsuperscript{$\S$}han.wang@ku.edu, \textsuperscript{$\P$}\{xiaochenli, qinzhan\}@zju.edu.cn
}
\IEEEauthorblockA{\textsuperscript{$*$}Equal Contribution} 
}

\maketitle

\begin{abstract}
Streaming data, crucial for applications like crowdsourcing analytics, behavior studies, and real-time monitoring, faces significant privacy risks due to the large and diverse data linked to individuals.
In particular, recent efforts to release data streams, using the rigorous privacy notion of differential privacy (DP), have encountered issues with unbounded privacy leakage. This challenge limits their applicability to only a finite number of time slots (``finite data stream'') or relaxation to protecting the events (``event or $w$-event DP'') rather than all the records of users. A persistent challenge is managing the sensitivity of outputs to inputs in situations where users contribute many activities and data distributions evolve over time.
In this paper, we present a novel technique for \textit{D}ifferentially \textit{P}rivate data streaming over \textit{I}nfinite disclosure (DPI) that effectively bounds the total privacy leakage of each user in infinite data streams while enabling accurate data collection and analysis. Furthermore, we also maximize the accuracy of DPI via a novel boosting mechanism. 
Finally, extensive experiments across various streaming applications and real datasets (e.g., COVID-19, Network Traffic, and USDA Production), show that DPI maintains high utility for infinite data streams in diverse settings. Code for DPI is available at \url{https://github.com/ShuyaFeng/DPI}. 
\end{abstract}

\section{Introduction}
\label{sec:intro}
Most computer systems and applications continuously generate data streams for analyses. For instance, online media such as YouTube \cite{ameigeiras2012analysis} or Instagram recommends videos to users based on their vast visiting records. Network monitoring \cite{thottan2010anomaly} detects abnormal behavior and identifies anomalies while tracking the network traffic. Water management \cite{ranjitha2015streaming} makes timely decisions on purifying water per the real-time water quality data collected by sensors. 

Existing solutions have not effectively addressed the privacy threats posed by streaming data, which is vulnerable due to its volume and diversity. Recent attempts to release data streams using the widely adopted privacy model, differential privacy (DP), have revealed persistent privacy leakage \cite{dwork2010differential, bolot2013private, chan, pegasus, 10dwork2010differential, perrier2018private}. In particular, Dwork et al. \cite{dwork2010differential} (STOC'10) found that in streaming data, as the number of rounds increases, an unavoidable logarithmic increase in error for an advancing counter occurs. Similarly, a logarithmic growth in privacy leakage is observed when the error is bounded. This issue arises from the infinite accumulation of two key factors in DP mechanisms: sensitivity (the maximum impact of each user on the result) and privacy budget composition (additional leakage from sequential result releases).

In the past decade, several solutions have been proposed to address ``event-level privacy on finite or infinite streams'' \cite{bolot2013private, chan, pegasus, 10dwork2010differential, perrier2018private, wang,10.14778/2732977.2732989, wang2016rescuedp}, or ``user-level privacy on finite streams under certain settings'' (e.g., partially disclosing some randomized results and applying interpolation \cite{fan2013adaptive,fanKellaris2013differentially,WangXH20,WangHKV20,LiuXWHBM21}). However, they have limitations in practice. The first category struggles when users' contiguous events collectively reveal sensitive information. Event-level privacy safeguards individual location visits but fails to protect a user's path across successive, potentially indefinite timestamps, like a sequence of locations. The second category is less practical for continuously and indefinitely collected data streams, since it is unrealistic to predict when services like traffic reporting will cease.

Thus, designing user-level DP mechanisms for infinite streaming involves addressing five fundamental challenges.

 \vspace{0.05in}

\begin{itemize}
        \item \textbf{Unboundedness}. The data size in the stream is continuously growing and unbounded, with each user's data potentially generated indefinitely, making it also unbounded in infinite data streams. 

        \vspace{0.05in}
        
    \item \textbf{Dynamic Data Distribution}. The data distribution changes dynamically over time, requiring DP mechanisms to consistently offer accurate analysis and privacy guarantees upon these changes.

   \vspace{0.05in}

\item \textbf{Accumultative Privacy Loss over Infinite Streams}. This has been widely recognized but not addressed yet. 

   \vspace{0.05in}
    
    \item \textbf{Utility-Privacy Tradeoff}. It is challenging to minimize the utility loss continually (under a bounded leakage) while leveraging dynamic changes in data distribution.

   \vspace{0.05in}
    
    \item \textbf{Real-Time Processing}. DP mechanisms for infinite data streams must balance real-time processing needs with privacy protection and computational efficiency.

    \vspace{0.05in}
\end{itemize}

\begin{table}[!h]
\caption{Representative DP data streaming methods (not scalable to user-level protection over infinite data streams). Other existing methods share similar properties.}
\label{table:sota}
\resizebox{\columnwidth}{!}{%
\begin{tabular}{|c|c|c|c|c|}
\hline
\multirow{2}{*} {Methods} &
\multicolumn{1}{c|} {Assumption} &
\multicolumn{3}{c|} {Properties in Representative Works}\\
\cline{2-5}
& Sensitivity & Privacy 
 & Error Bound & Applications \\
\hline
\cite{chan} & ($1$, event-level) & $O(\epsilon\cdot t)$ & $O(\frac{\log t^{1.5} \cdot \log \frac{1}{\delta} }{\epsilon})$ & count, top-k, range queries \\
\cite{pegasus} & ($1$, event-level)  & $O(\epsilon\cdot t)$ & - & count, event monitoring\\
\cite{perrier2018private} & ($1$, event-level)  & $O(\epsilon \cdot t)$ & $O(\tau(\log t)^{1.5})$ & count, range queries, etc.\\
\cite{wang} & ($1$, event-level) & $O(\epsilon \cdot t)$ & - & multiple analyses\\
\hline
\end{tabular}
}\vspace{-0.1in}
\end{table}

To our best knowledge, existing methods have not adequately addressed the significant challenges in infinite data streams. Challenges in current methods include not supporting infinite data disclosure, assuming limited sensitivity instead of indefinite full event sequence protection, and limiting applications to fixed time slot queries. Table~\ref{table:sota} provides an overview of the most crucial properties of representative existing works on data streaming methods with differential privacy. There are many limitations in existing methods. For instance, they all consider the special case of sensitivity (e.g., 1 for counting) and their privacy bound linearly increases with the number of rounds $t$. 

In this paper, we present a novel DP framework for infinite data streams, known as DPI (\emph{\underline{D}ifferential \underline{P}rivacy for \underline{I}nfinite Data Streams}). DPI is designed with three key components to address the above challenges: (1) \textit{Sensitivity Compression}, (2) \textit{DPI Boosting Mechanism}, and (3) \textit{Privacy Budget Allocation Mechanism}, discussed as below.  

\begin{figure*}[!tbh]
	\centering
		\includegraphics[width=1\linewidth]{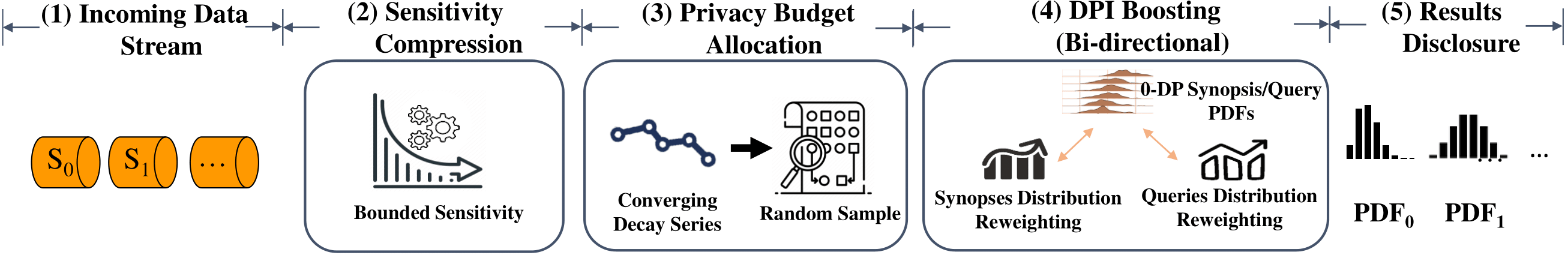}
  \vspace{-0.1in}
	\caption[Optional caption for list of figures]
	{DPI guarantees limited privacy bound and sensitivity by managing the pace of convergence and error margins, striving to attain the highest level of accuracy (\emph{the output streaming PDFs can support most real-time applications as the input stream}).}\vspace{-0.05in}
	\label{fig:intro}
\end{figure*}

\subsection{DPI in a Glimpse}
\label{subsec1-intro}

\noindent\textbf{Sensitivity Compression.} Bounding sensitivity is crucial for reliable DP models. DPI adopts a \emph{probability distribution data structure} to effectively limit the unbounded sensitivity by projecting/compressing the maximum output difference to the $\ell_1$-norm of probability distribution difference, offering advantages over unbounded structures like counting.  

\subhead{DPI Boosting Mechanism} 
The design of DPI has two inherent crucial challenges: (1) preserving a (lower-bound) meaningful utility for the data stream, and (2) developing methods to capture dynamic patterns within the new coming data. A sequence of randomization mechanisms (e.g., additive noise \cite{Dwork08differentialprivacy:}, sampling \cite{HongVLKG15} or randomized response \cite{ErlingssonPK14}) for all the time slots, inevitably destroys both cumulative and instantaneous information. For instance, applying a series of $\epsilon$-DP Laplace mechanisms over a dataset $D$ will generate results with an error of $O\left(2^n/\left(|D|\cdot\epsilon\right)\right)$~\cite{10.1561/0400000042}. 

Unlike atomic DP mechanisms, additive or accumulative learning algorithms are especially effective for streaming data, where they continuously update models in real-time to adapt to evolving patterns. In particular, Dwork et al.~\cite{dwork2010boosting} and Moritz et al.~\cite{5670948} demonstrated that the differential privacy version of boosting mechanisms can be effective in achieving cumulative learning outcomes for a finite number of queries. However, direct application of these approaches to continual streaming data results in \emph{unbounded privacy leakage} and \emph{significant computational overhead} \cite{dwork2010boosting,5670948}. We have identified that both deficiencies result from the use of an expensive \emph{base synopsis generator} in \cite{dwork2010boosting,5670948}.\footnote{Such generator acts as a compact representation of query answers over a dataset and can be implemented using various forms, such as a synthetic database or any data structure such as kernel or probability density function (PDF) that approximates query results (``synopsis'') \cite{dwork2010boosting,5670948}.} 

To handle infinite streams, we have devised a novel DPI boosting mechanism, which guarantees both bounded and efficient computation while integrating an infinite number of tiny-budget DP mechanisms into an additive learning game. While the outcomes of various queries on streaming data may vary over time, it is still possible to precisely associate each result with an adequately accurate state, such as a numerical value or a PDF. The DPI boosting mechanism introduces a groundbreaking approach to synopsis generation by creating a significantly smaller set of synopses from PDFs, resulting in faster processing. Such non-trivial synopsis generation is designed as a privacy-free (0-DP) mechanism, generating an ample number of synopses that precisely address a set of queries across an infinite number of rounds. It necessitates a rigorous analytical analysis to determine the number of synopses required to achieve a meaningful level of accuracy while ensuring privacy. Also, the existing uni-directional boosting mechanism \cite{dwork2010boosting,5670948} (updating weights solely over queries based on data-dependent synopses) should be revised into a \emph{bi-directional boosting} approach that considers both data-independent synopses and queries for infinite data streaming (\emph{otherwise, the privacy budgets will not be sufficient}). With these considerations in DPI, we aim to achieve provable utility while maintaining the bounded privacy guarantees.

\subhead{Privacy Budget Allocation and Boosting}
By saving budget on the synopsis generation process in DPI, we develop a bounded budget boosting mechanism that optimally utilizes an infinite number of tiny boosting mechanisms for making refinements, focusing the privacy budgets solely on \emph{reweighting} and not synopses. To achieve this, we establish an optimal allocation strategy with a converging sum series to maximize the utility.

To prevent budget depletion after a certain number of releases, typically expected with a decaying series, we employ a sophisticated deployment of this infinite decaying series of privacy budgets. During each iteration (time slot) of DPI, we randomly generate the privacy budget from the series, allowing for efficient and controlled management of the privacy budget throughout the continuous process.

Figure \ref{fig:intro} shows an overview of DPI, including the input data stream, the three key components, and the output results (as streaming PDFs). The DPI framework takes an incoming data stream as input and processes it through sensitivity compression, privacy budget allocation, and a boosting mechanism to output an infinite series of differentially private results that maximize accuracy and utility over time (see the detailed design in Section \ref{sec:workflow}). This end-to-end process allows DPI to ensure strong privacy for users while handling the challenges of infinite streaming data. Note that DPI may require a restart to refresh the privacy budget in possible uncommon cases of extremely long-term, highly dynamic data streams, as discussed in Section \ref{sec:discussion}. 

\subsection{Contributions}
\label{subsec2-intro}

\noindent To sum up, DPI makes the following major contributions:

\subhead{Unique and Significant Benefits} 
First, to our best knowledge, DPI provides the first user-level DP solution for protecting users in applications relying on continuously collecting data over infinite streams, e.g., YouTube \cite{ameigeiras2012analysis} or Instagram recommends videos to users based on their vast visiting records indefinitely. Second, unlike most existing works which usually take various limited assumptions (e.g., bounded sensitivity, protecting the presence/absence statuses of events rather than users, a limited number of disclosures), DPI strictly satisfies generic $\epsilon$-DP for protecting all the users (end-to-end) involved in the unlimited data streams. Third, DPI provably maximizes the accuracy of the disclosure in the long run with substantial theoretical analyses. 
      
\subhead{Novel Methods} 
We have designed a novel dynamic boosting mechanism (DPI boosting) that optimally balances both privacy and utility over time. DPI significantly differs from the ``AdaBoosting''~\cite{dwork2010boosting} for DP: (1) DPI replaces the base synopsis generator with a pre-determined set of synopses \emph{independent of the dataset}. We demonstrate that these synopses can be created by the performance of the output and strong synopses can be integrated into the boosting mechanism with no impact on privacy, and (2) DPI boosting has been designed to generate a PDF data structure that tracks the overall PDF of the stream instead of accurately answering pre-defined queries. We have also established lower bounds on the utility of DPI based on this PDF data structure.

\subhead{Comprehensive Evaluations} 
We comprehensively evaluate the performance of DPI on different applications and datasets. Our primary objectives in these evaluations are: (1) investigating the impact of DPI's hyperparameters on its performance over time (in Section \ref{sec:tuning}), (2) evaluating DPI on different datasets with different domain sizes/distributions and a wide variety of real applications, such as statistical queries, anomaly detection, recommender systems (in Sections \ref{sec:utility_eva} and \ref{sec:highly}), and (3) compare DPI with existing methods and show that they violate the privacy requirement $\epsilon$ after a few rounds (in Appendix \ref{sec:sotaloss}) while DPI strictly bounds the privacy leakage with $\epsilon$ indefinitely and still has considerable share of unspent privacy budget after a large number of rounds.

\section{Preliminaries}
\label{sec:pre}
\subsection{Streaming Model}
\label{subsec:system}

We consider data collection and analysis in the context of interactive and continuous observation over infinite data streams. In this setting, a trusted curator maintains a dynamic dataset, represented as a collection of streaming input $S = \{S_1, S_2, \ldots\}$, where $S_t$ refers to the indexed input collected at time slot $t$. The universe of possible individual records is denoted as $X$.

\subhead{Interactive Queries} 
The data analyst interacts with the curator by submitting a sequence of queries over time. Let $Q$ be the space of all possible queries, and at each time slot $t$, the data analyst requests a subset of queries $Q_t \subseteq Q$. Each query $q_t$ is a function that takes the data until the current stream $S_t$ and outputs a response $y_t = q_t(S_1, S_2, \ldots S_t)$, which is released by the curator. Note that $q$ can remain the same as previous time slots or dynamically change.

\subsection{Privacy Models}
\label{subsec:privacy}
To protect the privacy of individual records in the context of continual observation, the trusted curator employs Differential Privacy (DP) mechanisms represented by $\mathcal{M}(\epsilon)$, where $\epsilon > 0$ is the privacy budget allocated for each query in the infinite data stream. Each query $q_t$ outputs a response $y_t = q_t(S_1, S_2, \ldots, S_t)$, and the curator releases a private response denoted as $\tilde{y_t}$. The DP mechanism $\mathcal{M}(\epsilon_t)$ ensures that $\tilde{y_t}$ does not disclose any information about the presence or absence of sensitive information from any particular individual in all the streamed data. By limiting the privacy budget $\epsilon_t$, the curator controls the amount of privacy provided for each query $q_t$. 
 
We adopt the standard definition of $\epsilon$-DP~\cite{Dwork08differentialprivacy:}. Let $\Gamma$ be the domain of datasets, and let $D, D'$ be two datasets in $\Gamma$ such that $D'$ can be obtained from $D$ by adding or removing the data of a single individual (user). A randomization mechanism $\mathcal{M}: \Gamma \times \Omega \to \mathbb{R}$ is $\epsilon$-differentially private if, for all $\Theta \subseteq \mathbb{R}$ and for all $D, D' \in \Gamma$ such that $D'$ is adjacent to $D$, i.e., $D$ and $D'$ differ in the data of a single individual, the following inequality holds:

\vspace{-0.1in}

\begin{align}
\Prob(\mathcal{M}(D) \in \Theta) \leq e^{\epsilon} \Prob(\mathcal{M}(D') \in \Theta)
\end{align}

where $\Omega$ is a sample space by randomization $\mathcal{M}$ to generate the output. To achieve differential privacy, we need to consider the concept of sensitivity. The sensitivity $\Delta$~\cite{dwork2006calibrating} of a query $q$ is defined as the maximum $\ell_1$-distance between the exact query answers on any two neighboring databases $D$ and $D'$, i.e., $\Delta q = \max_{D,D'} \left\|q(D) - q(D')\right\|_1$. 
Bounding sensitivity is crucial for DP models. Even with a constrained privacy budget, unbounded sensitivity can result in overly sensitive outputs, risking privacy breaches. 

\subsection{Boosting Mechanisms}
\label{synopsis1}

We next introduce some background for the AdaBoosting mechanism and DP with AdaBoosting, which will be significantly revised on both \emph{privacy} and \emph{efficiency} for DPI.  

\subsubsection{AdaBoosting Mechanism}

AdaBoosting (Adaptive Boosting) \cite{dwork2010boosting} is a powerful machine learning technique used to enhance the accuracy of weak classifiers by combining them into a strong ensemble classifier. Formally, given a training dataset $D = {(x_1, y_1), (x_2, y_2), \ldots, (x_n, y_n)}$, where $x_i$ and $y_i$ represent the data points and their respective labels, AdaBoosting constructs an ensemble classifier $H(x) = \text{sign}\left(\sum_{t=1}^T \alpha_t h_t(x)\right)$. Here, $T$ is the number of weak classifiers, $h_t(x)$ is the $t$-th weak classifier, and $\alpha_t$ is the corresponding weight assigned to $h_t(x)$. The objective of boosting is to convert a weak learner, which produces a hypothesis only slightly better than random guessing, into a strong learner that achieves high accuracy. The AdaBoosting proceeds in rounds. In each round, two steps are performed:

\vspace{0.05in}

\begin{enumerate}
    \item  The base learner is run on the current distribution $\mathcal{D}_t$, generating a classification hypothesis $h_t$.

\vspace{0.05in}
    
    \item  The hypotheses $h_1, h_2, \ldots, h_t$ are used to reweight the samples and define the updated distribution $\mathcal{D}_{t+1}$.
\end{enumerate}

It continues for a predefined number of rounds or until the combination of hypotheses is sufficiently accurate. 

\subsubsection{DP with AdaBoosting}
DP with AdaBoosting, introduced by Dwork and Rothblum \cite{dwork2010boosting}, combines AdaBoosting with differential privacy. As detailed in Figure~\ref{fig:dwork} (in Appendix \ref{app:ada}), the process iteratively updates the weights assigned to each query, and samples queries based on their weights as follows. Let $Q = \{q_1, q_2, \ldots, q_m\}$ represent a set of queries. The method assigns uniform weights $w_i$ to each query $q$ in $Q$ in its initialization phase ($P_Q(t=1)\rightarrow$ uniform). Next, it will operate on the ``base synopsis generator'' to produce ``weak'' DP synopses $\mathcal{A}^t_{Q_s}$, which are moderately accurate for $r$ sampled queries $Q_s$. 

\begin{definition} [Sampled Queries $Q_s$]
\label{def:endanger}
Let $Q$ be the set of requested queries, and $P_Q(t)$ be the PDF for the query sampling at iteration $t$ in DP AdaBoosting~\cite{dwork2010boosting}. At each iteration $t$, the set of sampled queries $Q^t_s$ is defined as a $k$-size subset of $Q$, sampled from $Q$ per $P_Q(t)$. The sampling process reflects AdaBoosting's practice of assigning larger weights to weaker queries to improve their performance. 
\end{definition}

\begin{definition}[Synopses] 
\label{def:synopses}
A synopsis \cite{dwork2009complexity,dwork2010boosting} generated from a dataset $D$ is defined as a compact representation of answers to a set of queries $Q=\{q_1,q_2,\cdots, q_k\}$ over $D$. Synopsis can be in the form of a synthetic database, or any arbitrary data structure, e.g., a kernel or PDF, that can be queried to return an approximation of the query's result.
\end{definition}

This data structure is referred to as a \textit{synopsis} of the database. General privacy-preserving synopses are of interest as they may be easier to construct than privacy-preserving synthetic databases. Also, stronger limitations are known for constructing synthetic databases than general privacy-preserving synopses.

\begin{definition} [Base Synopsis Generator~\cite{dwork2010boosting}]
For any data universe \(X\) and set of queries \(Q : \{X^n \rightarrow \mathbb{R}\}\) with sensitivity \(\Delta_Q\), a \((r, \lambda, \varepsilon_1, \kappa)\)-base synopsis generator for \(Q\) is a method that:
\begin{itemize}
    \item Outputs a synthetic database \(y\) of size \(n\) based on \(r = O(n \cdot \log |X| \cdot \kappa)\) queries from \(Q\), where \(\kappa\) is the probability that the chosen synopsis is \(\lambda\)-accurate (the synopsis closely approximates the true data within a \(\lambda\) margin of error).
    
    \item Adds Laplace noise $O(\kappa \sqrt{r}\cdot \frac{\Delta_Q}{\varepsilon_1})$ to each query answer.
    
    \item Guarantees \((\varepsilon_1, \exp(-\kappa))\)-differential privacy.
    
    \item It has a complexity of \(|X|^n \cdot \text{poly}(n, \kappa, \log(1/\varepsilon_1))\). 
\end{itemize}

The base synopsis generator outputs the lexicographically first database \(y\) satisfying \(|q(y) - \hat{q}^t_s(x)| \leq \lambda/2\) for every query \(q\) in the sampled query set \(Q_s\), where $\hat{q}^t_s$ is the noisy query answers. If no such database exists, it outputs \(\bot\) (i.e., fails). The \((r, \lambda, \varepsilon_1, \kappa)\)-base synopsis generator privately generates a synthetic database approximating the responses of the original database to a set of queries \(Q\).
\end{definition}
    
The goal is to iteratively improve the base generator using the boosting algorithm to create a synopsis that performs well for all requested queries $Q$ while preserving privacy. 

\begin{figure*}[!h]
	\centering
		\includegraphics[width=1\linewidth]{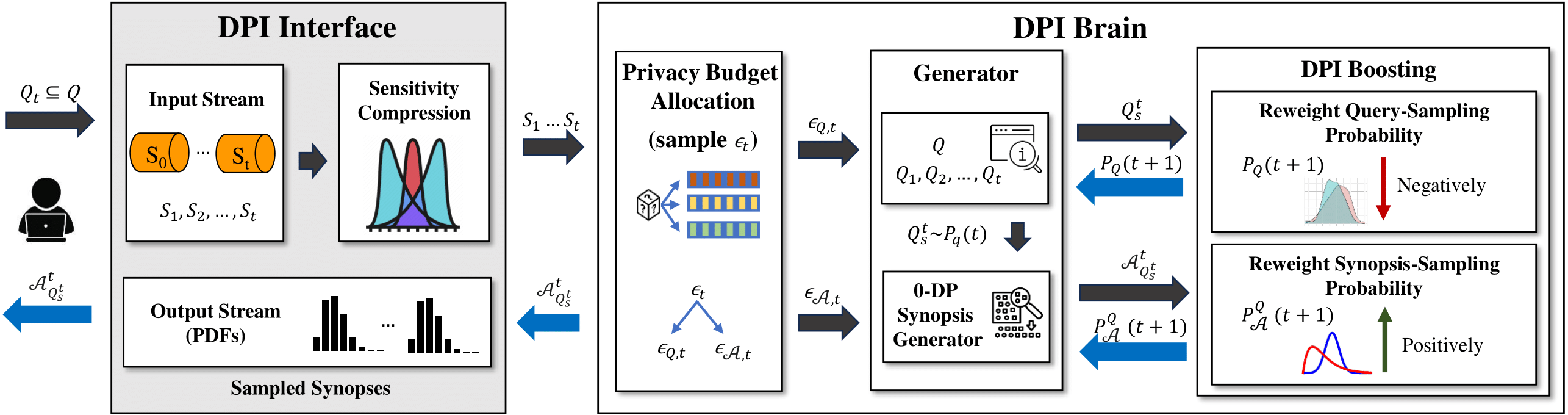}
  \vspace{-0.1in}
	\caption[Optional caption for list of figures]
	{The DPI framework in time slot $t$. (1) $Q^t_s$ can be a single query or a subset of queries, and can remain the same or dynamically change over time, (2) all the privacy budgets are pre-computed with converging series and $\epsilon_t$ is sampled in real-time, and halved for query-sampling and synopsis-sampling, (3) generator returns the PDF output (sampled synopsis) for the current time slot $t$ to the interface/analyst, and also sends it for boosting, and (4) DPI boosting (bi-directional) reweights the PDFs (both query-sampling probability and synopsis-sampling probability) for the next time slot $(t+1)$.}\vspace{-0.1in}
	\label{fig:flowchart}
\end{figure*}

The next steps aim to strengthen the queries with bad results by adaptive sampling and refining them. At each iteration, the query sampling probability $P_Q(t)$ is reweighted based on the performance of the base synopsis generator over all queries such that the probability modification is to iteratively assign higher weights to the (new) queries with bad results, allowing the subsequent weak model to focus on these challenging examples and improve overall accuracy. Finally, it integrates a learning rate $\eta^Q_t$ in the reweighting process, which is proportional to the privacy budget allocated for boosting ($\epsilon_2$). Then, two separate privacy budgets are required: one for generating synopses ($\epsilon_1$) and the other for the reweighting step ($\epsilon_2$). Consequently, AdaBoosting with DP \cite{dwork2010boosting} becomes a privacy-expensive mechanism. 

DP with AdaBoosting's runtime depends on the number of queries and the base synopsis generator's runtime. The accuracy of the mechanism scales with $\sqrt{n}\log m$ ($n$: size of the data, $m$: number of queries). For details of the synopses and base synopsis generator, see \cite{dwork2009complexity, dwork2010boosting}.

\section{Overview of DPI Framework}
\label{sec:frame}
In this section, we present the DPI framework comprising the DPI Interface and the DPI Brain. The DPI Interface manages temporal input streams, stores sequential DP results for pre-defined queries (``query synopses''), and delivers query results to analysts over time. On the other hand, the DPI Brain focuses on online learning. The details of the DPI framework are shown in Figure \ref{fig:flowchart}.

\subsection{DPI Components}

\noindent\textbf{DPI Interface} plays a central role in managing the interaction between analysts and the underlying data, including:

\begin{itemize}
  \item \emph{Input Module} collects temporal data streams from users over successive timestamps. It also receives temporal queries $Q_1,\dots, Q_t$ from the analyst.

\vspace{0.05in}

\item \emph{Sensitivity Compression} processes the input data stream with a PDF data structure to bound sensitivity. 

\vspace{0.05in}

  \item \emph{Output Module} returns temporal output query synopses as a sequence of PDFs with DP guarantees for a pool of queries $Q_1,\dots, Q_t$, 
   which can be converted to support most data analyses on the input stream. 

\end{itemize}

\vspace{0.05in}

\noindent\textbf{DPI Brain} comprises the following modules.

\begin{itemize}

\item \emph{Privacy Budget Allocation} generates the privacy budget $\epsilon_t$ for each time slot $t$. 

\vspace{0.05in}

    \item \emph{Generator Module} generates the sampled queries ($Q_s^t$) and a set of sampled synopses ($\mathcal{A}^t_{Q_s^t}$) for each query $q \in Q_s^t$. It interacts with the DPI Boosting module, receiving the updated configurations in terms of PDFs to effectively adjust the sampling process.

\vspace{0.05in}

    \item \emph{DPI Boosting Module (bi-directional)} effectively captures the dynamic nature of data streams (new coming data and queries) when answering queries. 

\end{itemize}

\subsection{DPI Workflow}

\vspace{0.05in}

\begin{itemize}
    \item \textbf{Step 1:} \emph{Input Streaming and Sensitivity Compression.} See details in Section \ref{compress}. 

\vspace{0.05in}

\item \textbf{Step 2:} \emph{Privacy Budget Allocation.} In time slot $t$, DPI randomly samples a pre-computed privacy budget $\epsilon_t$ from the series, halved for query-sampling and synopsis-sampling (see details in Section \ref{sec:privacyallocation}).

\vspace{0.05in}
    
    \item \textbf{Step 3:} \emph{0-DP Synopsis Generation.}  The Generator module generates a set of parameters, including a sampled subset of queries $Q^t_s$ over the pre-defined pool of queries $Q$ and a set of sampled synopses $\mathcal{A}^t_{Q_s^t}$ for queries $q \in Q^t_s$. To derive the sampled synopses $\mathcal{A}^t_{Q_s^t}$, we need to generate a pool of synopses $\mathcal{A}$ that can represent all possible distributions in a dataset. This generation can avoid consuming privacy budget (0-DP, see details in Section~\ref{sec:synopsis2}). 

\vspace{0.05in}
    
    \item \textbf{Step 4:} \emph{DPI Boosting.} The DPI Boosting module improves accuracy via bi-directional reweighting of queries and synopses (see details in Section \ref{sec:dpboosting}).
\begin{itemize}

  \item  \emph{Reweighting Synopsis-Sampling Probability.} In this phase, the DPI Boosting module uses a portion of the privacy budget allocated for DPI to perform reweighting of the Synopsis-Sampling Probability. This step aims to improve the accuracy of synopses and update the sampling PDF accordingly.

\vspace{0.05in}

    \item \emph{Reweighting Query-Sampling Probability.} In this phase, another portion of the privacy budget is used to reweight the Query-Sampling Probability. This step ensures that the most informative queries are given higher probabilities in the sampling process.
\end{itemize}

\vspace{0.05in}

    \item \textbf{Step 5:} \emph{Output Private Streams (PDFs).} DP query results (temporal) are provided to analysts over time.
\end{itemize}

\section{Detailed Design in DPI}
\label{sec:workflow}
\subsection{Sensitivity Compression}
\label{compress}

To enable effective privacy for infinite data streams, the sensitivity of the analysis must be bounded. Sensitivity represents the maximum change in the output resulting from the presence or absence of any single user's data. This motivates the need for sensitivity compression in the infinite stream setting. Our key insight is adopting a PDF data structure to represent the data stream. This provides an inherent sensitivity bound, as changing one user's data can only alter the PDF by a maximum total variation distance of $2$ (considering the two PDFs of two adjacent inputs are two normalized vectors with all the entries summing up to 1 in each). Denote $D$ and $D'$ as the output distributions of two neighboring inputs for any query at time slot $t$. Thus, for all $D$ and $D'$, their sensitivity is derived as $\Delta_Q = ||D - D'||_1 \leq 2$. By bounding the PDF difference, we effectively compress sensitivity to a fixed constant for all the queries. This allows formal privacy guarantees to be achieved over infinite data. Furthermore, the PDF representation maintains high utility for the analyses in most real-time applications.

\subsection{0-DP Synopsis Generator}
\label{sec:synopsis2}

The synopsis generation in DPI is responsible for generating a sampled synopsis for each query $q \in Q_s$. Our key observation is that establishing a universal synopses pool can be done in a privacy-free (0-DP) and efficient way, enabling the system to privately answer a wide range of queries.

\subhead{Establishing a Pool of PDFs} 
DPI utilizes a pioneering 0-DP synopsis generation process to navigate this challenge (AdaBoosting \cite{dwork2010boosting} cannot solve). This method constructs a set of PDFs that can universally emulate any PDF with minimal error (see detailed analysis in Appendix~\ref{subsec:pool}). These PDFs are then employed to generate answers for a set of queries, delivering approximated results without requiring access to raw data. Various types of queries can be addressed using these PDFs, including but not limited to Point Queries, Range Queries, Aggregation Queries, Top-K Queries, Join Queries, and Correlation Queries \cite{wang2013survey, soliman2008probabilistic,ilyas2004supporting,li2003multi}. 

\subhead{Quantization} 
The first step in the 0-DP synopsis generation process is quantization, which discretizes the values that PDFs can take. We define the precision level of quantization as $p$, representing one unit of probability (e.g., $p = 0.001$). By dividing the range of possible values into intervals of size $p$, we can express the quantized value of a continuous variable $x$ as $\text{quantize}(x) = \text{round}(x/p) \cdot p$. This ensures that the values are discretized into intervals of width $p$.

\subhead{Sampling Synopses with Multinomial Distribution} 
After quantization, we can generate all possible PDFs by distributing the available units of probability (denoted as $n_p = \frac{1}{p}+1$) among the domain of interest, denoted as $\mathcal{K}$. Since $n_p$ should be the same as the dataset size $n$, we can derive the value of $p$. Then, we can model this distribution using a multinomial PDF, denoted as $P(x)$, where $x$ represents the random variable that takes on values from the domain $\mathcal{K}$. The multinomial PDF captures the probabilities of outcomes for this discrete random variable with $k$ categories.

During each multinomial sampling, we treat each category equally and sample them uniformly for $n_p$ times, leading to a generation of a synopsis (the distribution of $n$ data points with $k$ categories). To create the 0-DP synopses pool, we leverage the multinomial sampling. Let $n$ be the number of samples drawn from the multinomial distribution and denote the sampled values as ($x_1, x_2, \ldots, x_n$), where each $x_i$ represents a possible outcome from the multinomial distribution defined by $P(x)$. We should repeat the multinomial sampling process for $N$ times to get the pool of synopses ($P_1, P_2, \ldots, P_N$). The algorithm of $0$-DP synopsis generator is shown in Appendix \ref{subsec:pool}.

Next, for queries in $Q^t_s$, a set of synopses $\mathcal{A}^t_{Q^t_s}$ should be sampled according to the most recently updated sampling probability $P_\mathcal{A}(t)$. We will discuss how to update the sampling weights for different synopses in the DPI Boosting. The DPI Boosting algorithm iteratively improves the accuracy of DP queries on streaming data through a bi-directional reweighting process that favors challenging queries and high-quality synopses. By continually reweighting query and synopsis sampling probabilities based on estimated errors, DPI Boosting adapts the new coming data to optimize accuracy under DP constraints.

\subhead{0-DP Guarantee for the Synopsis Generator}
The synopsis generation process in DPI is uniquely designed to ensure that it does not consume any privacy budget since it operates without accessing to any raw data. Initially, the Synosis Generator invokes quantization to discretize potential values for any PDFs without using any raw data. Subsequently, through the procedure of multinomial synopses sampling, all potential PDFs are generated by distributing available units of probabilities across a defined domain of interest, bypassing the need for raw data access. The detailed steps and theoretical analysis are given in Appendix \ref{subsec:pool}.

\subsection{DPI Boosting (Bi-directional)}
\label{sec:dpboosting}

The DPI Boosting continuously evaluates and adjusts the importance of a subset of queries (``sampled queries''). It allows less accurate sampled queries to be reweighted higher, while the reweighting process is applied to focus more on accurate synopses for each sampled query. This bi-directional boosting process effectively captures the dynamic nature of data streams when answering queries. In this section, we will provide a comprehensive understanding of the DPI Boosting module and its role in enhancing the utility and privacy guarantees of the DPI.

\subhead{Bi-directional Boosting}
DPI Boosting significantly revises the AdaBoosting under DP constraints \cite{dwork2010boosting} to reweight the sampling weights. Specifically, we focus on the optimization techniques used to reweight the query and synopsis distribution during each round of DPI Boosting. The goal is to assign higher weights to poorly handled queries, and lower weights to well-handled queries, ultimately leading to enhanced accuracy and utility of the final synopsis. 

Recall that we sample a subset of synopses $\mathcal{A}^t_{Q^t_s}$ for the sampled subset of queries $Q^t_s$. In the DPI Boosting, we should evaluate the difference between the true query result and the query result based on the synopsis. Our weight boosting differs from the AdaBoosting \cite{dwork2010boosting} from two aspects: (1) there is a subset of synopses instead of one synopsis to reweight the sampling probability, and (2) the reweighted sampling probability should be used both for queries in a pool and synopses in a pool. The details of DPI boosting are shown in Algorithm \ref{algm:mergedDPI}. Notice that the size of query pool $Q$ may be different from the size of synopses pool $\mathcal{A}$ and it should not affect the reweighting process since normalization is applied for queries and synopses.  

\begin{algorithm}
\footnotesize
\caption{DPI Algorithm}
\label{algm:mergedDPI}

\KwIn{data stream $S$, privacy budget $\epsilon$, synopsis error bound $\lambda$, error bound for overfitting $\mu$, decay rate of decaying series $\zeta$, query pool $Q$, synopsis pool $\mathcal{A} $ (Algorithm \ref{algm:zero-syno}) }
\KwOut{private output streams $\mathcal{A}^t_{Q_s^t}$ (PDFs)}

\DontPrintSemicolon

Initialize $P_Q(1)$, $P_\mathcal{A}$ to uniform distribution 

\ForEach{$S_t$ in Stream $S$}{
    \tcc{Random Budget Allocation} 
    $\epsilon_t= \text{RBA}(t)$ \tcc{Algorithm \ref{algm:RBA1}}
    $\epsilon_{Q,t} = \epsilon_{\mathcal{A},t} = \frac{\epsilon_t}{2}$\;
    \tcc{$Q^t_s $ sampling} 
    $Q^t_s\sim   P_Q(t)$\;
       \tcc{$\mathcal{A}^t_{Q_s^t}$ sampling}  
       $\mathcal{A}^t_{Q_s^t} \sim P_\mathcal{A}(t) $\;
        
        \tcc{DPI Boosting}
        \If{$|q(S_t)-\mathcal{A}^t_{Q_s^t}| < \lambda $ }{$P_\mathcal{A}(t)[q] \leftarrow 1$, $P_Q(t) \leftarrow -1$}
        \If{$|q(S_t)-\mathcal{A}^t_{Q_s^t}| \geq \lambda + \mu $ }{$P_\mathcal{A}(t)[q] \leftarrow -1$, $P_Q(t) \leftarrow 1$}
        \Else{$P_\mathcal{A}(t)[q], P_Q(t) \leftarrow 1 - 2(|q(D)-\mathcal{A}^t_{Q_s^t}|-\lambda)/\mu$}
        $P_\mathcal{A}(t+1)$ = Normalize $P_\mathcal{A}(t)$\;
    }
    $u_{q,a} \leftarrow exp(-\alpha_t \cdot \sum_{j=1}^t P_\mathcal{A}(t)[q])$, where $\alpha_t=\frac{1}{2}\ln[(1+2\eta_t)/(1-2\eta_t)]$ and $\eta_t$ is the learning rate in time slot $t$\;
normalization factor $Z_t \leftarrow \sum_{q \in Q} u_{q,a}$\;
update $P_Q(t+1)= u_{q,a}/Z_t$, $P_\mathcal{A}(t)[q] = u_{q,a}/Z_t$\;
 
    \Return{$\mathcal{A}^t_{Q_s^t}$} as private output streams\;

\end{algorithm}

\subhead{Choice of Parameters in DPI Boosting} The reweighting depends on the parameters $\lambda$, $\mu$, and $\eta$, which all evaluate the accuracy of synopses. Note that $\lambda$ denotes the error bound of the chosen synopses, while $\mu$ refers to an additional error bound that is tolerated to avoid overfitting during DPI boosting. 
The parameter $\eta$ is determined by the specified privacy budget $\epsilon$. We find the optimal value of $\eta_t$ for each time slot with the given $\epsilon$ to minimize the error bound of DPI framework under the privacy constraints. The details are presented in Section \ref{sec:utility} (Theorem \ref{thm:final}). 

\subsection{Random Budget Allocation (RBA)}
\label{sec:privacyallocation} 
As outlined in Algorithm~\ref{algm:RBA1} (in Appendix \ref{sec:RBA2}) to mitigate the risk of \textit{systematic depletion} over time, DPI employs a random selection process when defining $\epsilon_t$ (``Random Budget Allocation (RBA)'') rather than sequentially picking elements from the infinite series of potential budget values. It employs an exponential PDF with $p(\epsilon_t=x) = \exp(-\Lambda \cdot x)$, where $\Lambda$ is an astronomically large number. This could generate very tiny samples, with the probability of obtaining such small values approaching 1.

After generating such very tiny values, they are mapped to the elements in the series $\varepsilon$ sampled from the exponential PDF. The closest one to the sampled value is then returned, denoted as $\epsilon_t \leftarrow \argmin\limits_{\epsilon \in \varepsilon_t } \left|\epsilon - S \right|$. We note that since RBA generates samples without replacement in the series, they tend to get smaller and smaller over time. However, the sum of $t$ independent samples from an exponential PDF with the rate parameter $\Lambda$ follows a Gamma distribution with a shape parameter $t$ and a rate parameter $\Lambda$. The mean ($\upsilon$) of a Gamma distribution with shape parameter $\omega$ and rate parameter $\theta$ is given by $\frac{\omega}{\theta}$. In our case, for the sum of $t$ exponential samples, given the shape parameter $t$ and the rate parameter $\Lambda$, a mean ($\upsilon$) of $\frac{t}{\Lambda}$ can be generated. Therefore, for any $t < \Lambda$, we can effectively preserve $\epsilon - \frac{t}{\Lambda}$. In practice, this bound must be much better as samples of the series are tending to be smaller than exponential PDF's samples due to sampling without replacement. An alternative RBA scheme based on \emph{ranges} is given in Appendix \ref{sec:RBA2}. 

Finally, as shown in Figure \ref{fig:flowchart}, each $\epsilon_t$ is further equally divided into $\epsilon_{Q,t}$ and $\epsilon_{\mathcal{A},t}$ for reweighting the query and synopses distribution, respectively. Specifically, it determines the learning rate that controls how aggressively the query sampling distribution is updated based on the empirical errors. Also, $\epsilon_{\mathcal{A},t}$ plays a key role in governing the privacy budget utilized when reweighting the synopses distribution $P_\mathcal{A}(t)$. By separating the budgets for querying and synopses sampling, DPI allows more flexible control to boost each of these two key components in the overall mechanism. Figure \ref{fig:remaining_budget} plots the remaining budgets after $t$ time slots and it clearly demonstrates that RBA is a thoughtful budget consumption strategy over the infinite disclosure.

\subhead{DPI Framework}
Algorithm \ref{algm:mergedDPI} presents the details for DPI, which outputs PDF(s) in each time slot to dynamic queries with $\epsilon$-DP over infinite streams. W.l.o.g., we take the single-query case for $Q^t_s$ as an example. A set of queries for $Q^t_s$ only need to split the budget for multiple output PDFs.

\section{Privacy and Utility Analysis}
\label{sec:design}
\subsection{Privacy Loss after $t$ Iterations}
\label{thm:privacydpi}
To understand the privacy guarantees of our DPI, we analyze the privacy budget consumed over successive disclosures. Theorem \ref{thrm:privacy} shows how to calculate the total privacy budget $\epsilon$ by DPI after $t$ iterations (e.g., time slots).

\begin{theorem}
\label{thrm:privacy}
Consider a series of DPI mechanisms, each triggered with $\eta^{Q}_i$ and $\eta^{\mathcal{A}}_i$. The overall privacy budget of DPI after $t$ iterations is given by:
\begin{equation}
    \epsilon = \frac{4}{\mu} \sum\limits_{i=1}^{t} \log \left[\left(\frac{1+2\eta^{\mathcal{A}}_i}{1-2\eta^{\mathcal{A}}_i}\right) \left(\frac{1+2\eta^{Q}_i}{1-2\eta^{Q}_i}\right)\right]
\end{equation}
\end{theorem}

\begin{proof}
 Each round of disclosure $t$ results in the use of an additional privacy budget $D_\infty(\mathcal{D}_t||\mathcal{D}'_t)$. Dwork et al.~\cite{dwork2010boosting} showed that, under the assumption that synopses are known, this quantity is $\frac{2}{\mu}  \log \left(\frac{1+2\eta_t}{1-2\eta_t}\right)$ privacy budget. Briefly, for each query $q \in Q$, let $d_q^t = |q(x) - A_t(q)|$ and $d_q'^t = |q(x') - A_t(q)|$. It follows that $|d_q^t - d_q'^t| \leq \Delta_Q$. This implies $|P_\mathcal{A}(t)[q] - P'_\mathcal{A}(t)[q]| \leq \frac{2\Delta_Q}{\mu}$, and thus:
\[ e^{-2\alpha \cdot t \cdot \frac{\Delta_Q}{\mu}} \leq \frac{u_t,q}{u'_t,q} \leq e^{2\alpha \cdot t \cdot \frac{\Delta_Q}{\mu}} \]

Define normalization factor for two adjacent database as $Z_t = \sum_{q \in Q} u_t,q$ and $Z'_t = \sum_{q \in Q} u'_t,q$, we have:
\[ e^{-2\alpha \cdot t \cdot \frac{\Delta_Q}{\mu}} \leq \frac{Z_t}{Z'_t} \leq e^{2\alpha \cdot t \cdot \frac{\Delta_Q}{\mu}} \]

The claim follows from the above because $\mathcal{D}_{t+1}[q] = \frac{u_t,q}{Z_t}$, $\mathcal{D}'_{t+1}[q] = \frac{u'_t,q}{Z'_t}$ and replacing $\alpha=\frac{1}{2}\log \left(\frac{1+2\eta}{1-2\eta}\right)$. 

Therefore, in DPI, since we update two PDFs $P_Q(t)$ and $P_\mathcal{A}(t,q)$, each of these two PDFs will consume a privacy budget of $\frac{4}{\mu} \sum_{i=1}^{t} \log \left(\frac{1+2\eta^{\mathcal{A}}_i}{1-2\eta^{\mathcal{A}}_i}\right)$ and $\frac{4}{\mu} \sum_{i=1}^t \log \left(\frac{1+2\eta^{Q}_i}{1-2\eta^{Q}_i}\right)$, respectively. This completes the proof. 
\end{proof}

The convergence of the privacy bound in the DPI can be analyzed through the total privacy budget $\epsilon$ utilized after $t$ iterations. Next, assuming the privacy budet is divided equally between the two reweighting mechanisms, i.e., $\epsilon_Q=\epsilon_\mathcal{A}=\frac{\epsilon}{2}$, the optimal series can be derived.

\subsection{Total Privacy and Utility Loss}
\label{sec:utility}
According to Dwork et al.~\cite{dwork2010boosting} (Claim 4.3), after $t$ rounds of DP with AdaBoosting, while the mechanism produces $\lambda$-inaccurate results, the loss probability $\mathcal{L}(t)$, derived via the cardinality of a subset of queries ($Q_{\text{bad}} \subset Q$), is upper-bounded by $\left(\sqrt{1 - 4\eta^2} \cdot t\right) \cdot |Q|$. A very similar observation to the proof but utilizing different $\eta_i$ in each iteration $i$ (from a converging series) will lead to $\sum_{i=1}^{t} \mathcal{L}_i= \sum_{i=1}^{t} \frac{1}{2} \log \left[(1+2\eta_i)(1-2\eta_i)\right]$. 

\begin{theorem}[Proof in Appendix~\ref{thm:3}]
\label{thm:utilitydpi}
Consider a series of DPI mechanisms, and each triggered with $\eta^{Q}_i$ and $\eta^{\mathcal{A}}_i$.  The overall error of DPI, represented as $\sum_{i=1}^{t} \mathcal{L}_i=\left[\mathbb{P}(|\mathcal{A}(q,t) - q(S_1, S_2, \cdots, S_t)| > \lambda + \mu)\right]$, is upper-bounded by:
\begin{equation}
    \frac{1}{4} \sum_{i=1}^{\infty} \log \left[(1 - 4 (\eta^{Q}_i)^2) \cdot \left(1 - 4 (\eta^{\mathcal{A}}_i)^2\right)\right]
\end{equation}
\end{theorem}

Theorems ~\ref{thm:privacydpi} and ~\ref{thm:utilitydpi} allow us to formulate an optimization problem for calculating the parameter $\eta$ to maximize utility under a total privacy constraint.
\begin{eqnarray}
\label{eq:tradeoff}
& \min_{\eta_i, i\in\left[\mathbb{N}\right]} \frac{1}{4} \sum_{i=1}^{N \rightarrow \infty} \log \left[(1-4 (\eta^{Q}_i)^2)\left(1-4 (\eta^{\mathcal{A}}_i)^2\right)\right] \nonumber\\
&\text{w.r.t. } \frac{4}{\mu} \sum\limits_{i=1}^{N \rightarrow \infty} \log \left[\left(\frac{1+2\eta^{\mathcal{A}}_i}{1-2\eta^{\mathcal{A}}_i}\right) \left(\frac{1+2\eta^{Q}_i}{1-2\eta^{Q}_i}\right)\right]=\epsilon  
\end{eqnarray}

To derive the optimal series, we model the problem using the continuous representations $\eta(x)$ and $\mathcal{L}(\eta(x))$. This simplifies the analysis using the integral calculus.
\begin{definition}
Continuous functions $\eta:\mathbb{R}^{+}\rightarrow (0,0.5)$ and $\mathcal{L}:\eta(x)\rightarrow \frac{1}{2} \log (1-2 \eta(x))\left(1+2 \eta(x)\right)$ are defined to be, the continuous representation of $\eta_i$, and its corresponding continuous loss function obtained by interpolating a continuous function over the sequences of $<\eta^{\mathcal{A}}_i, \eta^{Q}_i>$.
\end{definition}
Moreover, we make two reasonable assumptions: 
\begin{enumerate}
    
    \item The interpolation error to be negligible. 

\vspace{0.05in}
    
    \item The function $\eta(x)$ is twice continuously differentiable, indicating that it belongs to the class $C^2$. 
\end{enumerate}

\subhead{Total Privacy over Infinite Iterations} 
Under this continuous model, we derive key relationships between the total privacy budget $\epsilon$ and the overall utility loss $L$ in Remark \ref{rem:utility} and Theorem \ref{rem1:privacy}. 

\begin{remark}[Proof in Appendix~\ref{rem1:utility}]
\label{rem:utility}
For the overall probability loss $L=\int_{1}^{\infty} \mathcal{L}(x) \,dx$, there exist constants $\zeta$, $m$, $M$ such that 
\small
\begin{align}
  L \geq  \frac{1}{4\zeta}\left(Li_2(M^2) - Li_2(m^2)+ 4Li_2(m)- 4Li_2(M)\right) 
\end{align}
\normalsize
where $Li_2(\cdot)$ is the poly-logarithm Spence's function.  
\end{remark}

\begin{theorem} [Proof in Appendix~\ref{rem:privacy}]
\label{rem1:privacy}
The total privacy bound of DPI over infinite time slots is 
\small
\begin{align}
    \epsilon \geq \frac{2\Delta_Q}{\mu \zeta}\left(Li_2(M^2) - Li_2(m^2)\right)
\end{align}
\normalsize
\end{theorem}

Since the $Li_2(M^2)$ is the poly-logarithms of the Spence's function, the overall privacy loss is growing logarithmically. Thus, Theorem \ref{rem1:privacy} signifies that the privacy loss is converging and effectively bounded in DPI.

\subhead{Utility Loss} 
We now present our final result: a series that maximizes the accuracy of PDF estimation in DPI. 
\begin{theorem}[Proof in Appendix~\ref{thm:4}]
\label{thm:final}
Given total privacy budget $\epsilon$, DPI mechanism achieves an optimal utility (in a probability loss sense) with the series $\eta_t=\frac{1}{2}\cdot\cfrac{e^X-1} {e^X+1}$:

\small
\begin{align}
  |L| \geq  \frac{1}{4\zeta} \left| \frac{\epsilon}{C} -\frac{2\pi^2}{3}+4 Li_2\left(\sqrt{Li_{2}^{-1}\left(\frac{\pi^2}{6}-\frac{\epsilon}{C}\right)}\right)\right|.
\end{align}
\normalsize

 where  $Li_2^{-1}(\cdot)$ represents the inverse of poly-logarithm of the Spence's function, $X=\frac{1-(Li_2^{-1}(\frac{\pi^2}{6}-\frac{\epsilon}{C}))^t}{t^2|\zeta|}$ $|\zeta|=\sup_{\{t=2,3,\dots\}} \{e^{|\eta_{k-1}-\eta_{k}|}\}$ and $C=\frac{\Delta_Q}{\mu|\zeta|}$.

\end{theorem}

DPI leverages the inverse of the poly-logarithm of the Spence's function which can be represented as a converging infinite series to guarantee the loss bound. As $t$ grows, the utility loss will converge. The converging series provides a tight utility guarantee in Theorem \ref{thm:final}.

\subhead{Summary} 
Our theoretical analyses have shown that DPI can provide strong cumulative privacy and high utility for infinite data streams. Specifically, we prove DPI's overall privacy budget grows logarithmically with iterations (Theorem \ref{rem1:privacy}). This allows continuous operation under a fixed total privacy budget $\epsilon$. We also derive an upper bound on DPI's total error over time (Theorem \ref{thm:utilitydpi}), quantifying its utility. Furthermore, we obtain the optimal strategy for $\eta_t$ across iterations to maximize utility under $\epsilon$ (Theorem \ref{thm:final}). These results show that DPI provably converges to O($\epsilon$) error for answering pre-defined queries under privacy budget $\epsilon$.

\section{Experimental Evaluations}
\label{sec:exp}
\subsection{Experimental Setting}

\noindent\textbf{Experimental Datasets.} We conduct all the experiments on three real-world and some synthetic streaming datasets to evaluate the performance of our DPI. Table \ref{table:data} shows the characteristics of the datasets (the time period is partitioned into specific numbers of equal-length time slots).

\small
\begin{table}[!h]
\small
\caption{Characteristics of datasets (after pre-processing)}
\vspace{-0.1in}
\begin{center}
\setlength\tabcolsep{1.5pt}
\begin{tabular}{|c|c|c|c|c|}
\hline
		\textbf{Dataset} &  \textbf{User $\#$} &\textbf{Domain Size}& \textbf{Slot $\#$} & \textbf{Period}\\
		\hline
            COVID-19 & $48,925$& $110$ &$1,000$&  6 month\\ \hline
            Network Traffic & $5,074,413$ & $65,534$ & $2,700$ & 1 hour \\ 
            \hline
            USDA Production & $1,695,038$& $111,989$ & $2,700$& 1960-2022\\ \hline
			
\end{tabular}
\label{table:data}
\vspace{-0.05in}
\end{center}
\end{table}

\normalsize

    \emph{COVID-19 Dataset} \cite{GAURAV_2023_covid} provides a chronological sequence of COVID-19 confirmed cases, deaths, and recoveries, allowing researchers to analyze the trend and pattern of the pandemic over time. Moreover, the attributes are disaggregated by 110 countries and subregions aiding in regional and comparative studies, enhancing the understanding of the geographical variation and spread of the virus.

    \emph{Network Traffic Dataset} \cite{networkdata} is used to study a variety of real-world DDoS attacks. This dataset contains network traffic such as source and destination IP addresses, etc. 
    $65,534$ source IP addresses and their network traffic records were extracted for our experiments.

    \emph{USDA Production Supply and Distribution Dataset} \cite{usdaproduction2020} collects the agricultural production data for 300 commodities with over 100 attributes, including $111,989$ categories of $1,695,038$ production records from 1960 to 2022.

\begin{figure*}[!tbh]
	\centering
	\subfigure[MSE vs $\epsilon$]{
		\includegraphics[angle=0, width=0.245\linewidth]{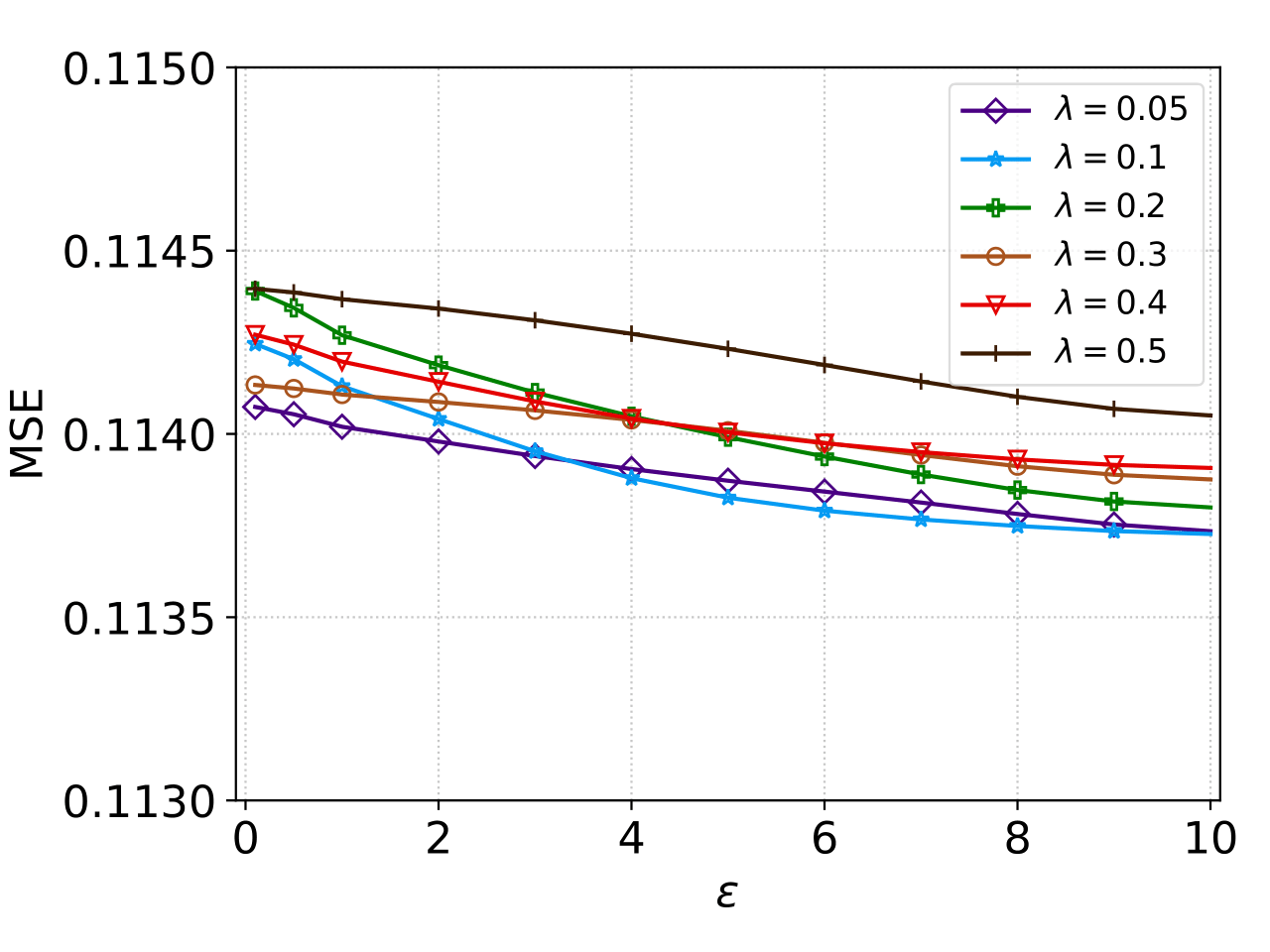}		
        \label{fig:hyper_lamda_mse1} }
		\hspace{-0.18in}
	\subfigure[MSE vs $\epsilon$]{
		\includegraphics[angle=0, width=0.245\linewidth]{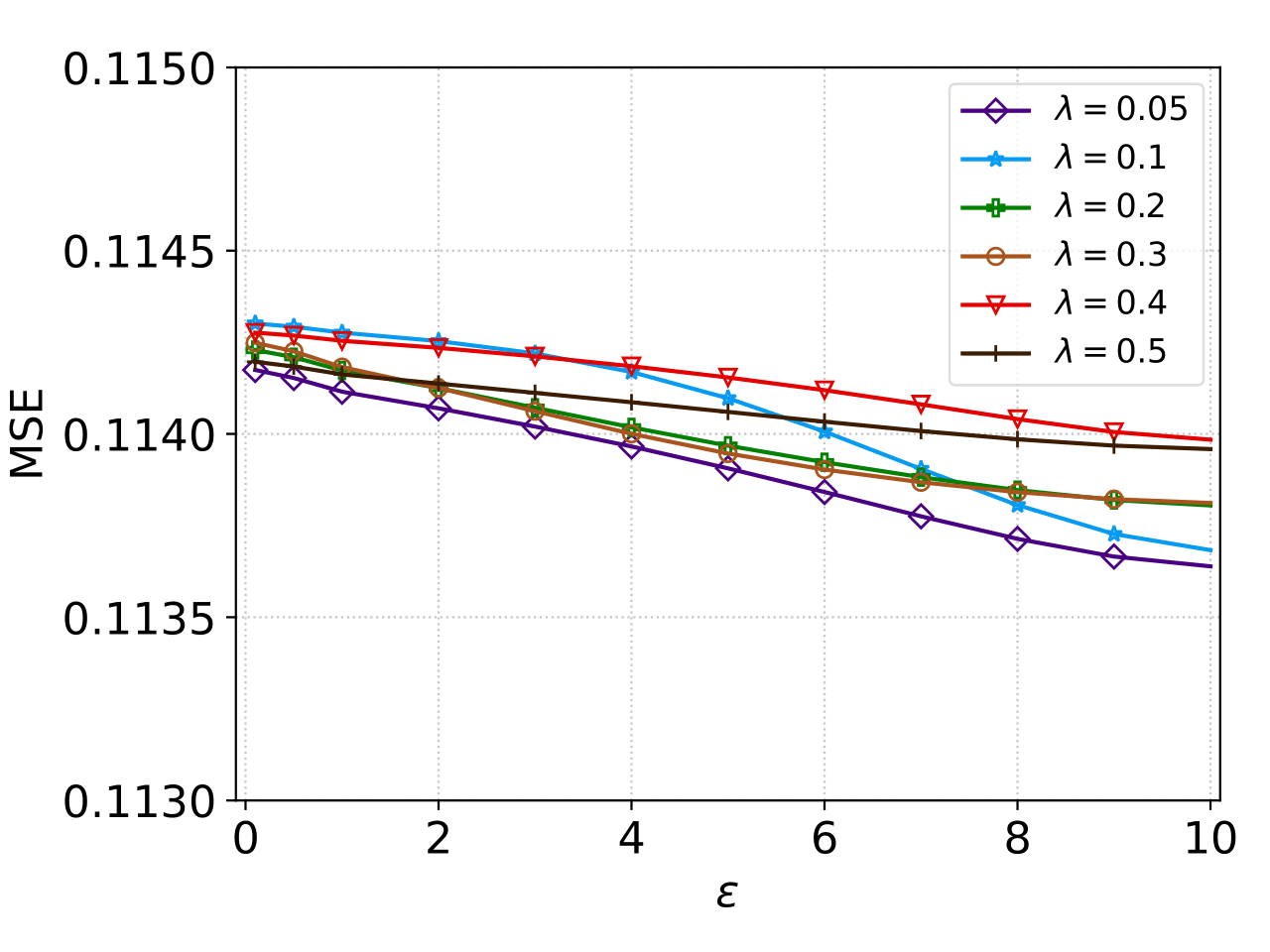}
		\label{fig:hyper_lamda_mse2}}
		\hspace{-0.18in}
	\subfigure[KL divergence vs $\epsilon$]{
		\includegraphics[angle=0, width=0.245\linewidth]{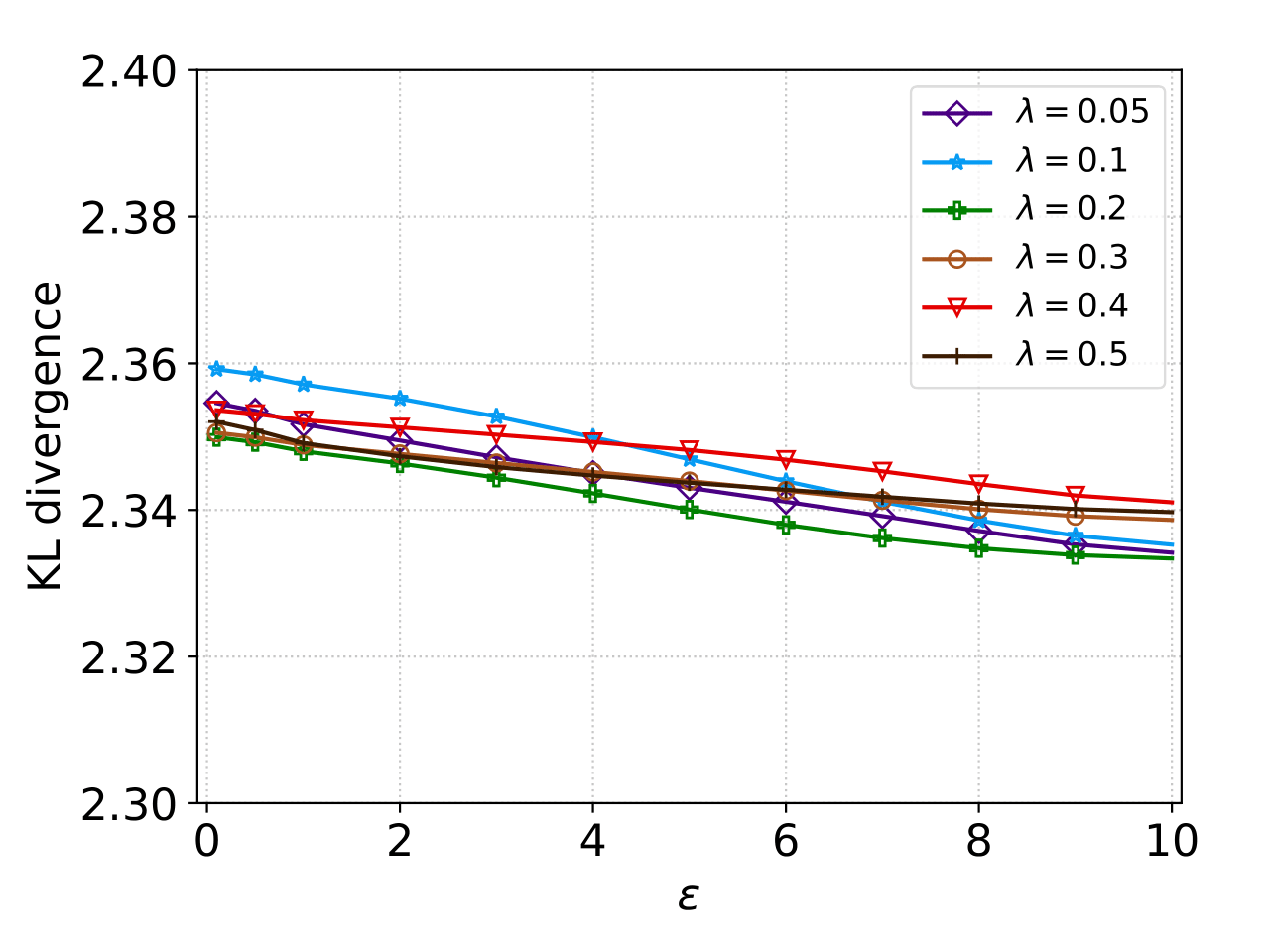}
		\label{fig:hyper_lamda_acc1} }
		\hspace{-0.18in}
	\subfigure[KL divergence vs $\epsilon$]{
		\includegraphics[angle=0, width=0.245\linewidth]{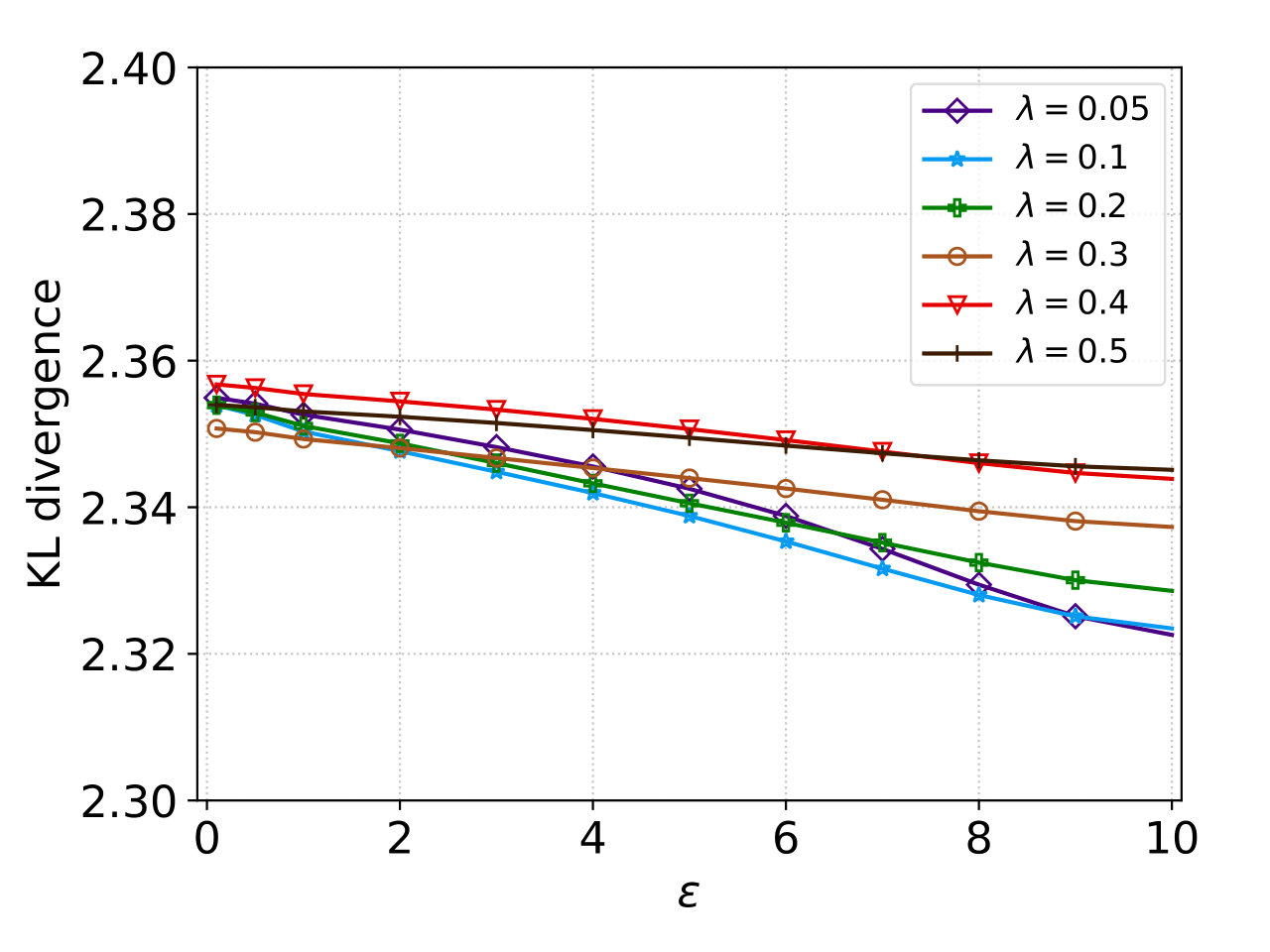}
		\label{fig:hyper_lamda_acc2} }
		\hspace{-0.18in}	
	\subfigure[MSE vs $\epsilon$]{
		\includegraphics[angle=0, width=0.245\linewidth]{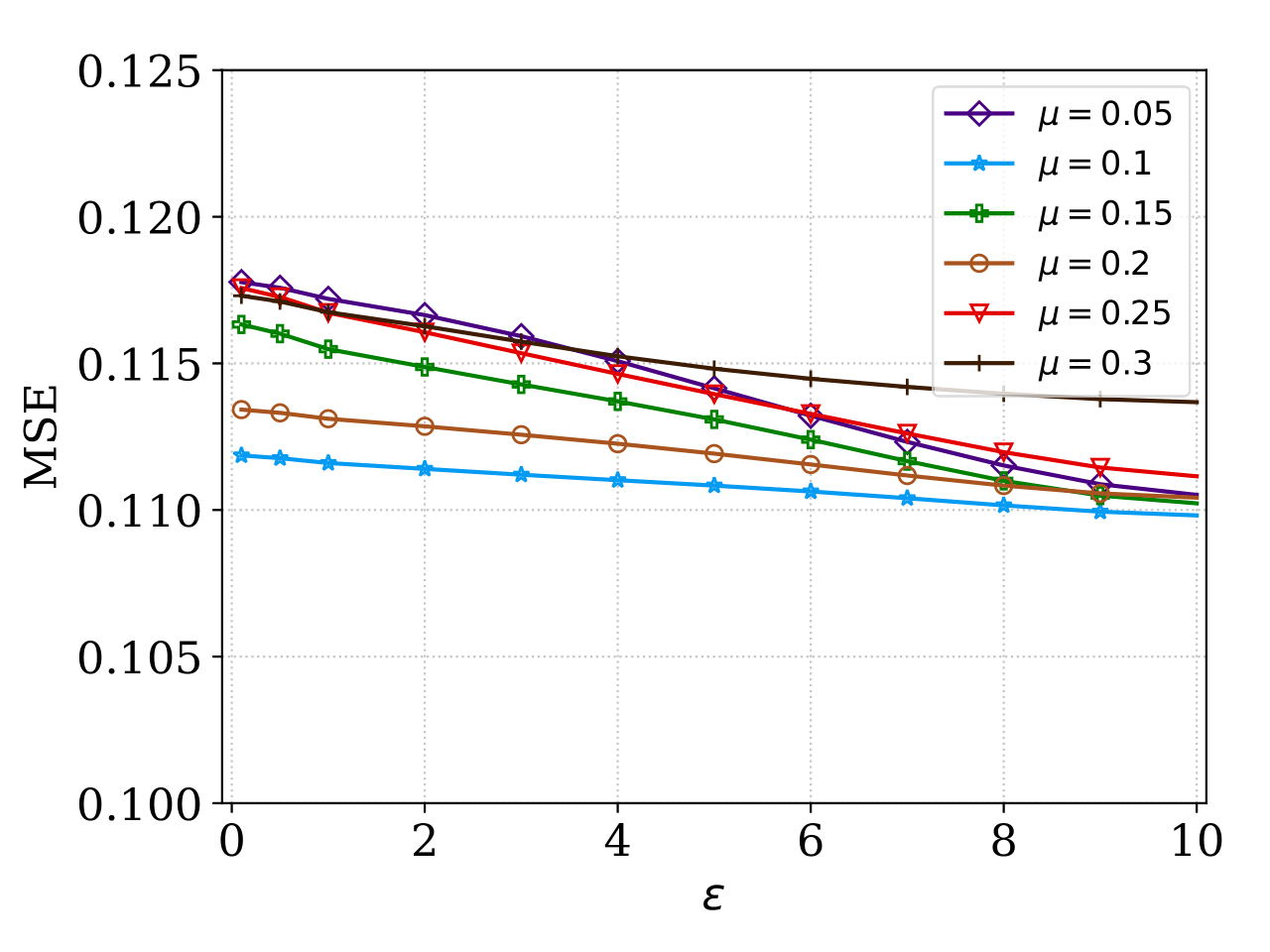}
		\label{fig:hyper_mu_mse1} }
		\hspace{-0.18in}	
	\subfigure[MSE vs $\epsilon$]{
		\includegraphics[angle=0, width=0.245\linewidth]{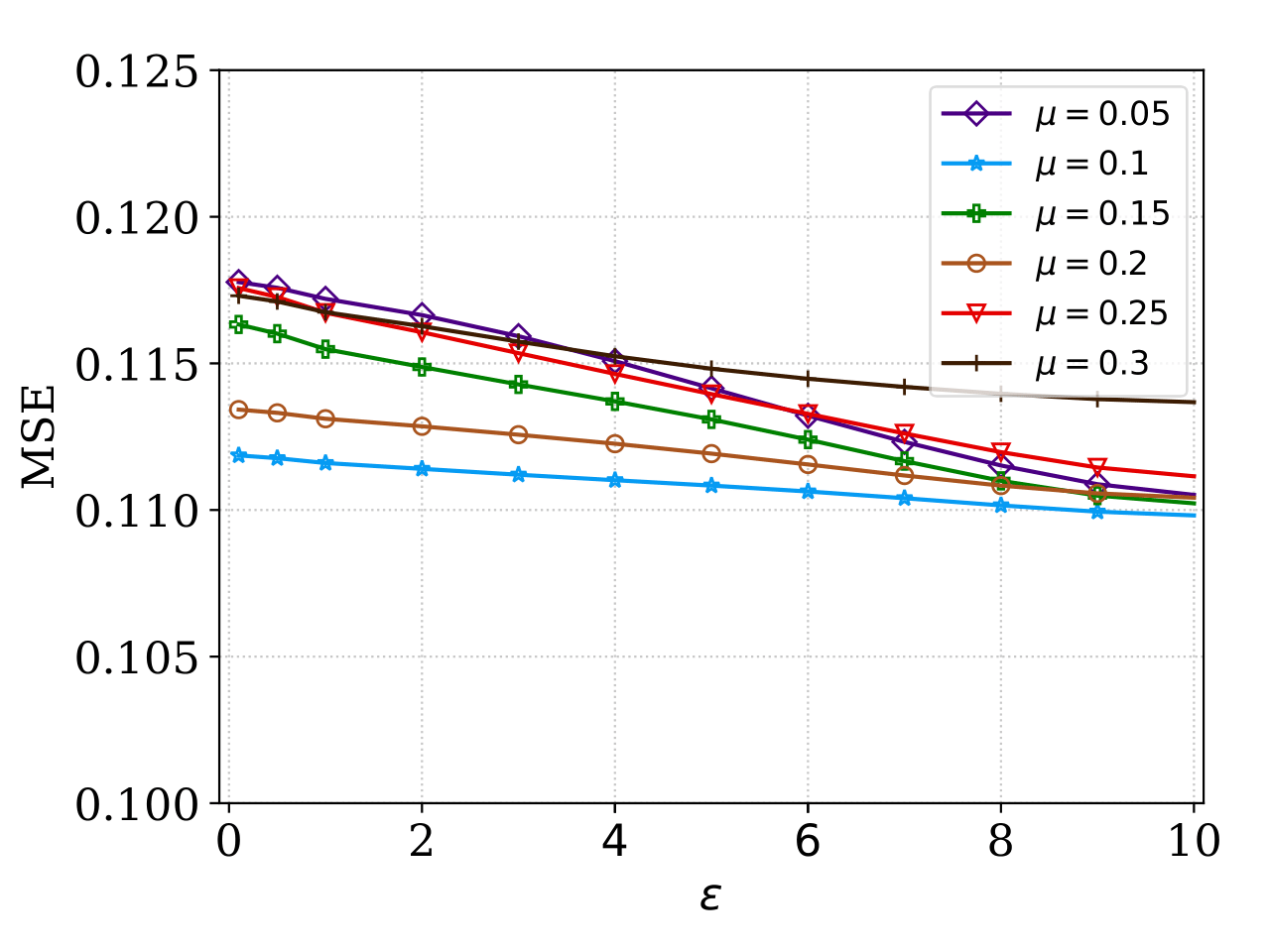}
		\label{fig:hyper_mu_mse2} }
        \hspace{-0.18in}
  	\subfigure[KL divergence vs $\epsilon$]{
		\includegraphics[angle=0, width=0.245\linewidth]{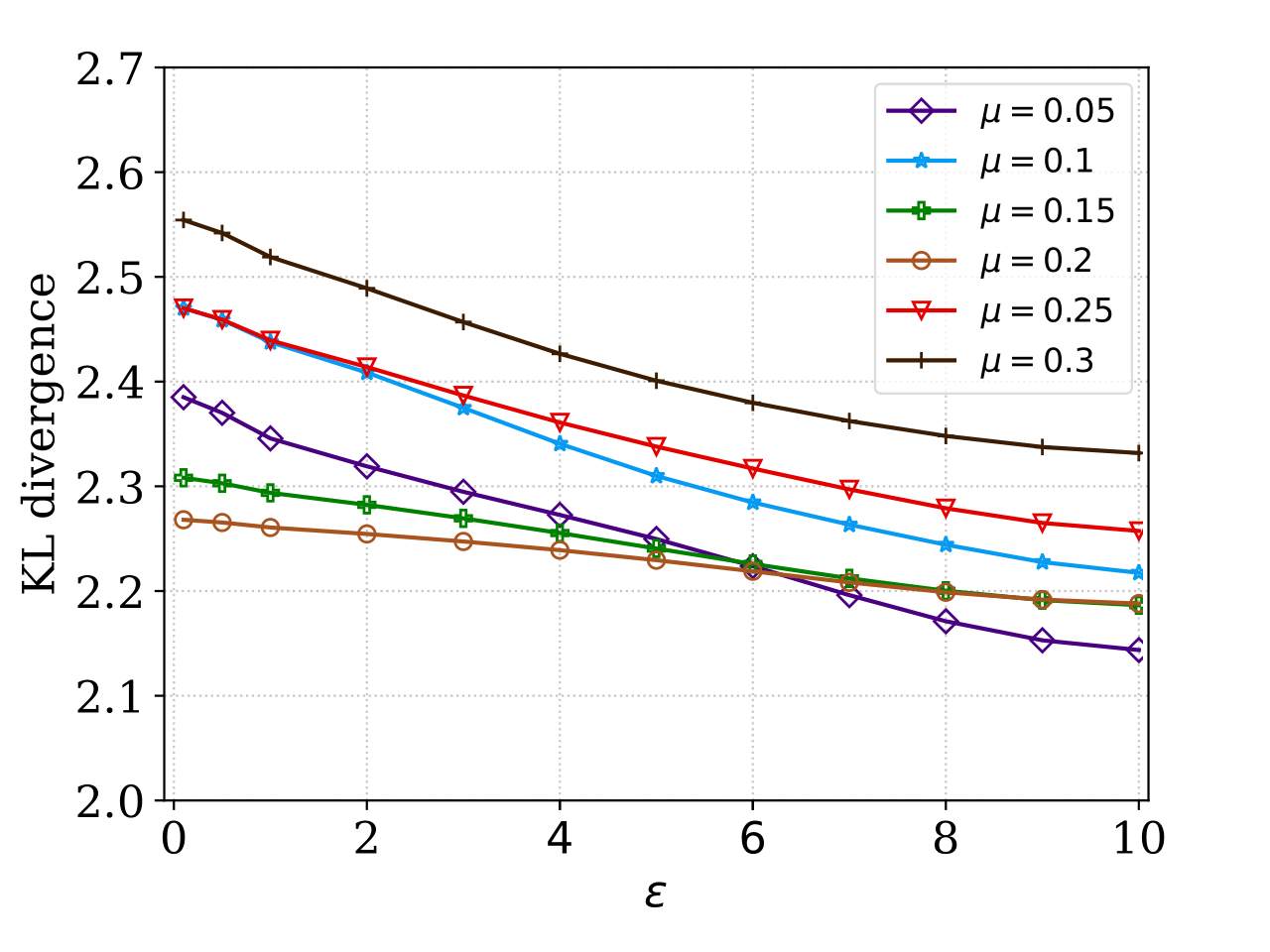}
		\label{fig:hyper_mu_acc1} }
		\hspace{-0.18in}
	\subfigure[KL divergence vs $\epsilon$]{
		\includegraphics[angle=0, width=0.245\linewidth]{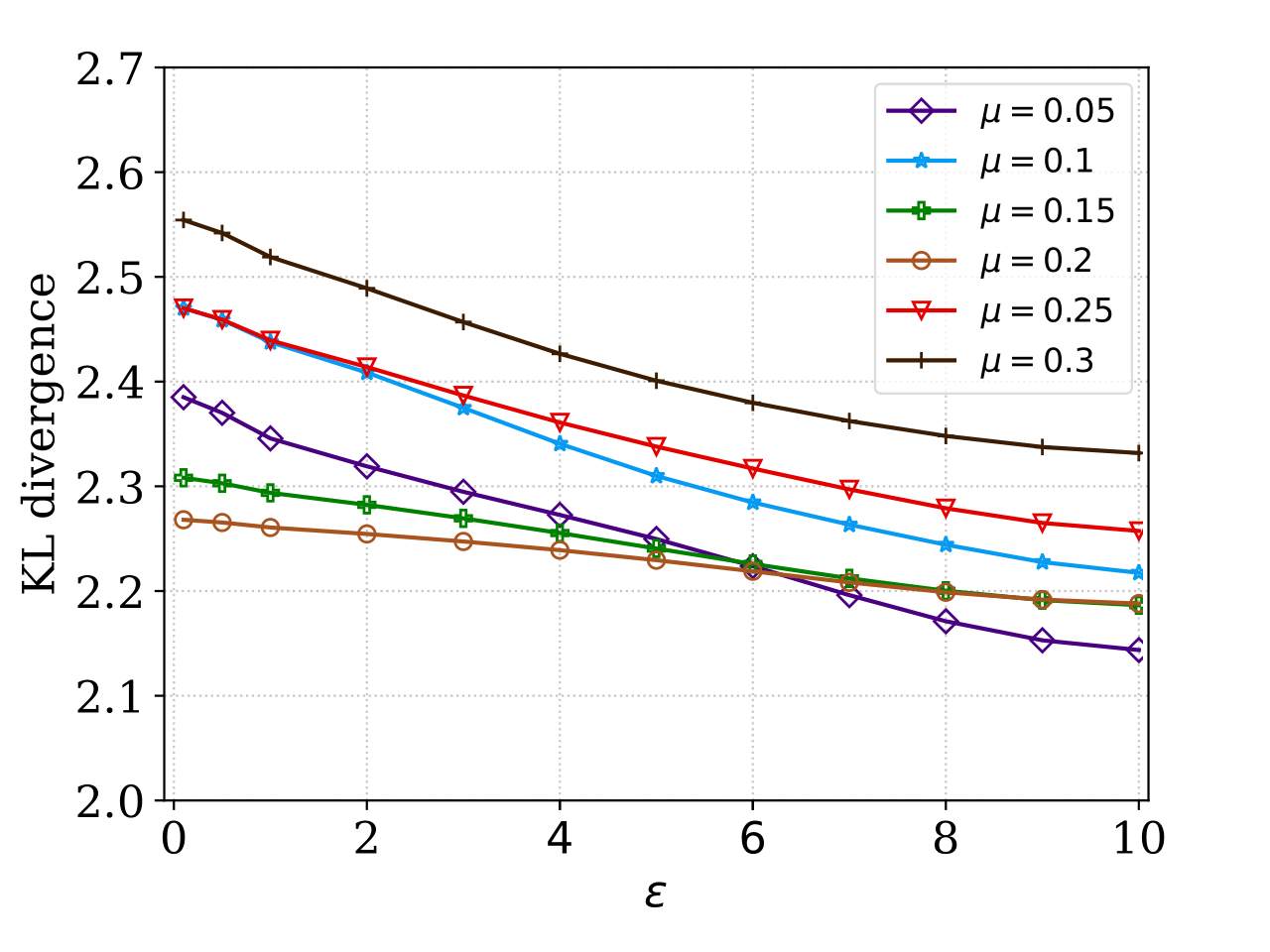}
		\label{fig:hyper_mu_acc2} }
        \hspace{-0.18in}
  	\subfigure[MSE vs $\epsilon$]{
		\includegraphics[angle=0, width=0.245\linewidth]{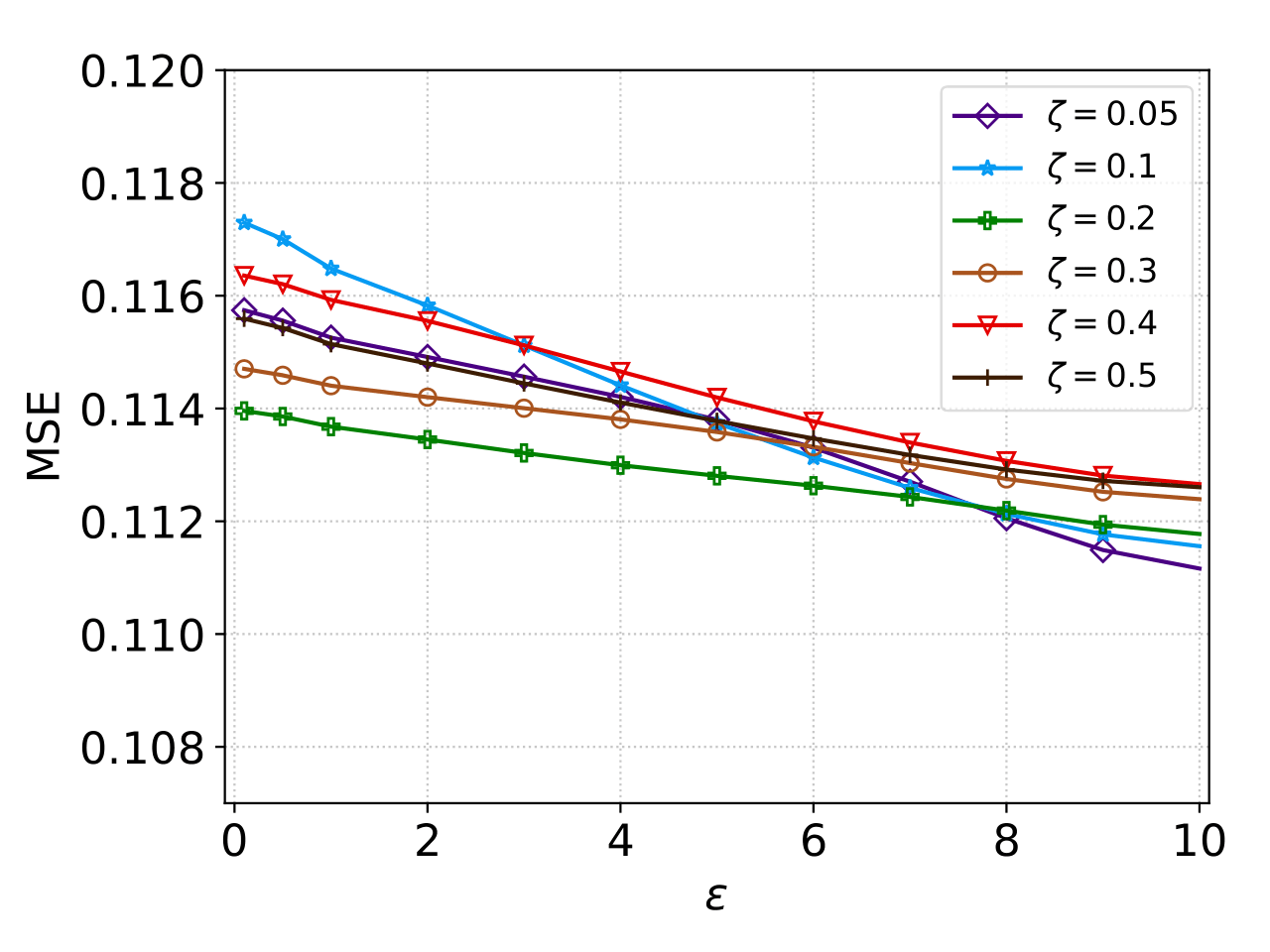}
		\label{fig:hyper_zeta_mse1} }
		\hspace{-0.18in}	
	\subfigure[MSE vs $\epsilon$]{
		\includegraphics[angle=0, width=0.245\linewidth]{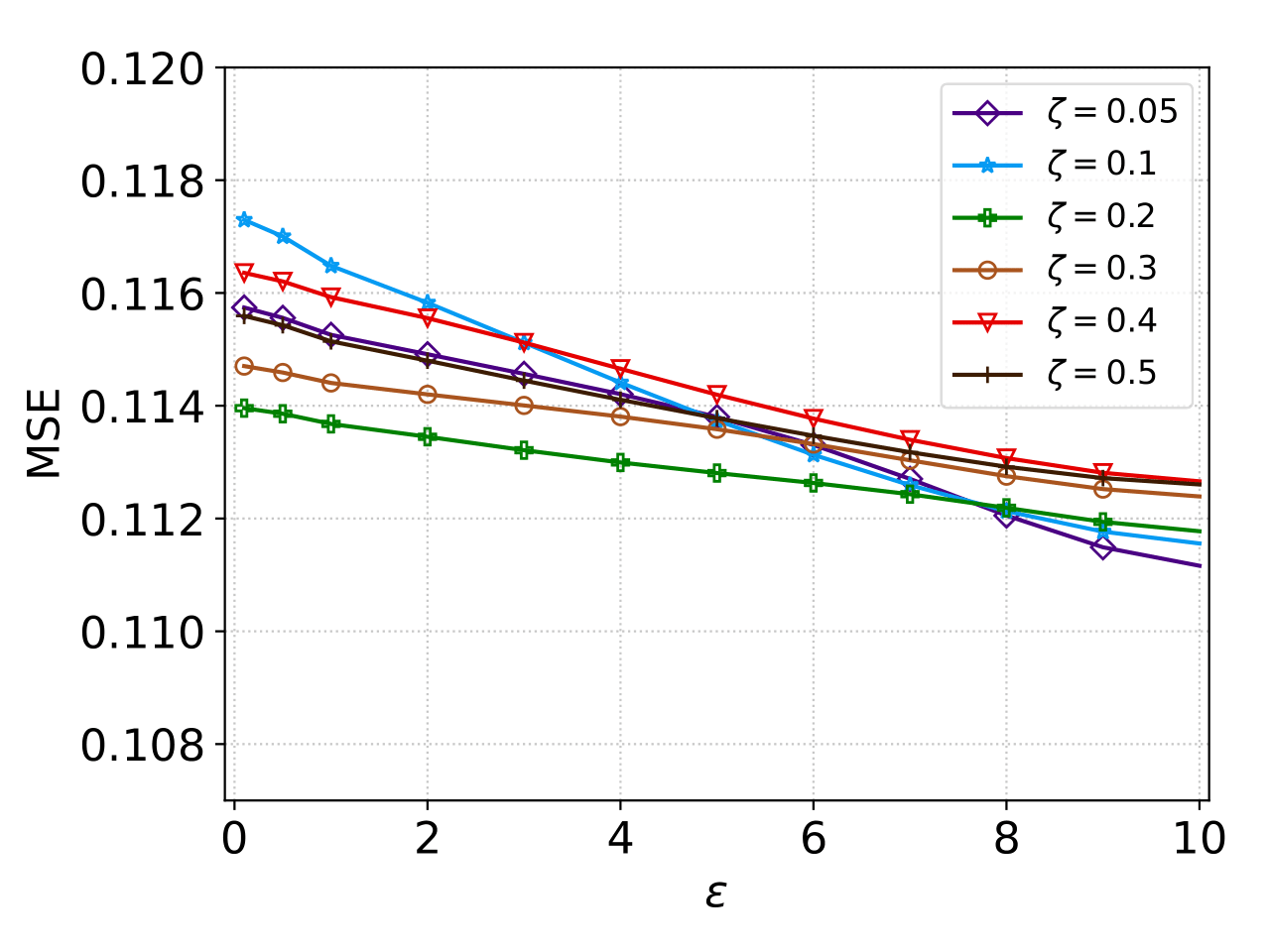}
		\label{fig:hyper_zeta_mse2} }
  \hspace{-0.18in}
  	\subfigure[KL divergence vs $\epsilon$]{
		\includegraphics[angle=0, width=0.245\linewidth]{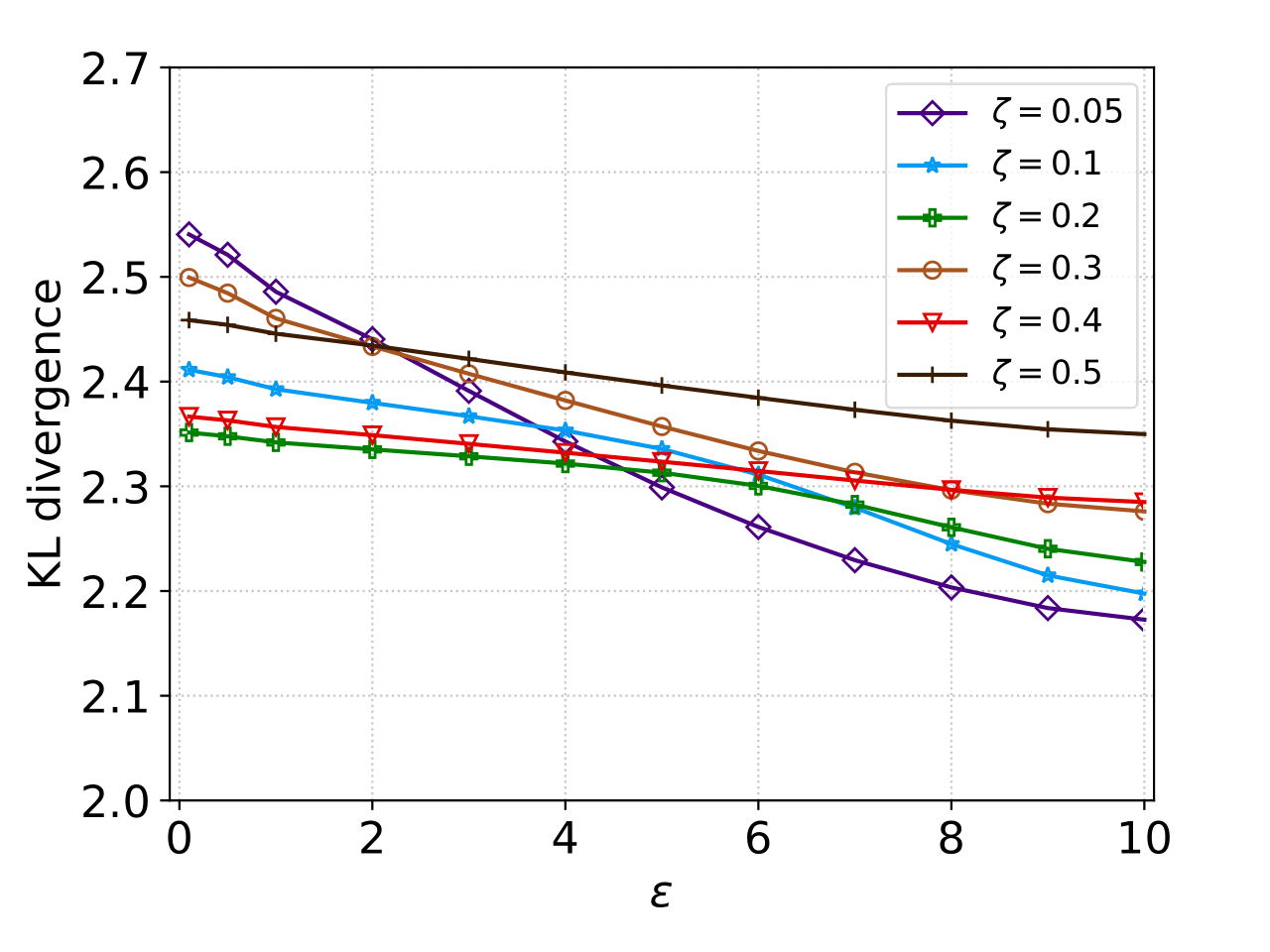}
		\label{fig:hyper_zeta_acc1} }
		\hspace{-0.18in}
	\subfigure[KL divergence vs $\epsilon$]{
		\includegraphics[angle=0, width=0.245\linewidth]{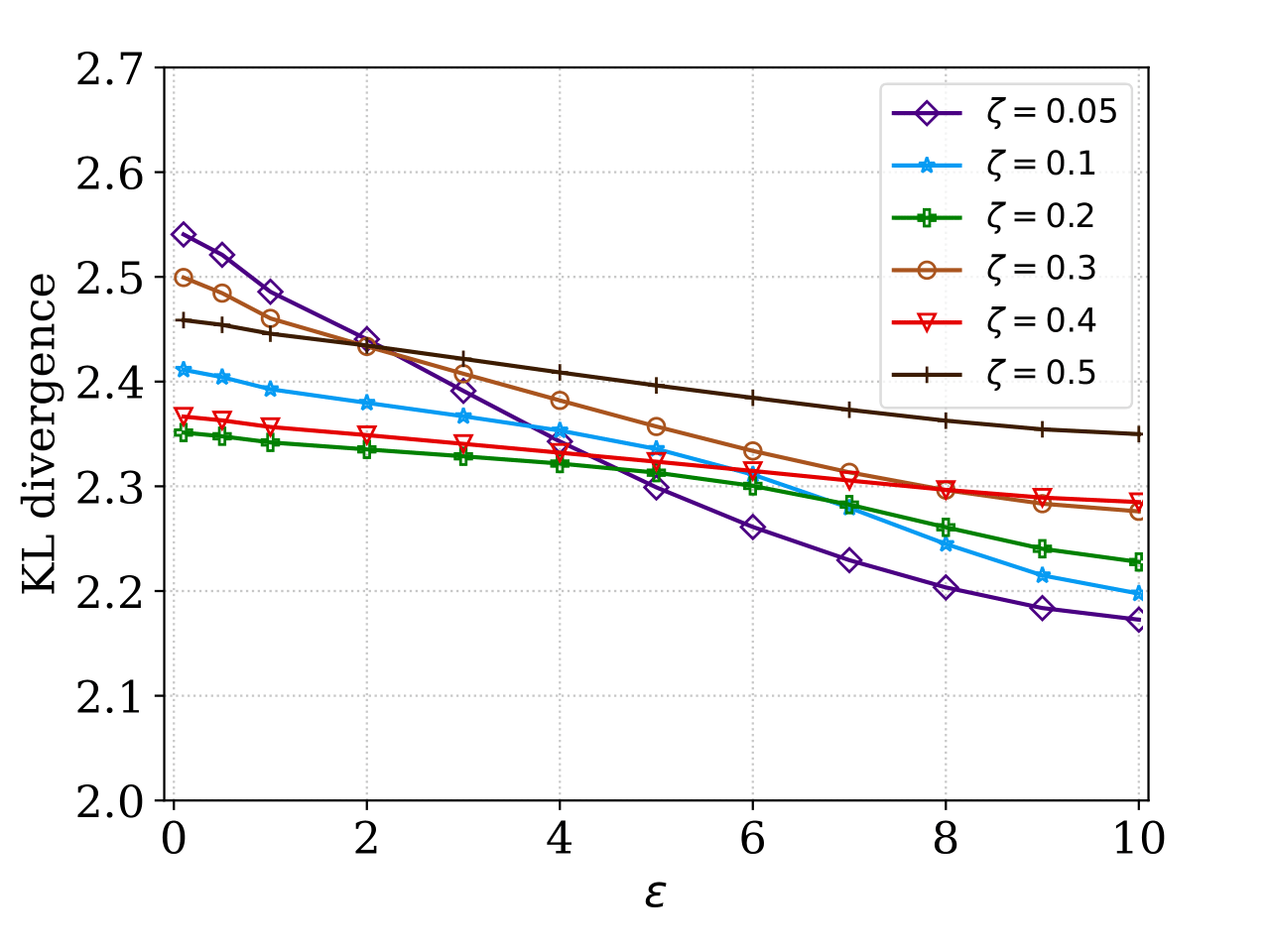}
		\label{fig:hyper_zeta_acc2} }
	\caption{Average MSE and KL divergence regarding $\epsilon$ for $2,700$ time slots. Figures (a), (c), (e), (g), (i), and (k) are MSE and KL divergence with different $\lambda$, $\mu$, and $\zeta$ values for accumulative data distribution disclosure. Figures (b), (d), (f), (h), (j), and (l) are the MSE and KL divergence with different $\lambda$, $\mu$, and $\zeta$ values for instantaneous data distribution disclosure.}
 \vspace{-0.1in}
	\label{fig:hyperparametertunning}
\end{figure*}

\subhead{Applications for Evaluation} We next introduce three real-world applications for evaluating the performance of DPI.

\emph{Statistical Queries.} Since our framework DPI could release the data distribution in each time slot, we first evaluate the utility for querying the item distribution: the difference between the original data's item distribution and the DPI output using the MSE and KL divergence metrics.

\emph{Anomaly Detection.} In addition to statistical queries, the data distribution released by DPI facilitates various downstream analyses \cite{coluccia2013distribution, an2015variational,zong2018deep,defard2021padim,AsifPV19, Giraldo2020AdversarialCU,mohammady2022dpoad,du2019robust}. In our experiments, we conduct experiments on anomaly detection (\emph{most existing methods cannot support such analysis}) by adopting a distribution-based anomaly detection method  \cite{goldstein2012histogram}  
to detect the anomalies with the DPI output and evaluate the accuracy. 

\emph{Recommender System.} DPI also enhances recommender systems by providing privacy guarantees without compromising accuracy. Singular Value Decomposition (SVD) algorithms \cite{sarwar2000application,zhou2015svd} used in recommender systems can investigate the relationships between items by using the distribution of users and items. In our experiments, DPI enables precise suggestions by providing differentially private data distributions to the SVD algorithm. 

\subhead{Benchmarks} We have shown that existing methods violate the privacy bound $\epsilon$ after a few time slots (see Appendix \ref{sec:sotaloss}). They are incomparable with DPI in our infinite settings due to their limitations (e.g., unbounded privacy, sensitivity assumption, event-level differential privacy).

\subsection{Performance on Hyperparameters}
\label{sec:tuning}

The utility of DPI is influenced by parameters such as $\epsilon$, $\lambda$ (error bound for synopsis), $\mu$ (error bound for avoiding overfitting), and $\zeta$ (decay rate of the decaying series), which are all related to the reweighting phase. To explore how they affect the utility in real-time disclosure, we choose the COVID-19 dataset to observe the utility trends with varying hyperparameters. The privacy requirement $\epsilon$ ranges from $0.1$ to $10$, and $\zeta$ varies from $0.1$ to $0.5$. $\lambda$ is set in the range $0.05$ to $0.5$ and $\mu$ varies from $0.05$ to $0.3$. 

We show the MSE and KL divergence regarding $\epsilon$ for $1,000$ time slots in Figure \ref{fig:hyperparametertunning}. Figures \ref{fig:hyper_lamda_mse1}, \ref{fig:hyper_lamda_acc1},
\ref{fig:hyper_mu_mse1}, 
\ref{fig:hyper_mu_acc1}, 
\ref{fig:hyper_zeta_mse1}, 
and \ref{fig:hyper_zeta_acc1} are the MSE and KL divergence with different $\lambda$, $\mu$, and $\zeta$ values for accumulative data distribution disclosure (from time slot 1 to $1,000$). Figures \ref{fig:hyper_lamda_mse2}, \ref{fig:hyper_lamda_acc2},
\ref{fig:hyper_mu_mse2}, 
\ref{fig:hyper_mu_acc2}, 
\ref{fig:hyper_zeta_mse2}, 
and \ref{fig:hyper_zeta_acc2} are the MSE and KL divergence with different $\lambda$, $\mu$, and $\zeta$ values for instantaneous data distribution disclosure (time slot $1,000$).

First, Figure \ref{fig:hyperparametertunning} demonstrates the MSE and KL divergence by varying $\epsilon$ for accumulative data distribution disclosure and instantaneous data distribution disclosure. With an increasing $\epsilon$, the MSE and KL divergence both decrease, demonstrating that the DPI output is closer to the true distribution (privacy/utility tradeoff). Second, Figures \ref{fig:hyper_lamda_acc1} and \ref{fig:hyper_lamda_acc2} show the KL divergence decreases faster when $\lambda$ is very small (purple line) since the parameter $\lambda$ evaluates the error of query based on the sampled synopsis and a small $\lambda$ means that the synopsis is more accurate. Thus, the output results of DPI tend to be accurate. Furthermore, the smaller additional error bound $\mu$ (for avoiding overfitting) indicating the confidence for the universal synopsis output also affects the performance of DPI. Figures \ref{fig:hyper_mu_acc1} and \ref{fig:hyper_mu_acc2} demonstrate that we can have good utility with small KL divergence given small $\mu$ (e.g., 0.05 or 0.1). In our settings, $\mu$ is set to be less than half of $\lambda$. Finally, Figures \ref{fig:hyper_zeta_acc1} and \ref{fig:hyper_zeta_acc2} validate that small $\zeta$ (e.g., 0.05, 0.1, or 0.2) returns higher accuracy (since smaller $\zeta$ produces larger $\eta$, leading to more accurate synopses).

Similarly, we can draw similar observations from Figures \ref{fig:hyper_lamda_mse1} and \ref{fig:hyper_lamda_mse2}, \ref{fig:hyper_mu_mse1} and \ref{fig:hyper_mu_mse2}, \ref{fig:hyper_zeta_mse1} and \ref{fig:hyper_zeta_mse2} (with the MSE metric). This confirms that the effect of different parameters on DPI aligns with our theoretical studies.

\subsection{Utility Evaluation}
\label{sec:utility_eva}

\noindent\textbf{Statistical Queries.} We consider two queries: sum and mean. We evaluate the KL divergence of total $1,000$ time slots with a varying $\epsilon$ and present the results in Figures \ref{fig:statistical_mse_eps1} and \ref{fig:statistical_mse_eps2}. The results show that the utility of DPI improves as $\epsilon$ increases. Even with a small $\epsilon$, DPI ensures low MSE values, e.g., $\sim0.12$ when $\epsilon = 0.1$. In addition, we also evaluate the MSE and KL divergence in each of the $1,000$ time slots, as shown in Figures \ref{fig:statistical_kl_eps} and \ref{fig:statistical_kl_eps2}. The MSE values for all statistical queries are quite low, up to $0.16$.

\subhead{Anomaly Detection} 
Another application focuses on anomaly detection in the Network Traffic dataset. In this task, we will apply the Histogram-based Outlier Score (HBOS) \cite{goldstein2012histogram} (equal to a distribution-based anomaly detection technique) and isolation forest \cite{liu2012} to identify anomalies. The HBOS value reflects the rarity of data in the learned distribution and the extent of deviation of certain features from this distribution. The network monitor can quickly identify and report any anomalies from incoming packets, along with their severity, to the respective tenants.

\begin{figure*}[!tbh]
	\centering
	\subfigure[MSE vs $\epsilon$]{
		\includegraphics[angle=0, width=0.245\linewidth]{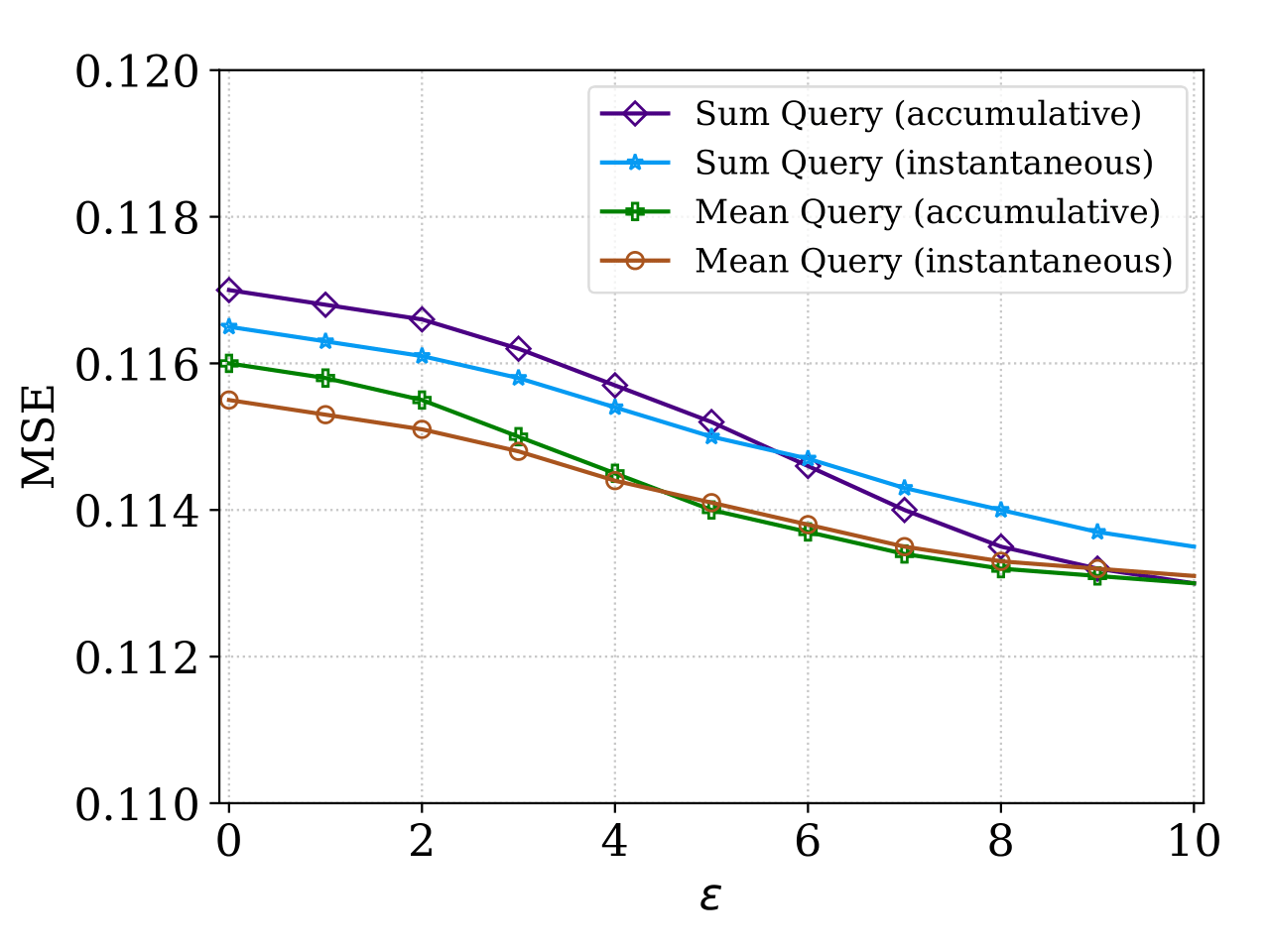}
        \label{fig:statistical_mse_eps1} }
		\hspace{-0.18in}
	\subfigure[KL divergence vs $\epsilon$]{
		\includegraphics[angle=0, width=0.245\linewidth]{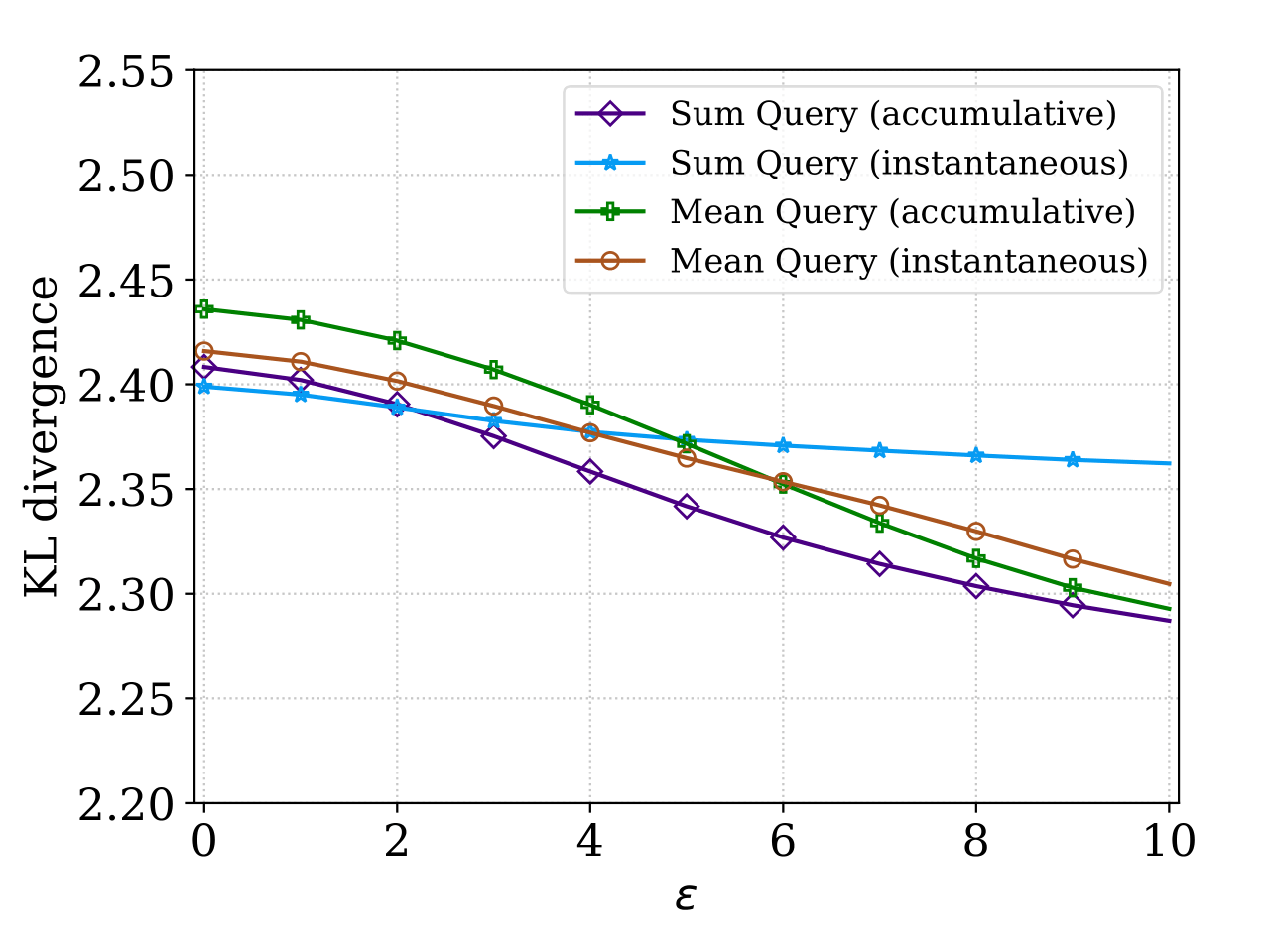}
		\label{fig:statistical_mse_eps2}}
		\hspace{-0.18in}
	\subfigure[MSE (moving average) vs $t$]{
		\includegraphics[angle=0, width=0.245\linewidth]{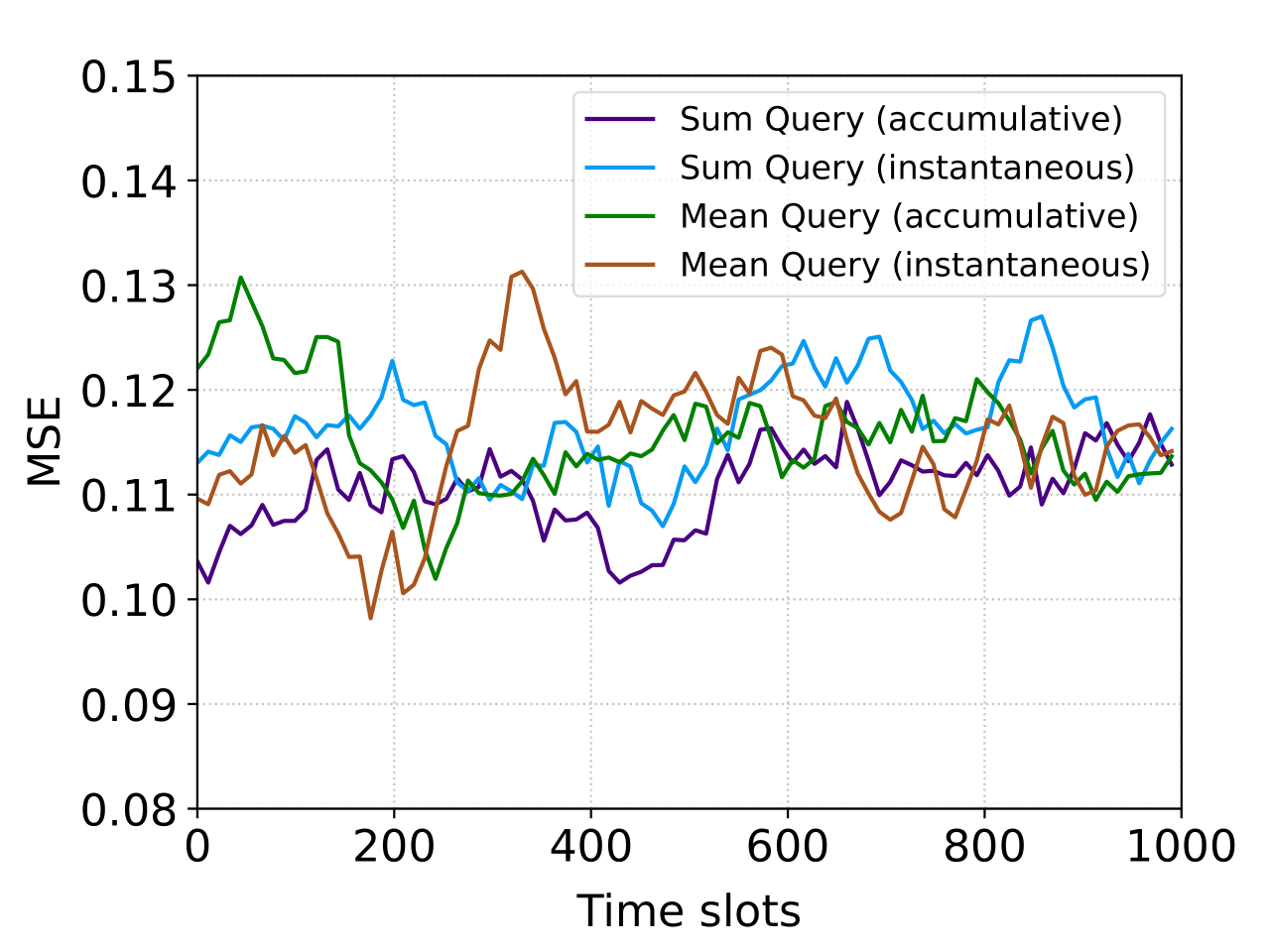}
		\label{fig:statistical_kl_eps} }
		\hspace{-0.18in}
	\subfigure[KL div. (moving average) vs $t$]{
		\includegraphics[angle=0, width=0.245\linewidth]{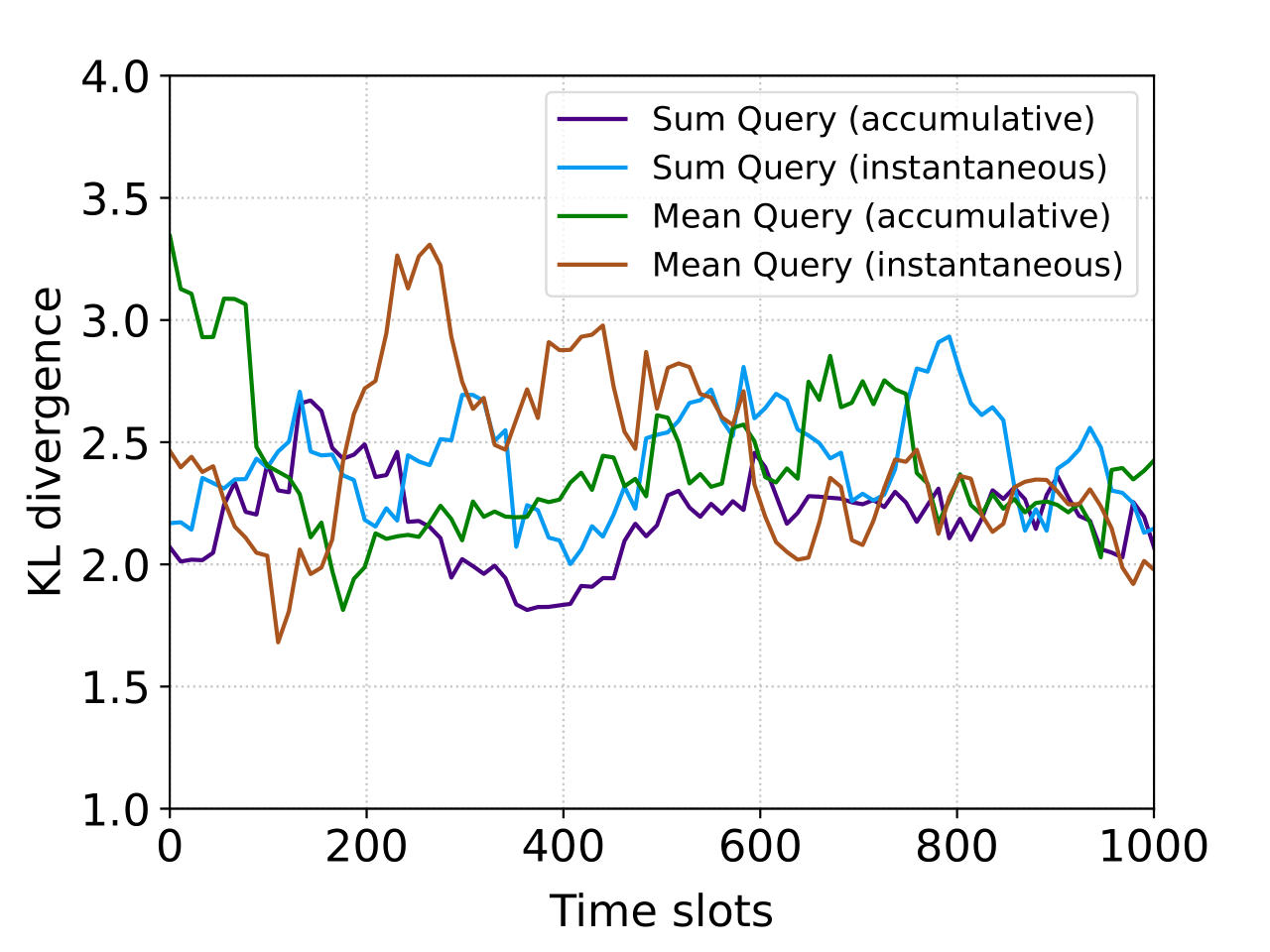}
		\label{fig:statistical_kl_eps2} }
	\caption{Average MSE and KL divergence of statistical queries for $\epsilon$ and $t$ across $1,000$ time slots on COVID-19 dataset. (a), (b): Average MSE and KL divergence with varying $\epsilon$. (c), (d): Average MSE and KL divergence over $t$ when $\epsilon$ is $2$.}
 \vspace{-0.1in}
	\label{fig:statisticalquries}
\end{figure*}

\begin{figure*}[!tbh]
	\centering
	\subfigure[Precision vs $\epsilon$]{
		\includegraphics[angle=0, width=0.245\linewidth]{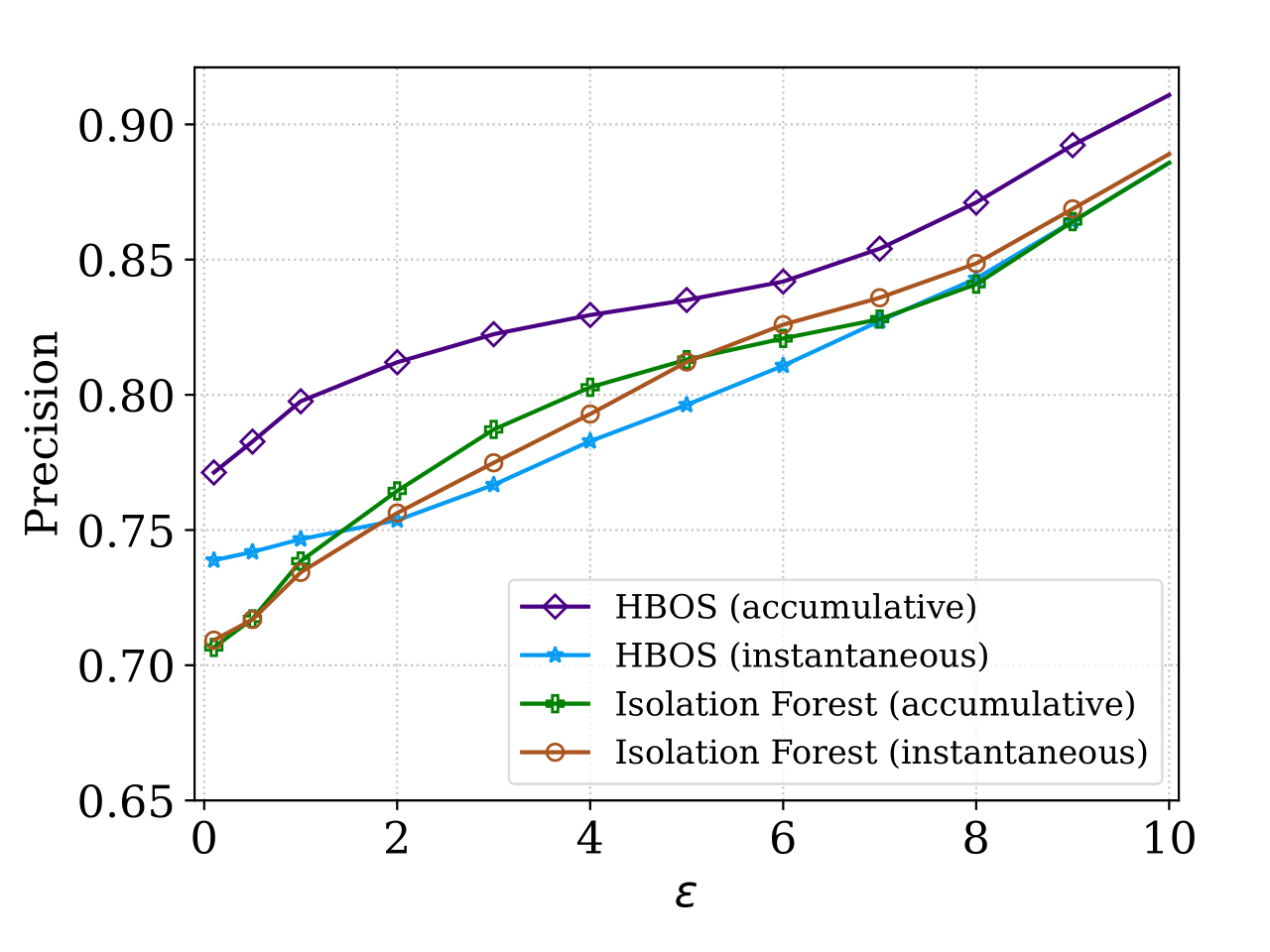}		
        \label{fig:anomaly_precision} }
		\hspace{-0.18in}
	\subfigure[Recall vs $\epsilon$]{
		\includegraphics[angle=0, width=0.245\linewidth]{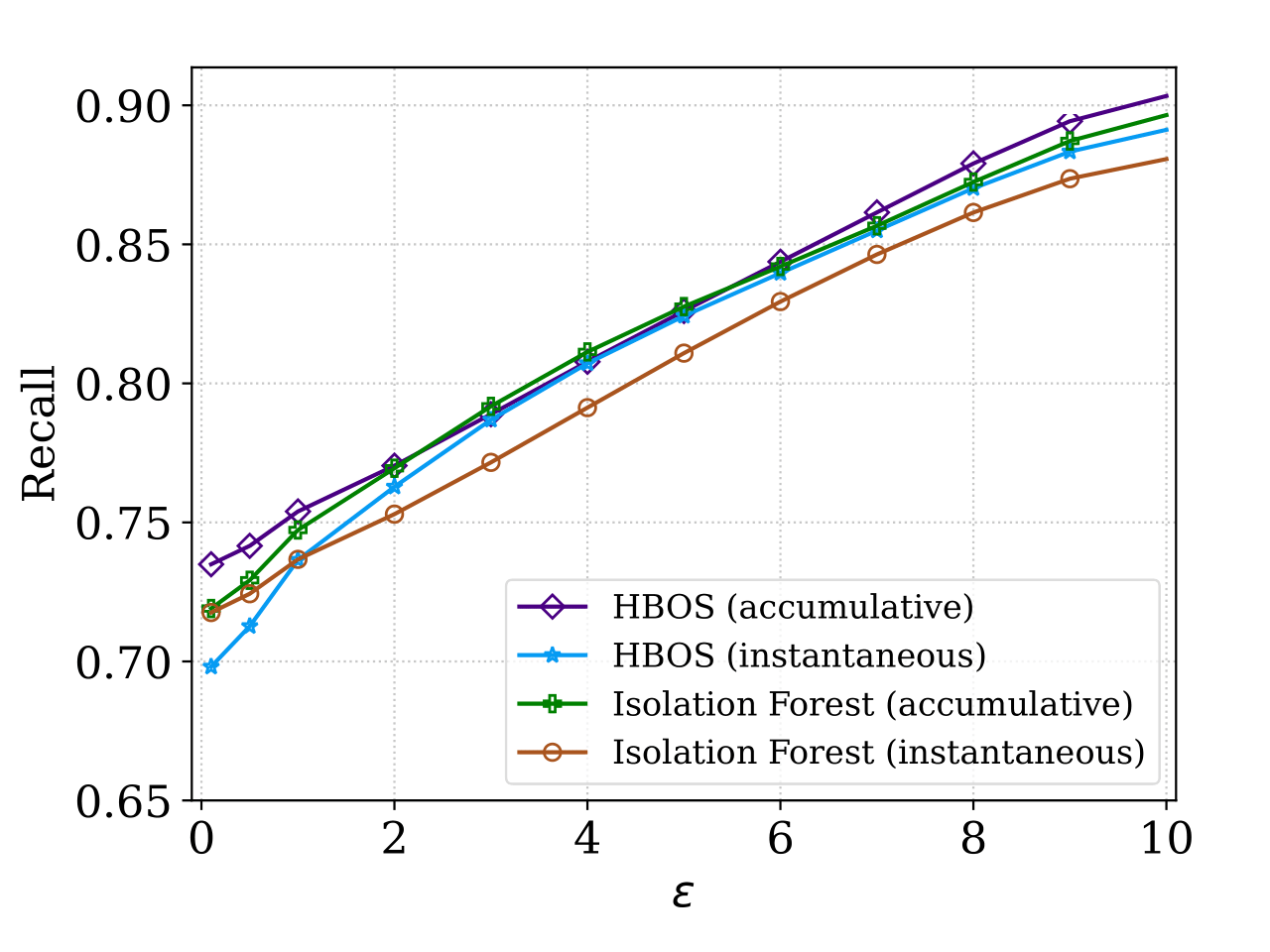}
		\label{fig:anomaly_recall}}
		\hspace{-0.18in}
	\subfigure[Precision (moving average) vs $t$]{
		\includegraphics[angle=0, width=0.245\linewidth]{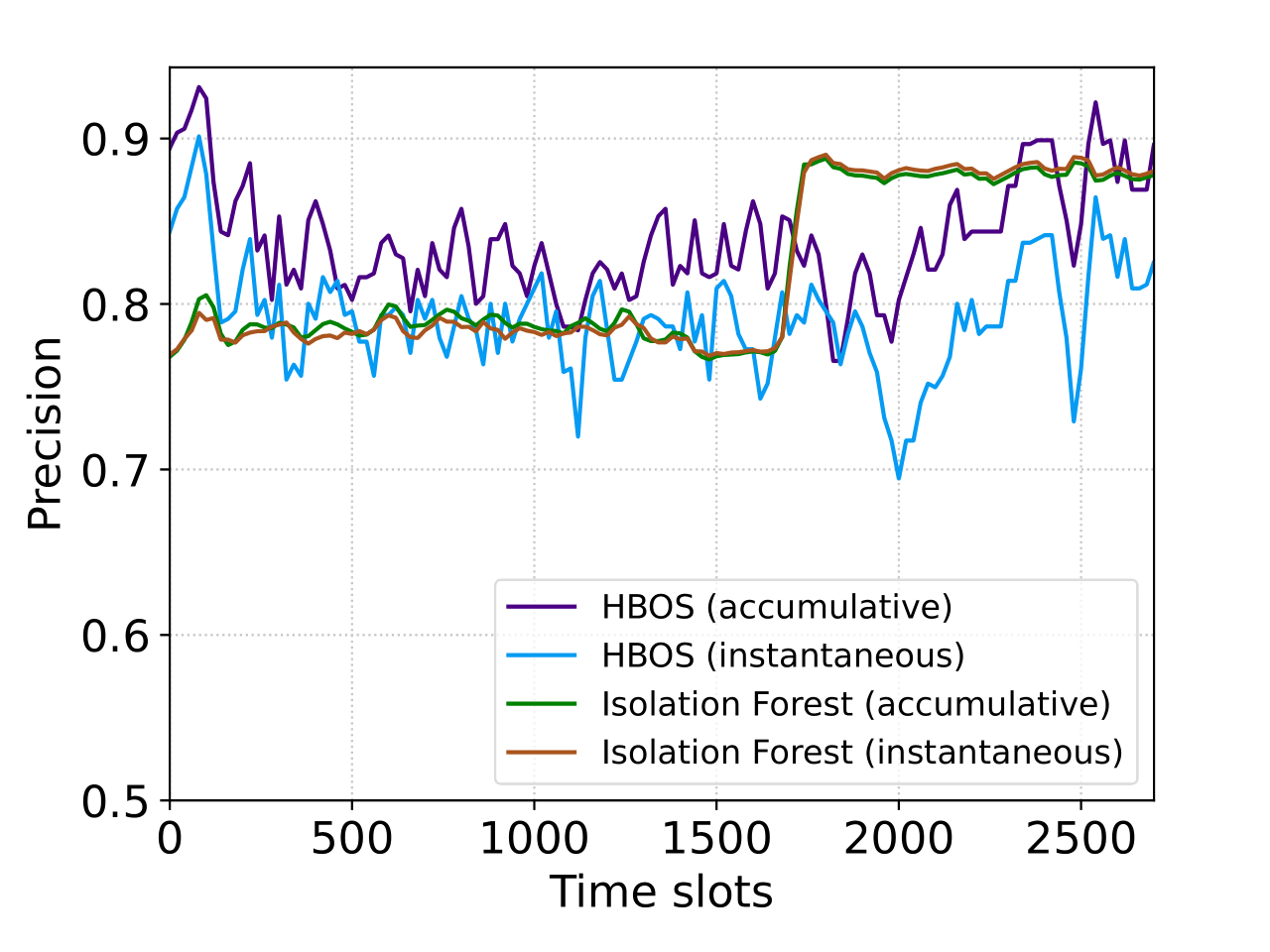}
		\label{fig:anomaly_precision_realtime} }
		\hspace{-0.18in}
	\subfigure[Recall (moving average) vs $t$]{
		\includegraphics[angle=0, width=0.245\linewidth]{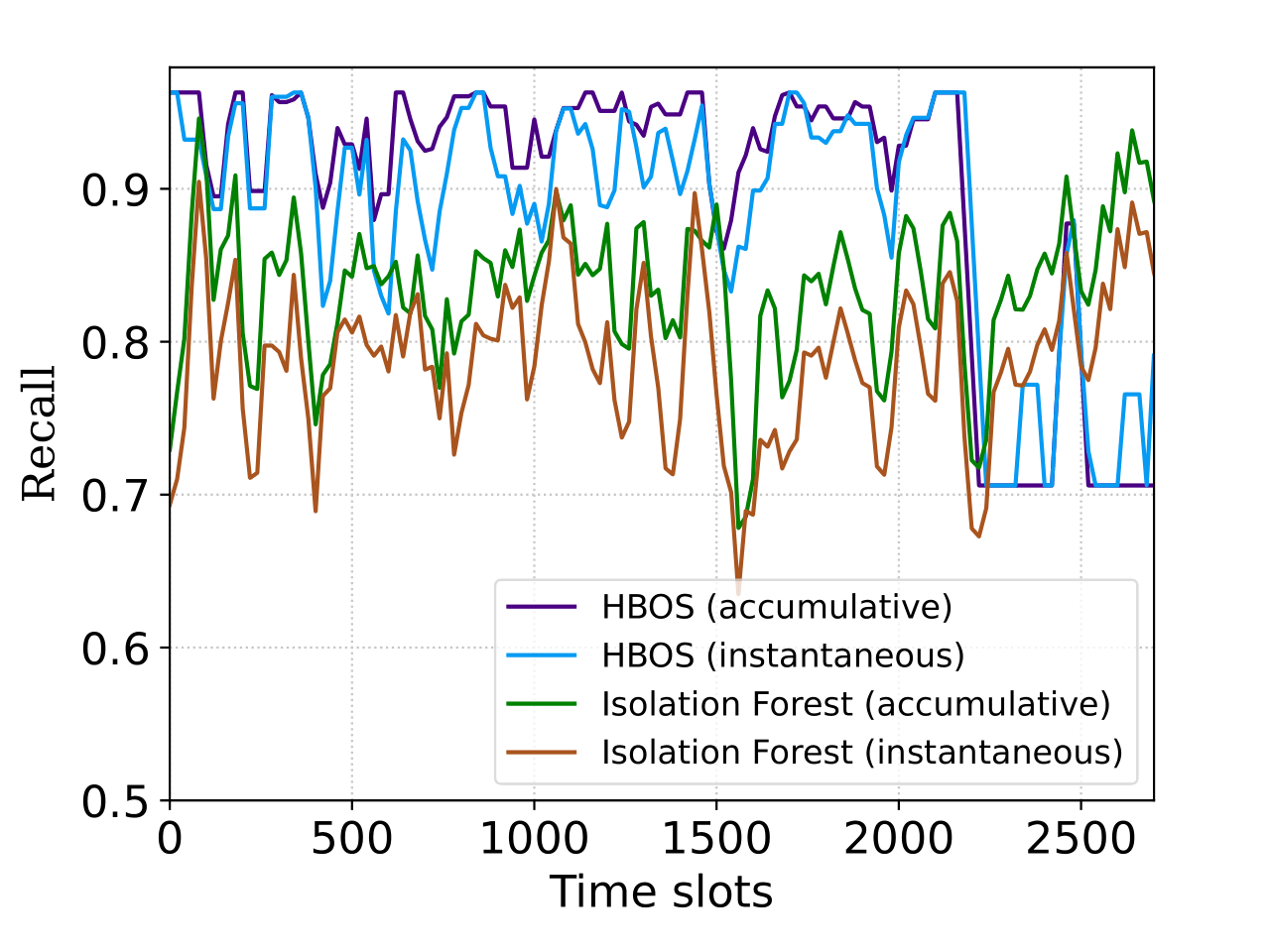}
		\label{fig:anomaly_recall_realtime} }
	\caption{Average Precision and Recall for anomaly detection over $2,700$ time slots on Network Traffic dataset. (a), (b): anomaly detection with different $\epsilon$. (c), (d): anomaly detection in time slot $t$ ($\epsilon = 2$). Due to unsupervised learning on unlabeled data, non-private results are treated as the ground truth.}
 \vspace{-0.1in}
	\label{fig:anomalydetection}
\end{figure*}

\begin{figure*}[!tbh]
	\centering
	\subfigure[Precision vs $\epsilon$]{
		\includegraphics[angle=0, width=0.245\linewidth]{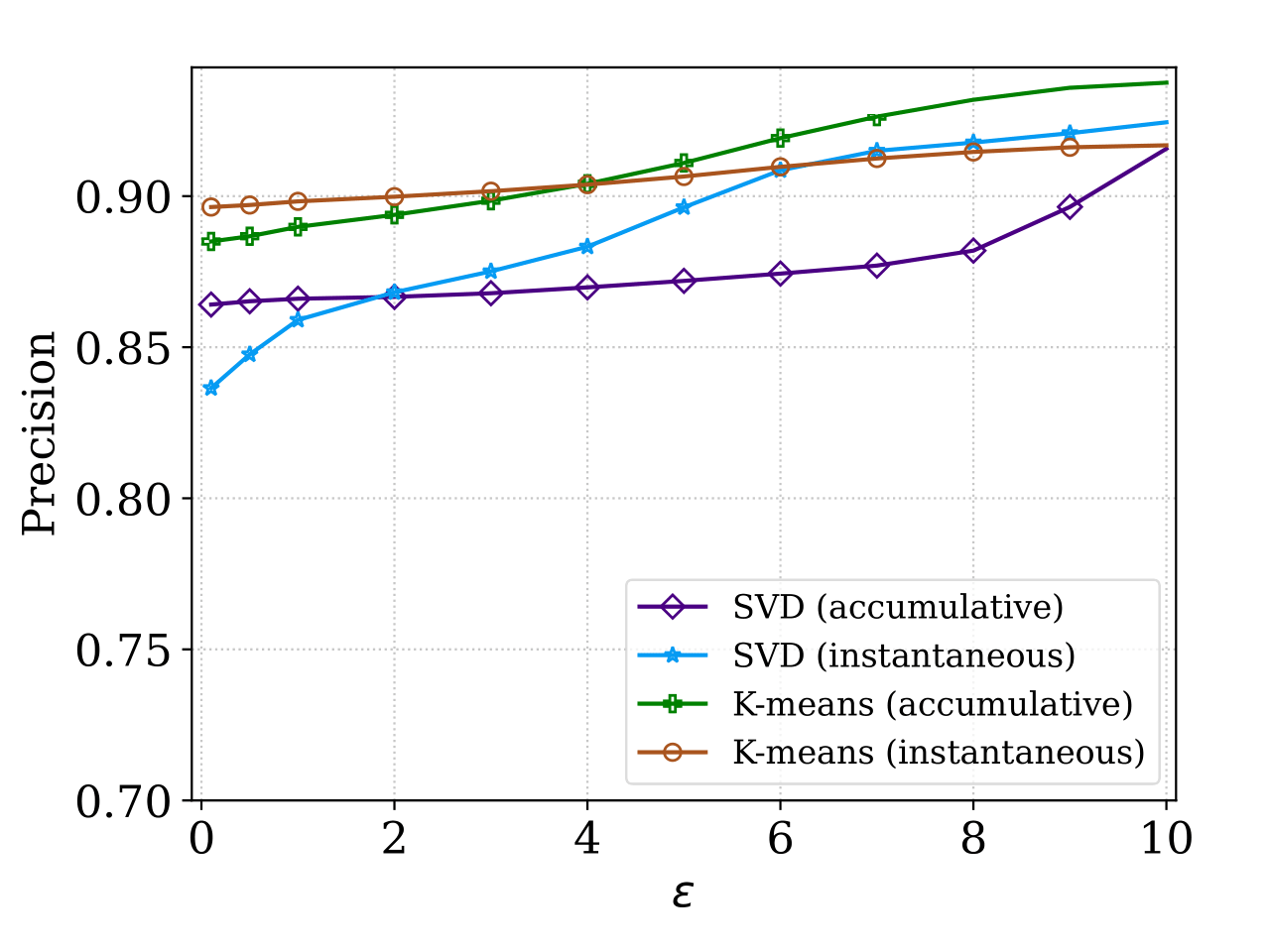}		
        \label{fig:recommend_acc1_eps} }
		\hspace{-0.18in}
	\subfigure[Recall vs $\epsilon$]{
		\includegraphics[angle=0, width=0.245\linewidth]{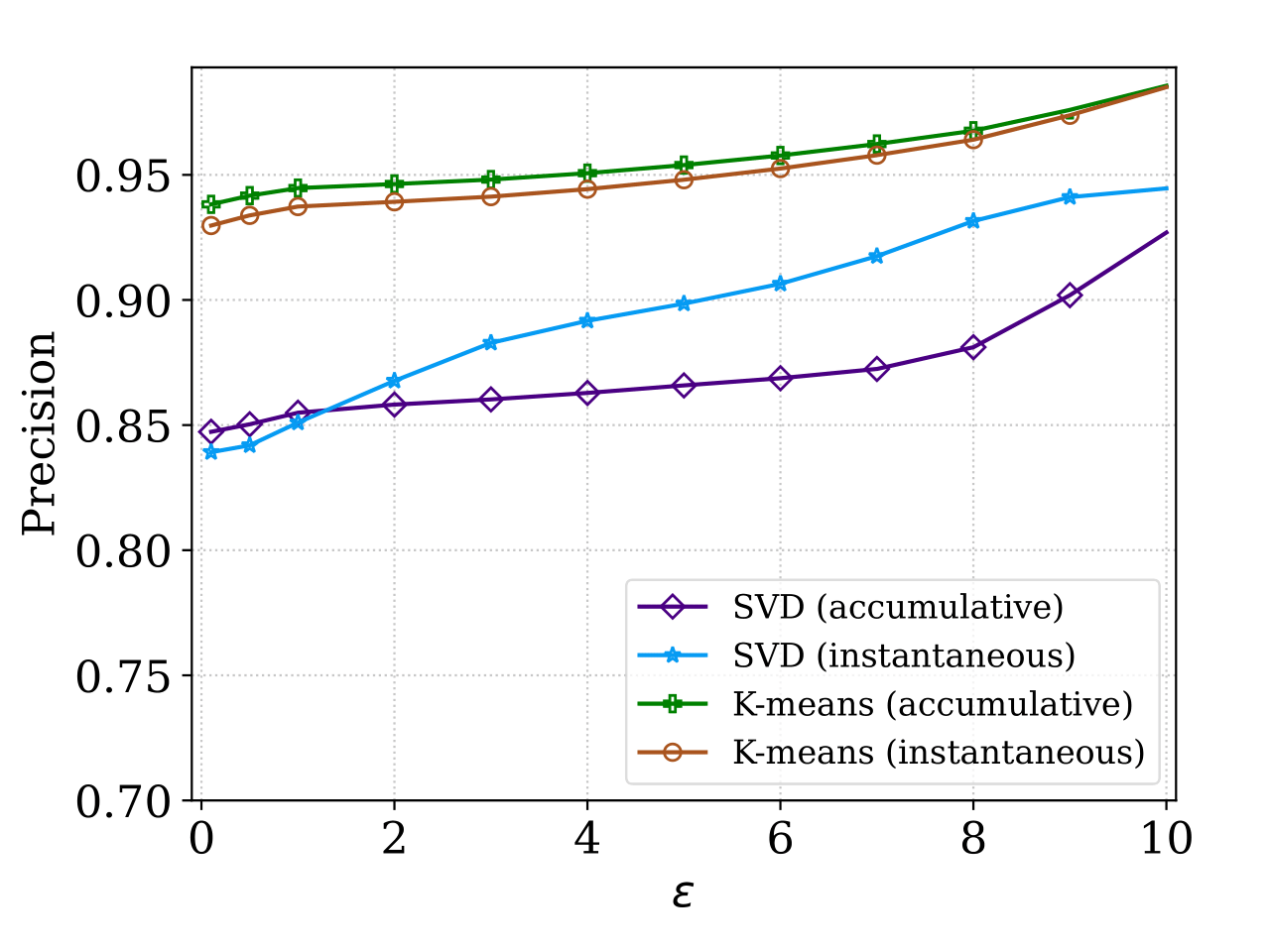}
		\label{fig:recommend_acc2_eps}}
		\hspace{-0.18in}
	\subfigure[Precision (moving average) vs $t$]{
		\includegraphics[angle=0, width=0.245\linewidth]{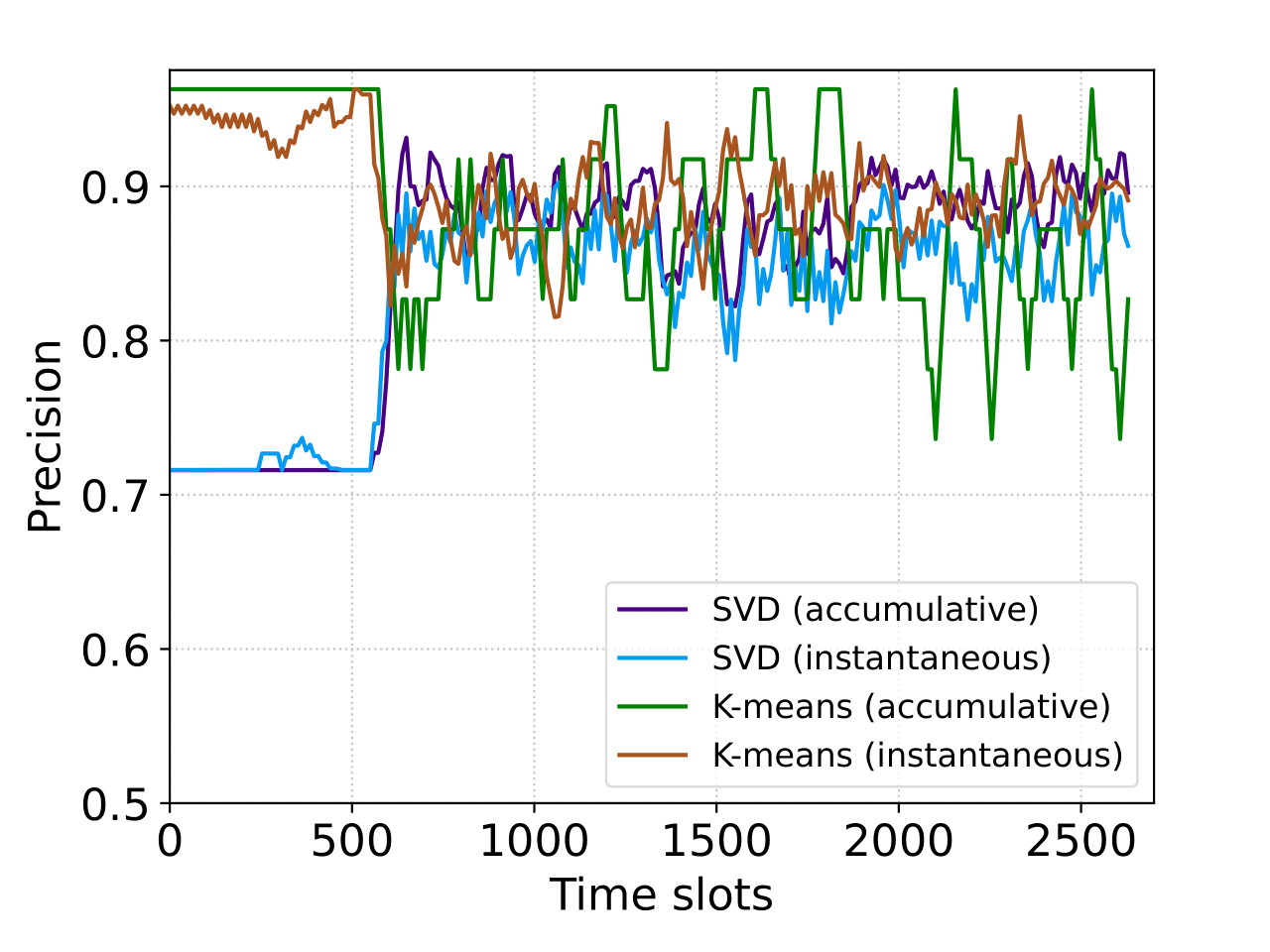}
		\label{fig:recommend_pre1_eps} }
		\hspace{-0.18in}
	\subfigure[Recall (moving average) vs $t$]{
		\includegraphics[angle=0, width=0.245\linewidth]{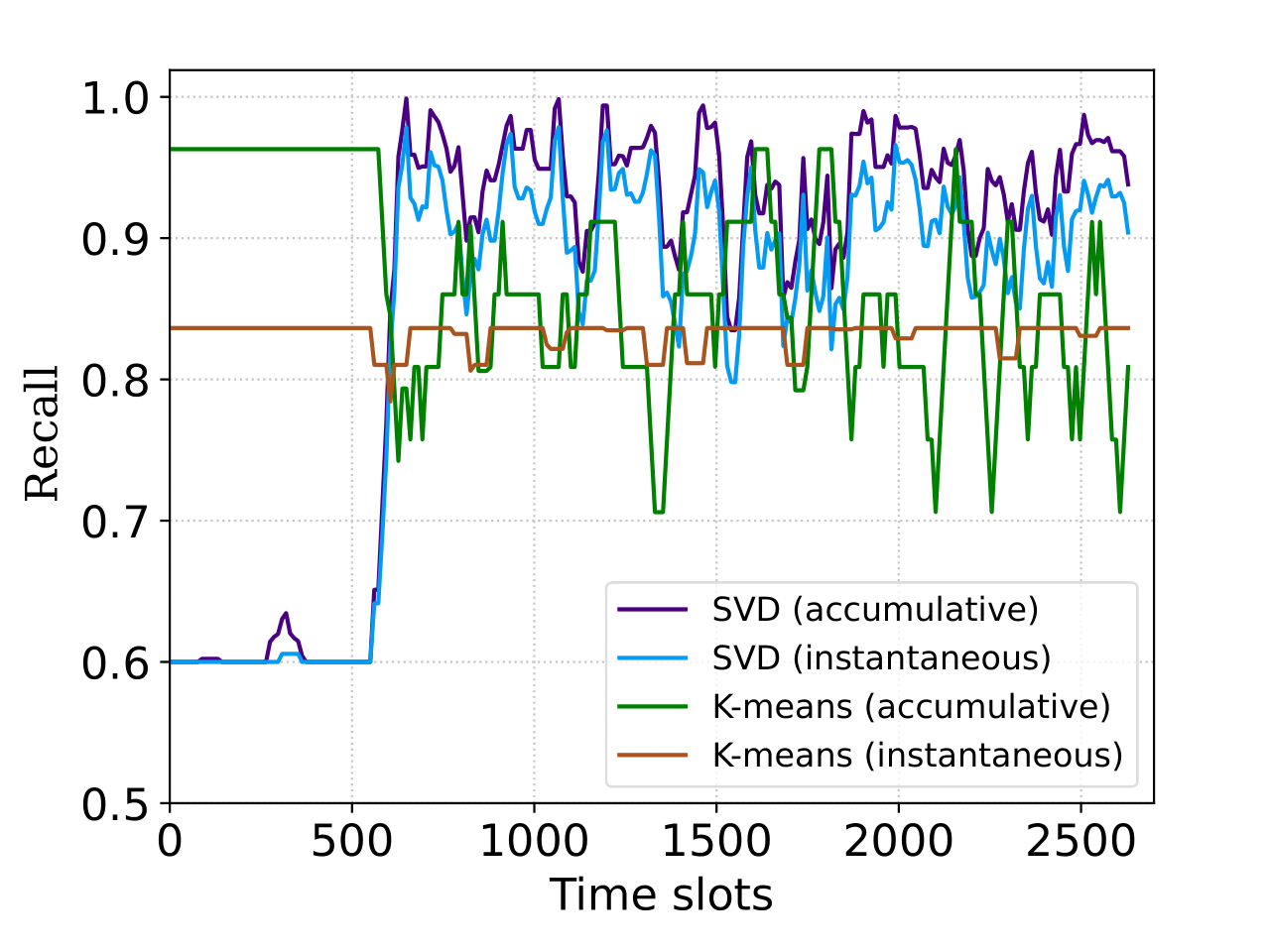}
		\label{fig:recommend_pre2_eps} }
	\caption{Average Precision and Recall for recommender system over $2,700$ time slots on USDA dataset. (a), (b): Precision and Recall vs $\epsilon$. (c), (d): Precision and Recall vs time slot $t$ ($\epsilon = 2$). Due to unsupervised learning on unlabeled data, non-private results are treated as the ground truth.}
 \vspace{-0.1in}
	\label{fig:recommendation}
\end{figure*}

We conduct experiments on the Network Traffic dataset and define the items with a score over a threshold as anomalies. DPI is expected to detect the anomalies from its real-time disclosures. We adopt Precision and Recall as the evaluation metrics compared to the non-private results, as shown in Figure \ref{fig:anomalydetection}. The Precision and Recall of both HBOS and Isolation Forest increase linearly with a growing $\epsilon$. Moreover, anomaly detection can be accurately performed by dynamically tracking the new data over a large number of time slots (as observed from the fluctuated Precision and Recall scores). Finally, anomaly detection over all the time slots (accumulative) tends to show relatively better results than that in specific time slots (instantaneous).

\subhead{Recomender System} 
For recommender systems, we utilize the SVD algorithm \cite{sarwar2000application,zhou2015svd}, which capitalizes on data distribution for delivering recommendations by efficiently managing dimensionality reduction and highlighting essential features. It decomposes a user-item matrix into three constituent matrices, capturing the interaction between users and items. By reducing the dimensionality, SVD effectively uncovers latent features indicative of user preferences. Then, SVD leverages the data distribution to recommend items that align with users' interests based on their past behavior. We also test the utility with the K-means \cite{zhang2018differential,yin2020improved} algorithm. Figure \ref{fig:recommendation} shows that both Precision and Recall increase as $\epsilon$ increases. They can be very high by considering the non-private results as the ground truth (e.g., 90\%+).

\subsection{Highly Dynamic Data Distributions}
\label{sec:highly}

To evaluate DPI on highly dynamic data streams, we generated a synthetic dataset over $3,000$ time slots. For each time slot, we generate a synthetic dataset by initializing $100$ items in each dataset and sampling their counts with the Gaussian distributions (mean $100$ and variance randomly chosen from $1, 4, 9, 16, 25, \dots, 100$). Figure~\ref{fig:streamplot} shows the true values in the synthetic datasets and the DPI outputs under different $\epsilon$ (since there are numerous values in each time slot, we plot the median values in the figure). Then, we can observe that the DPI outputs are close to the true values by well preserving the original data distribution.

Furthermore, Figure \ref{fig:dpi_changingpdf} shows the MSE and KL divergence of DPI outputs over time under different $\epsilon$. From the plots, we observe the values of MSE and KL divergence are very small even in case of small $\epsilon$ (e.g., $0.5$ and $1$). Moreover, although the distribution changes significantly over time, the MSE and KL divergence of DPI are very stable. The boxplots in Figure \ref{fig:dpi_boxplots} show the lower quantile (LQ), upper quantile (UQ), and median (M) of the DPI output densities. The output density here refers to the distribution of DPI outputs in terms of their frequency and spread. The spread of the distributions remains stable over time.

\begin{figure*}[!tbh]
	\centering
	\subfigure[MSE vs $t$]{
		\includegraphics[angle=0, width=0.245\linewidth]{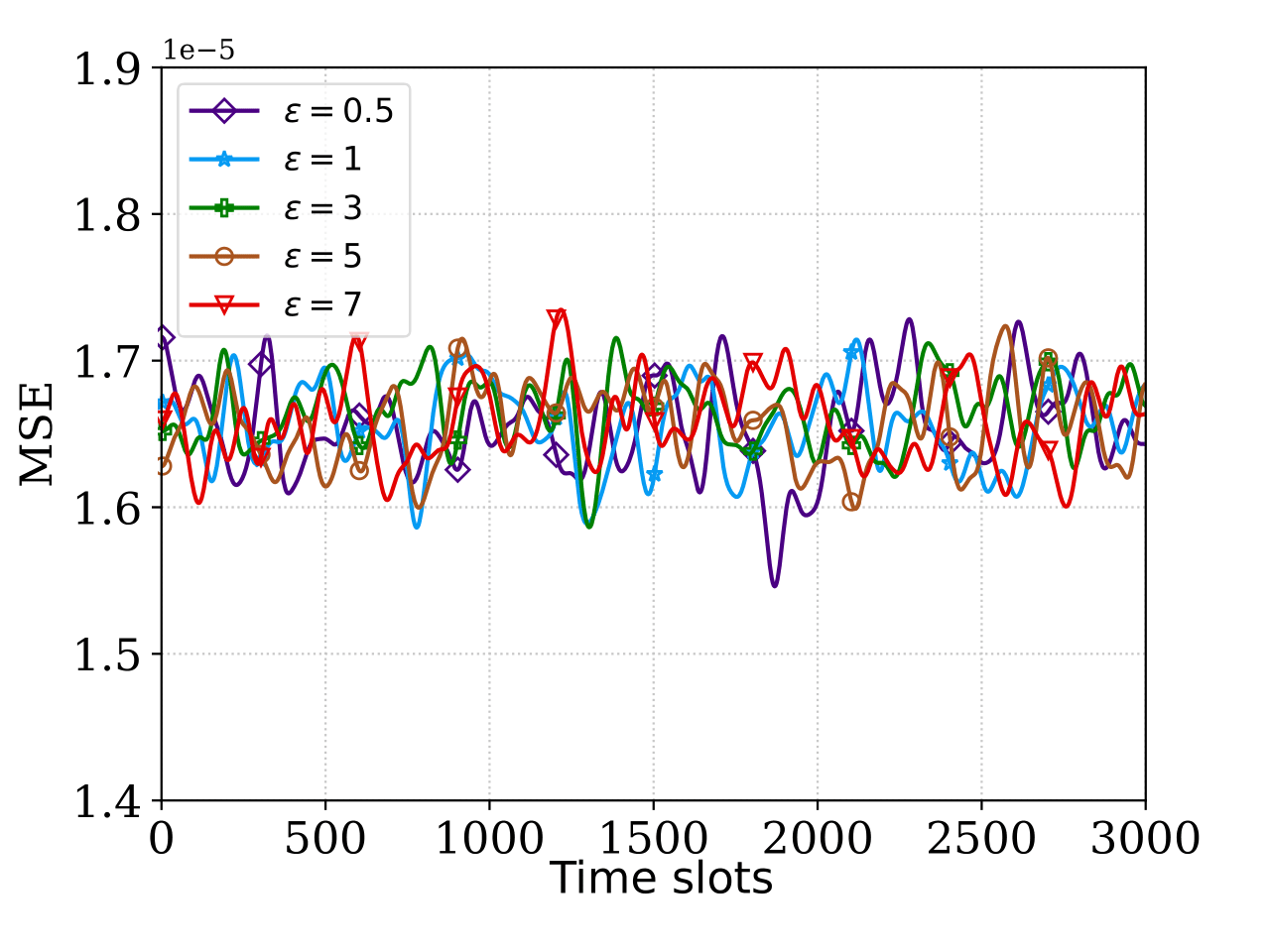}		
        \label{fig:dpi_changing_realtime1_mse} }
		\hspace{-0.18in}
	\subfigure[MSE  vs $t$]{
		\includegraphics[angle=0, width=0.245\linewidth]{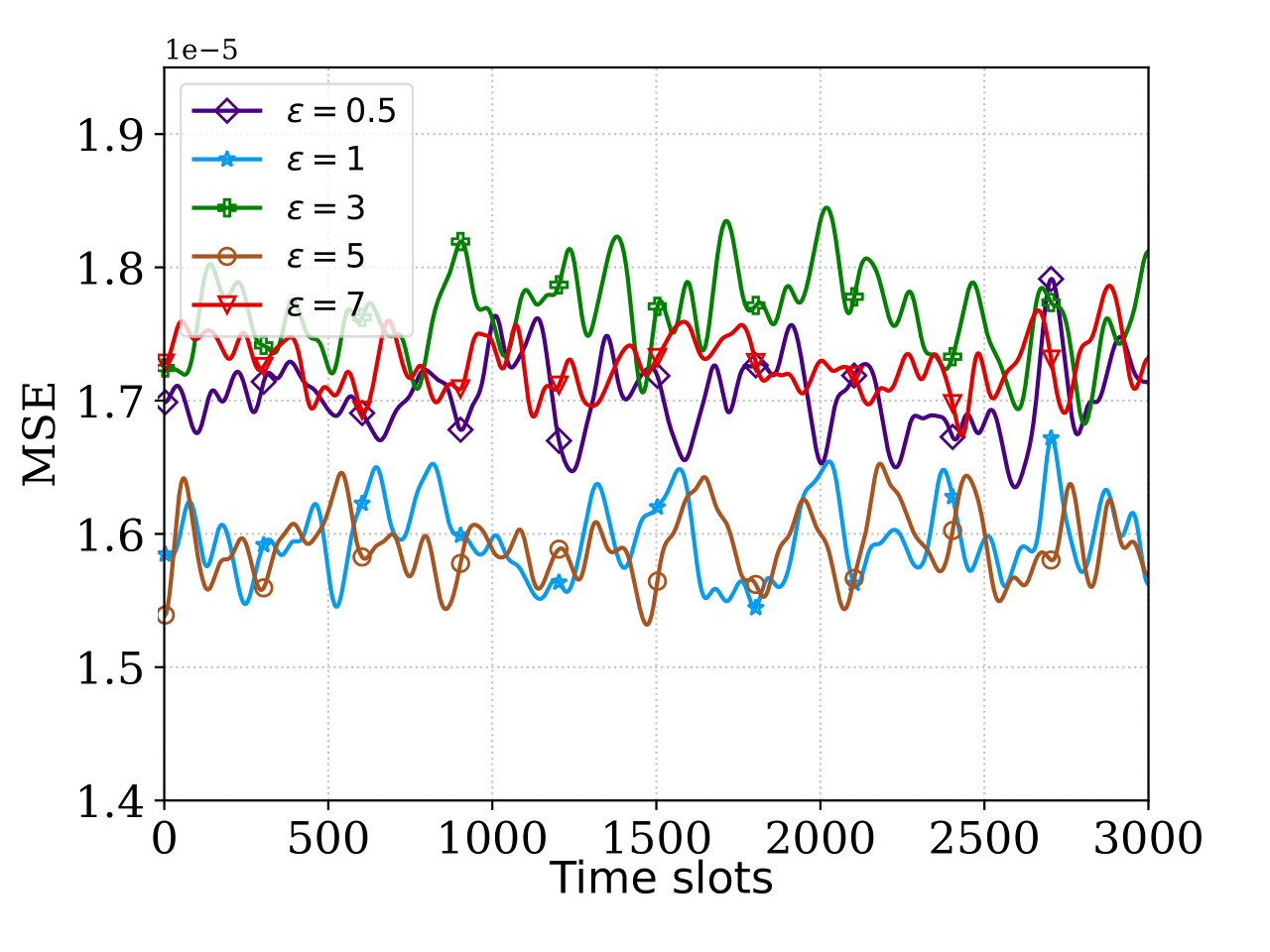}
		\label{fig:dpi_changing_realtime2_mse}}
		\hspace{-0.18in}
	\subfigure[KL div. vs $t$]{
		\includegraphics[angle=0, width=0.245\linewidth]{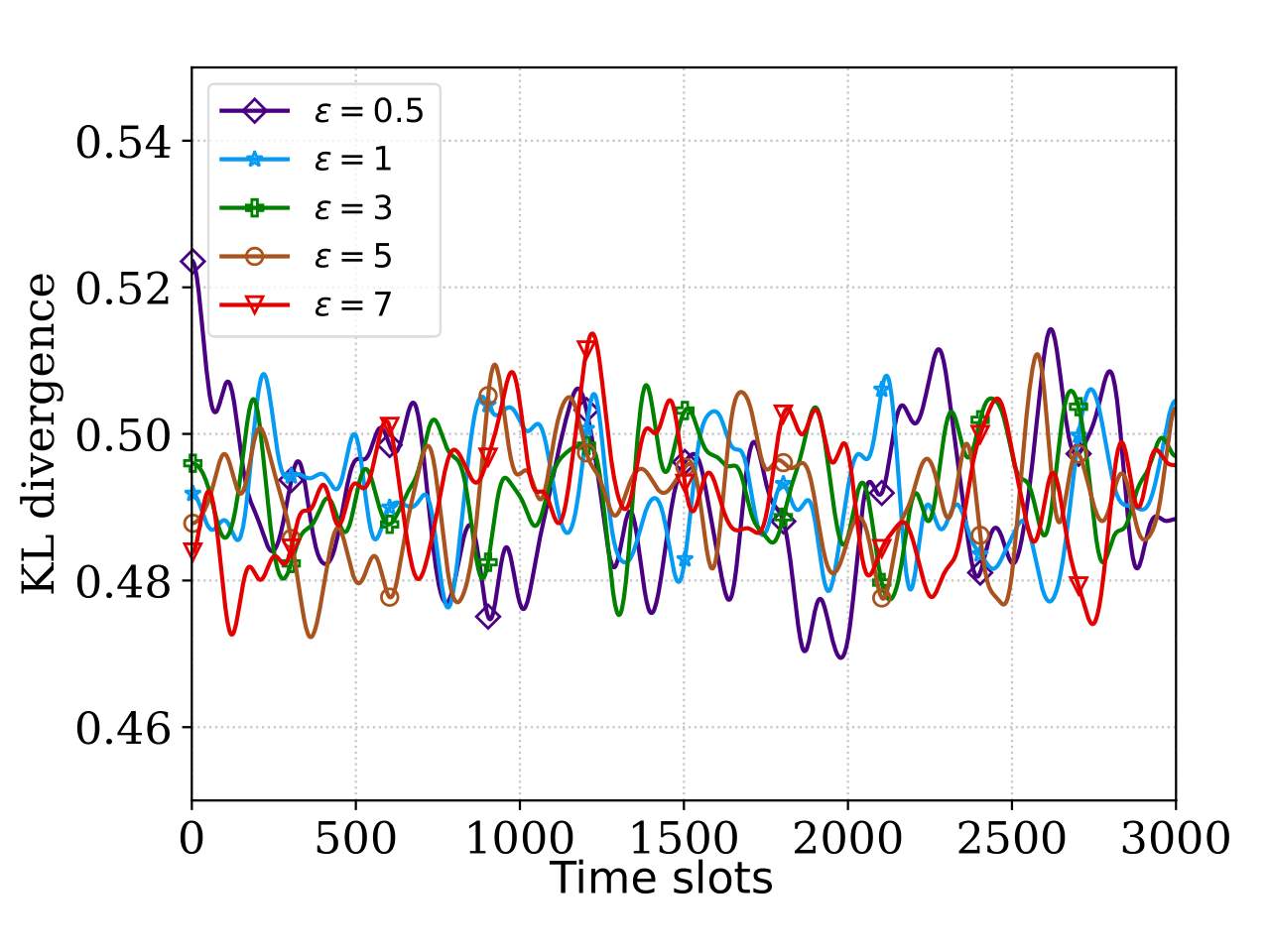}
		\label{fig:dpi_changing_realtime_kl} }
		\hspace{-0.18in}
	\subfigure[KL div. vs $t$]{
		\includegraphics[angle=0, width=0.245\linewidth]{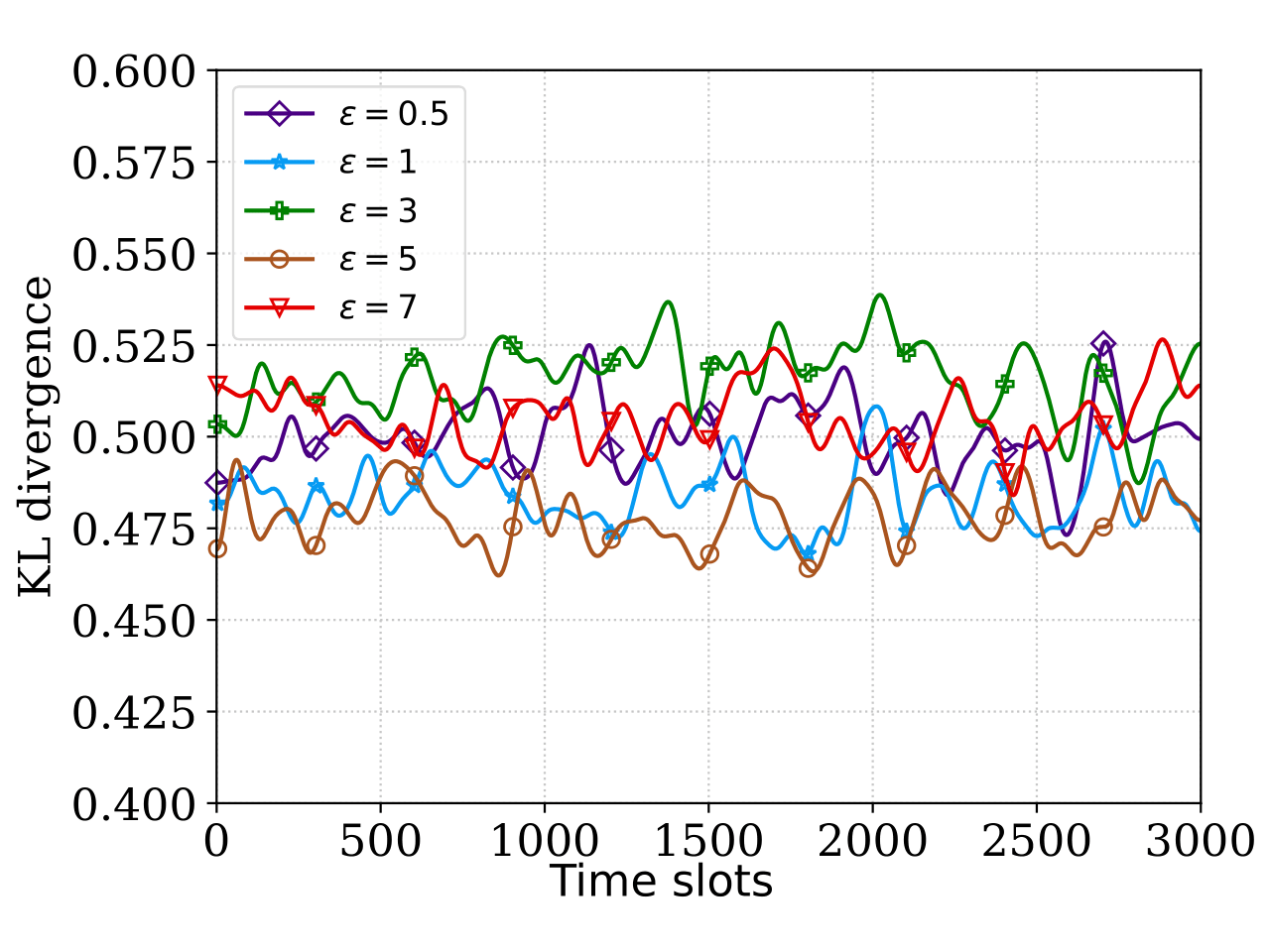}
		\label{fig:dpi_changing_realtime_kl2} }	
	\caption{Average MSE and KL divergence over $1,000$ time slots. (a) 
 (b): MSE vs $t$ (varying $\epsilon$). (c) (d): KL divergence vs $t$ (varying $\epsilon$).}
 \vspace{-0.1in}
	\label{fig:dpi_changingpdf}
\end{figure*}

\begin{figure*}[!tbh]
	\centering
 	\subfigure[Density vs $\epsilon$]{
		\includegraphics[angle=0, width=0.23\linewidth]{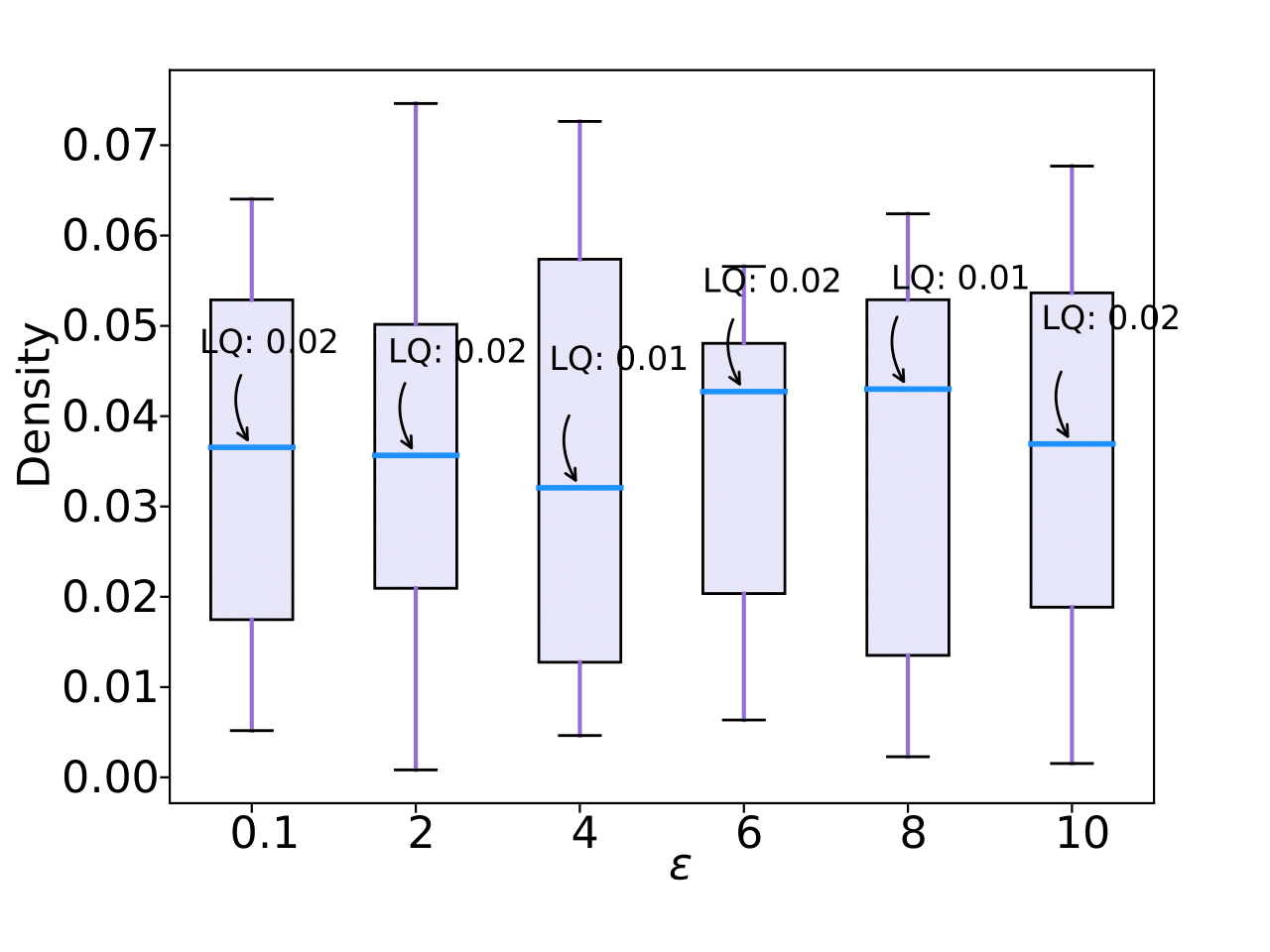}
		\label{fig:results_box_l} }
	\subfigure[Density vs $\epsilon$]{
		\includegraphics[angle=0, width=0.23\linewidth]{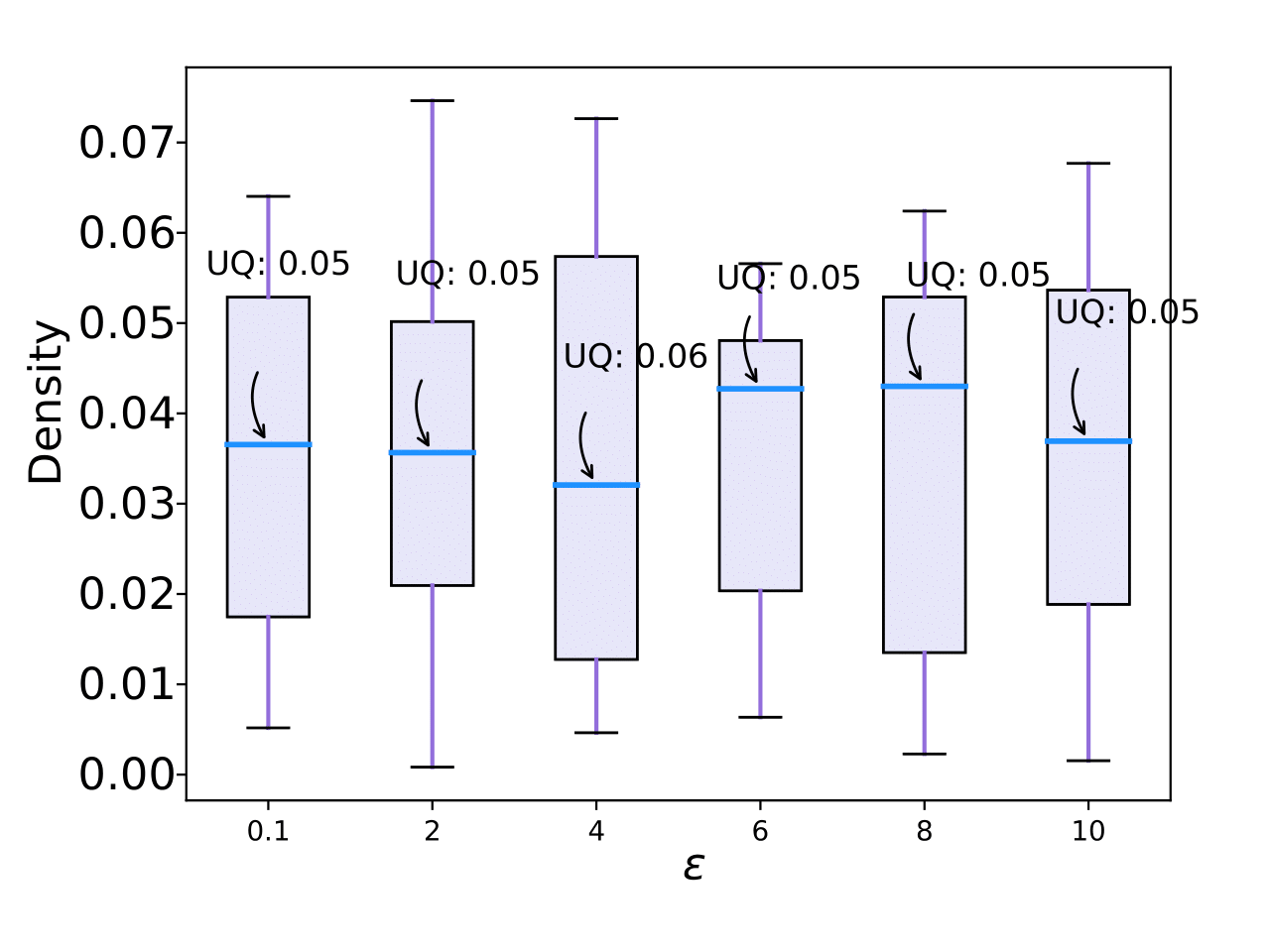}
		\label{fig:results_box_u} }
   	\subfigure[Density vs $t$]{
		\includegraphics[angle=0, width=0.23\linewidth]{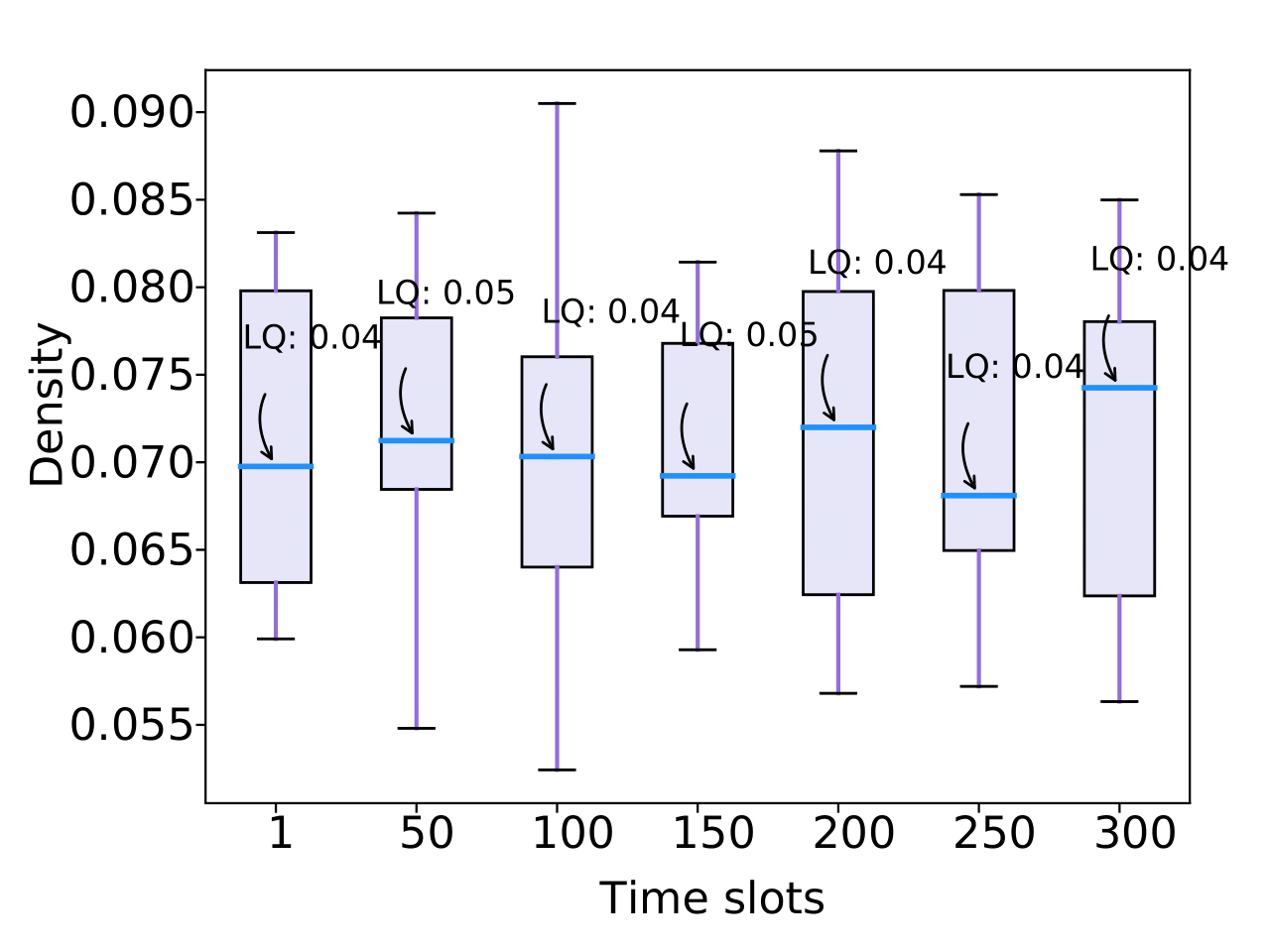}
		\label{fig:results_box_time_LQ} }
	\subfigure[Density vs $t$]{
		\includegraphics[angle=0, width=0.23\linewidth]{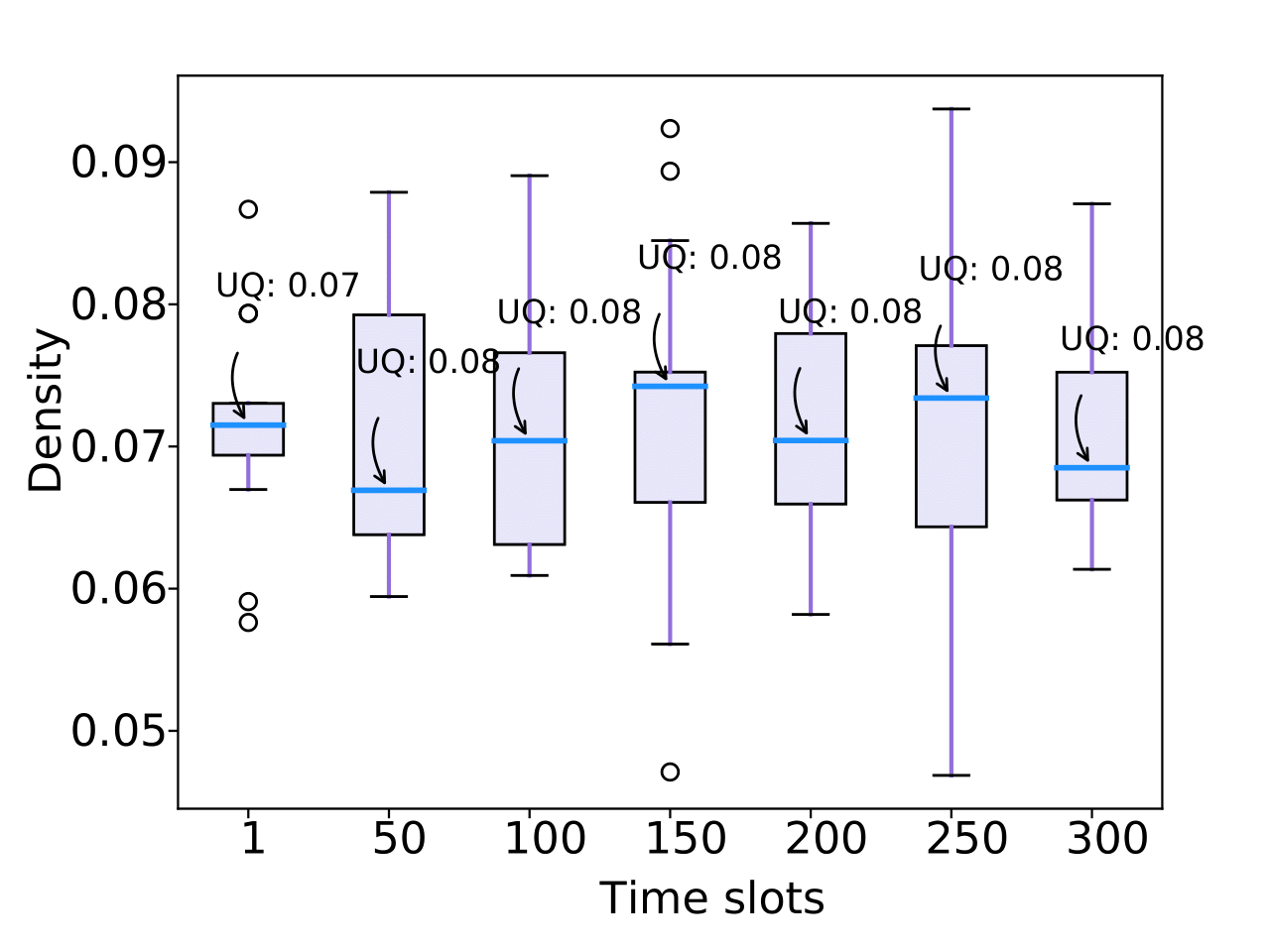}
		\label{fig:results_box_time_UQ} }
	\caption{Distribution spread of DPI output on the synthetic dataset with changing distribution. (a) (b): DPI output vs $\epsilon$. (c) (d): DPI output vs $t$. The boxplots show the lower quantile (LQ), upper quantile (UQ), and median (M) of the DPI output densities. The spread of the distributions remains stable.}
 \vspace{-0.1in}
	\label{fig:dpi_boxplots}
\end{figure*}

\begin{figure}[!h]
    \centering
	\subfigure[Median values vs $t$]{
		\includegraphics[angle=0, width=0.48\linewidth]{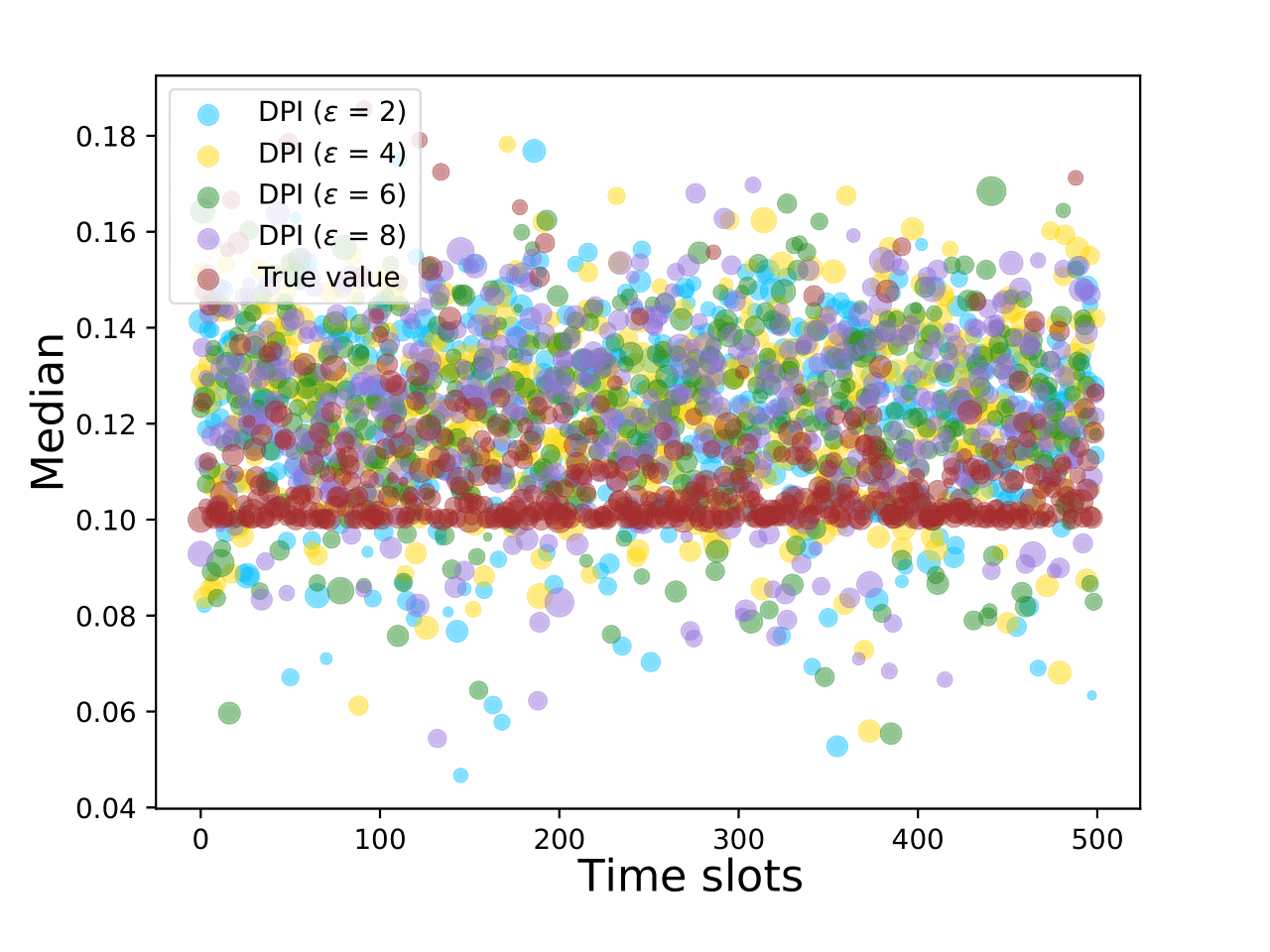}
		\label{fig:stream_mean_acc}}
	\hspace{-0.18in}
	\subfigure[Median values vs $t$]{
	\includegraphics[angle=0, width=0.48\linewidth]{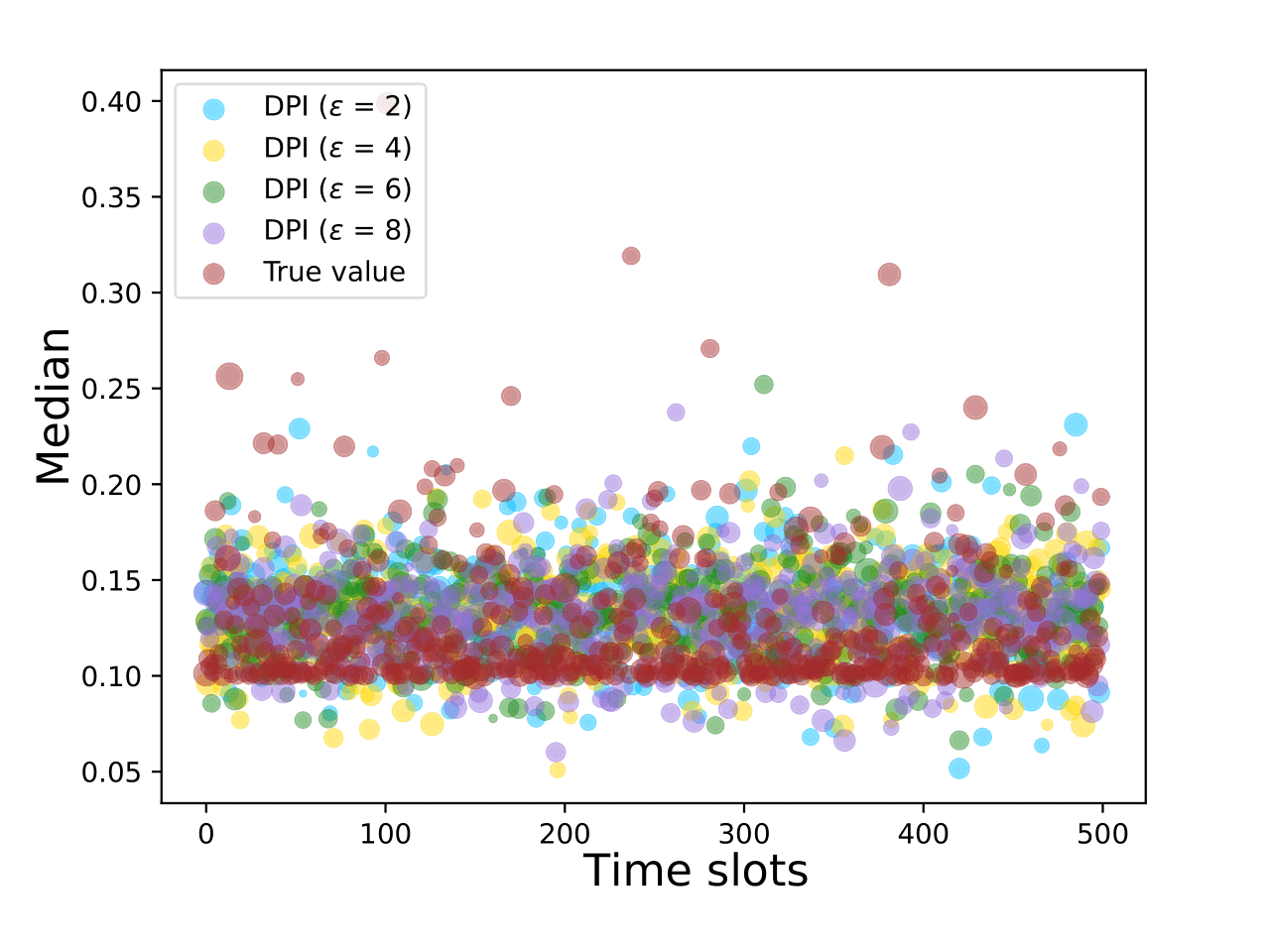}
	\label{fig:stream_mean_current}}
\caption{Real-time DPI outputs and true data (median).}\vspace{-0.1in}
\label{fig:streamplot}
\end{figure}

Finally, we conduct experiments on $20$ more synthetic datasets and different $\epsilon\in[0.5,10]$ (see Figure \ref{fig:box_plot_synthethic} in Appendix \ref{sec:synthetic}). First, we generate $10$ distinct synthetic datasets by initializing $1,000$ items in each dataset and sample their counts with the Gaussian distributions (mean $1,000$ and variances $1, 4, 9, 16, 25, \dots, 100$ for $10$ datasets, respectively). Second, we generate another $10$ distinct synthetic datasets with the domain size $10,000$ and similar settings. Figure \ref{fig:box_plot_synthethic} further proves the stability and effectiveness of DPI across varying data distributions and domain sizes.

 \section{Discussion}
\label{sec:discussion}
 
\subhead{Boosting and DPI} Boosting is a powerful technique in machine learning that combines multiple weak models to produce a stronger and more robust model. In the case of DPI, boosting is utilized to enhance privacy and accuracy through the continuously updated streaming data in each dynamic batch. In particular, we only run $T$ boosting rounds in the first batch to get a good sampler PDF. Subsequently, DPI will continuously update the synopses sampling distribution in each time slot to the underlying data distribution with privacy preservation. Upon that, the efficiency of DPI can be significantly improved. DPI can also be modified to utilize boosting technique for each round. In this way, the accuracy can be better but cost more privacy budget. 

\subhead{Complexity Analysis} 
The complexity analysis for DPI focuses on the primary computational components: the Boosting Algorithm, and the $0$-DP Synopsis Generation algorithm. In DPI, the Boosting Algorithm updates and normalizes distribution synopses with a complexity of $O(n)$, and the $0$-DP Synopsis Generation algorithm generates synopses through probabilistic sampling also with a complexity of $O(n)$. The overall time complexity $O(n)$ also aligns well with the standard streaming algorithms (over time series). Moreover, DPI consumes memory proportional to the domain size (which is tolerable even for large domains). 

\subhead{No Numerical Overflow for Extremely Tiny Budgets} 
The randomization in DPI is based on ``DP-Boosting'', which effectively translates extremely tiny budgets into negligible or even no reweighting. Thus, different from noise-additive DP mechanisms (e.g., Laplace and Gaussian), extremely tiny budgets in DPI do not entail any numerical overflow.

\subhead{Limitations} 
First, DPI presents a new privacy-utility tradeoff compared to traditional noise-additive mechanisms. Utility loss manifests as ineffective tracking of new stream data, potentially in extreme cases of dynamic data over extremely long periods, where privacy budget depletes, possibly necessitating DPI system restart for optimal utility with a renewed budget. However, DPI may not necessarily need to restart in practice for two reasons: (1) our random budget allocation ensures that budgets are not depleted even after a large number of time slots (Figure \ref{fig:remaining_budget} demonstrates that budgets are adequate after $500$ time slots), and (2) given extremely tiny privacy budgets, DPI maintains low MSE and KL divergence for different data distributions over extremely long periods. Figure \ref{fig:extremetimeslot} shows that the results are still stable even after $200,000$ time slots (no need to restart), evaluated on the synthetic data generated in Section \ref{sec:highly}. Indeed, we cannot find extreme cases of the low utility of DPI in our extensive empirical studies.

\vspace{-0.15in}

\begin{figure}[!h]
    \centering
	\subfigure[Remaining budget vs $t$]{
        \hspace{0.04in}
		\includegraphics[angle=0, width=0.48\linewidth]{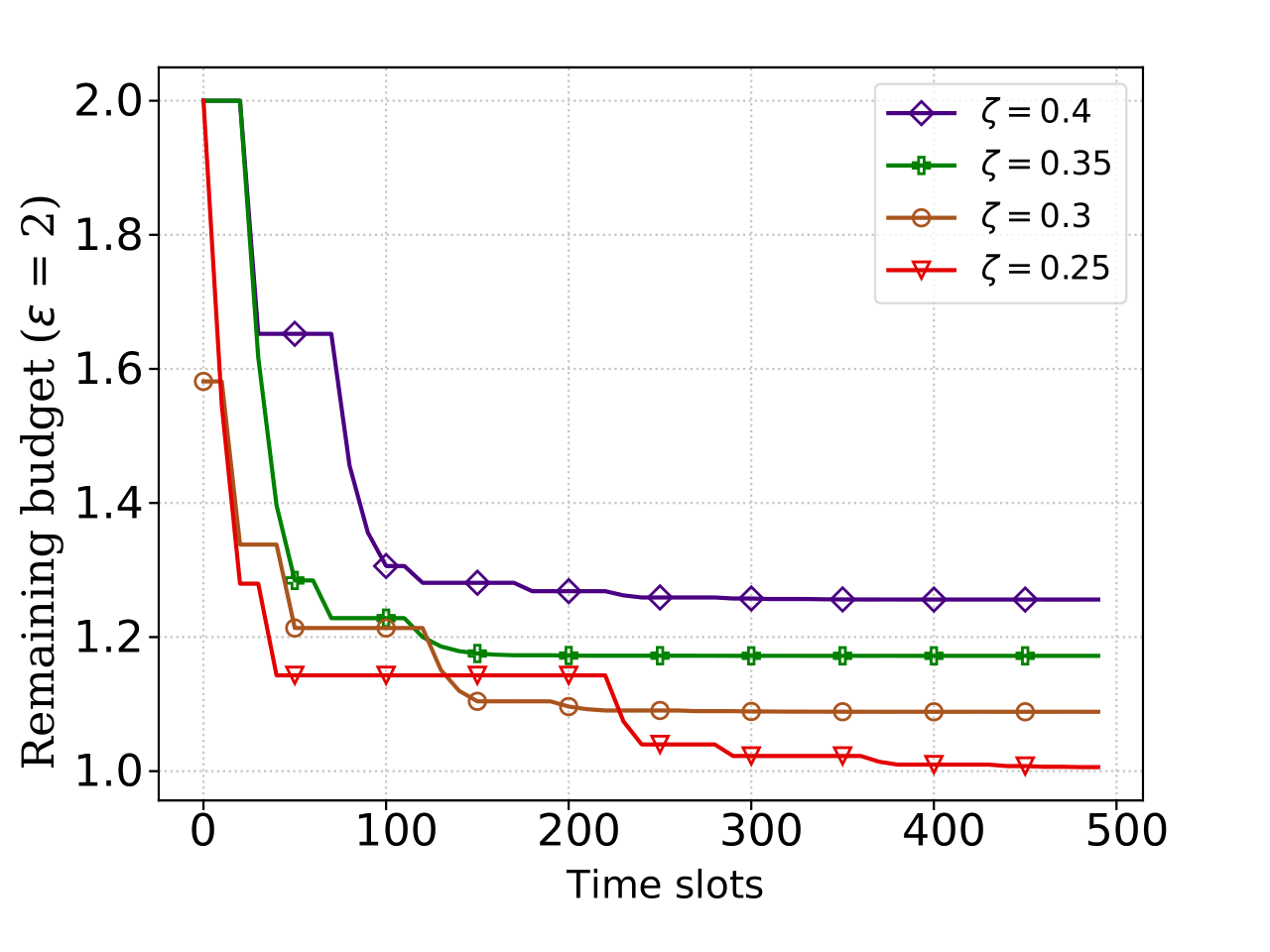}
		\label{fig:remaining_budget1}}
	\hspace{-0.24in}
	\subfigure[Remaining budget vs $t$]{
	\includegraphics[angle=0, width=0.48\linewidth]{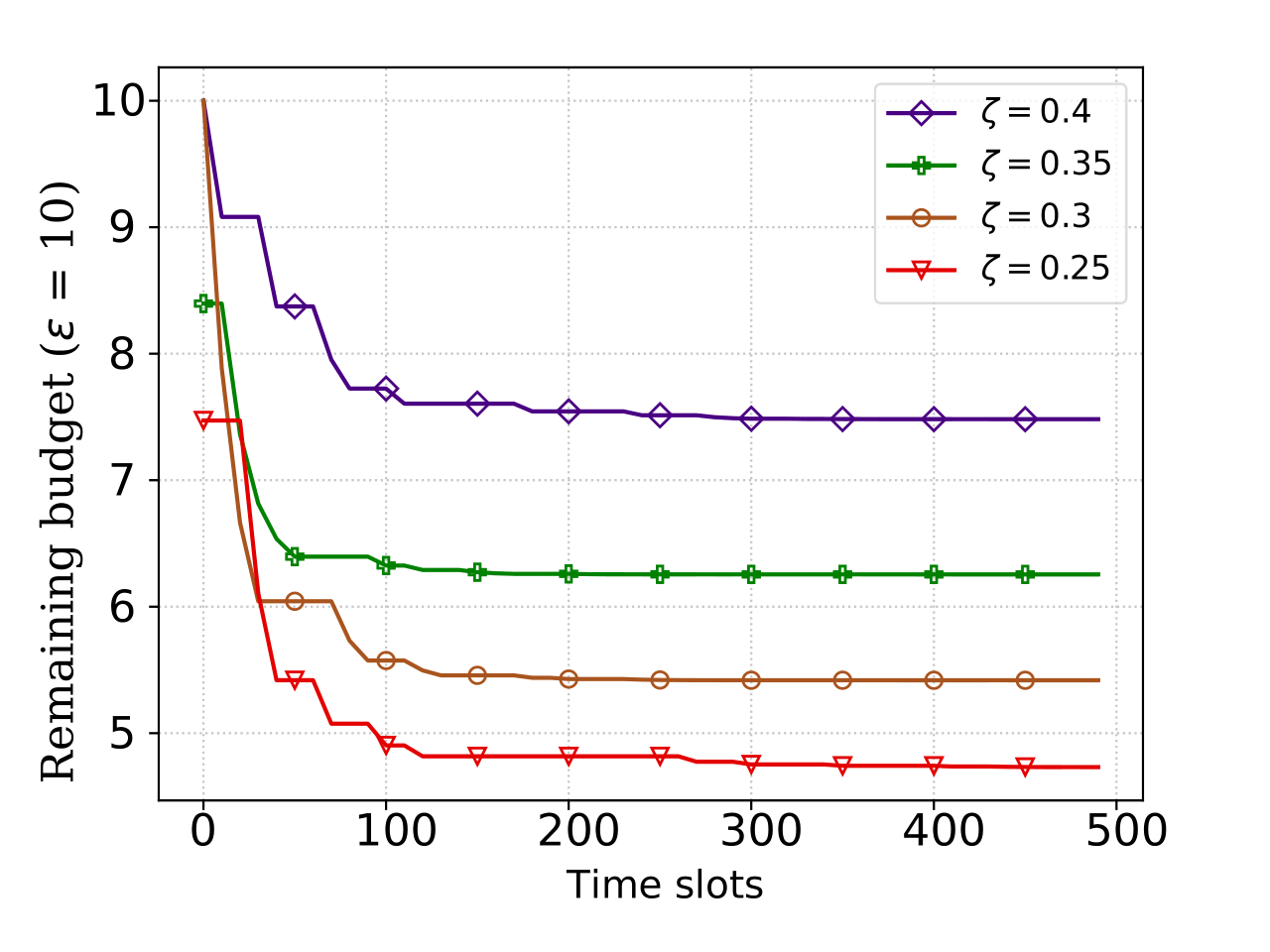}
	\label{fig:remaining_budget2}}\vspace{-0.05in}
\caption{Remaining budget for DPI with the random privacy budget allocation: (a) $\epsilon = 2$, and (b) $\epsilon = 10$.}
\label{fig:remaining_budget}
\end{figure}

\begin{figure}[!h]
    \centering
	\subfigure[Median (moving average) vs $t$]{
		\includegraphics[angle=0, width=0.48\linewidth]{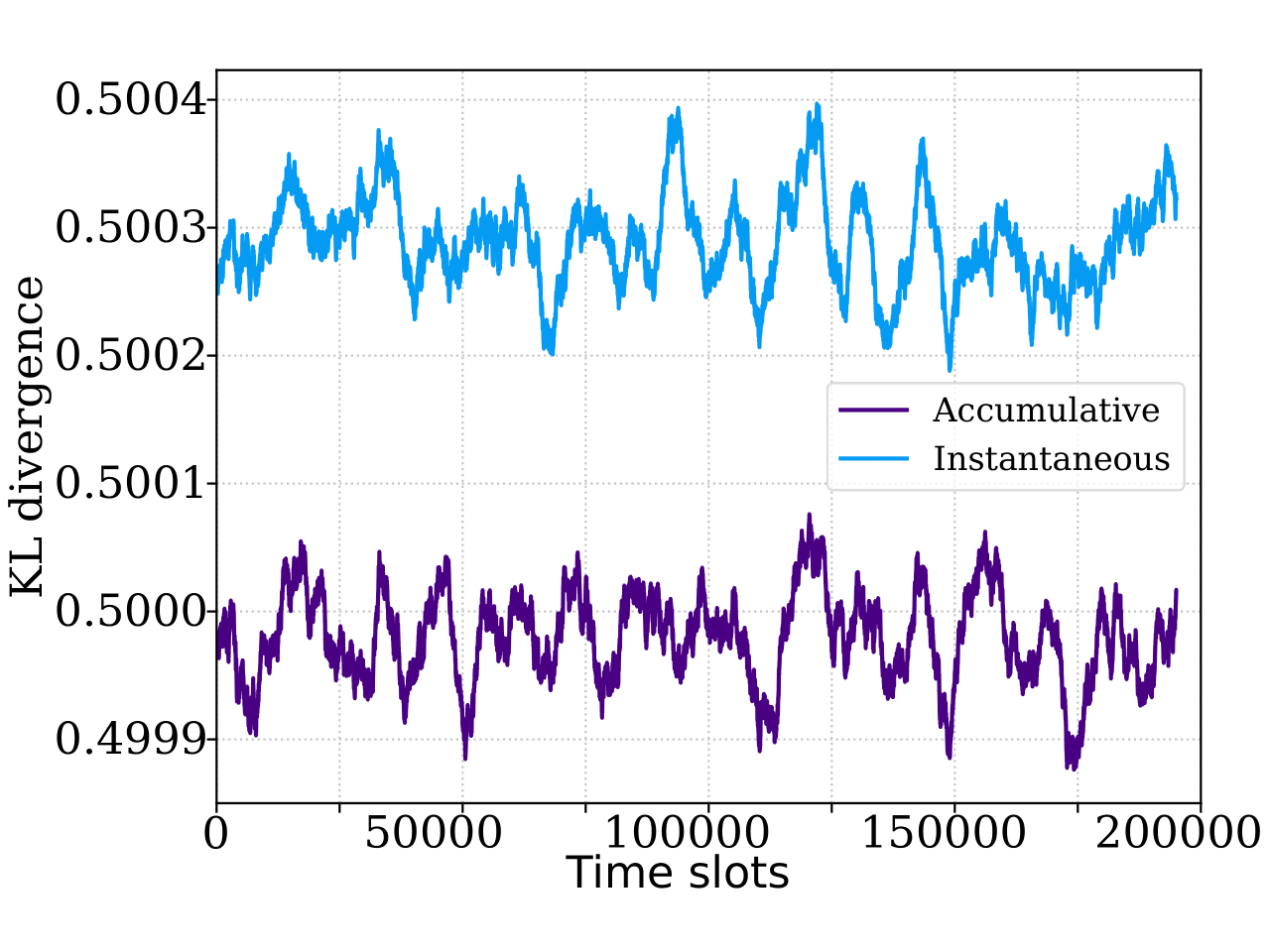}
		\label{fig:extreme_timeslot_kl}}
	\hspace{-0.18in}
	\subfigure[Median (moving average) vs $t$]{
	\includegraphics[angle=0, width=0.48\linewidth]{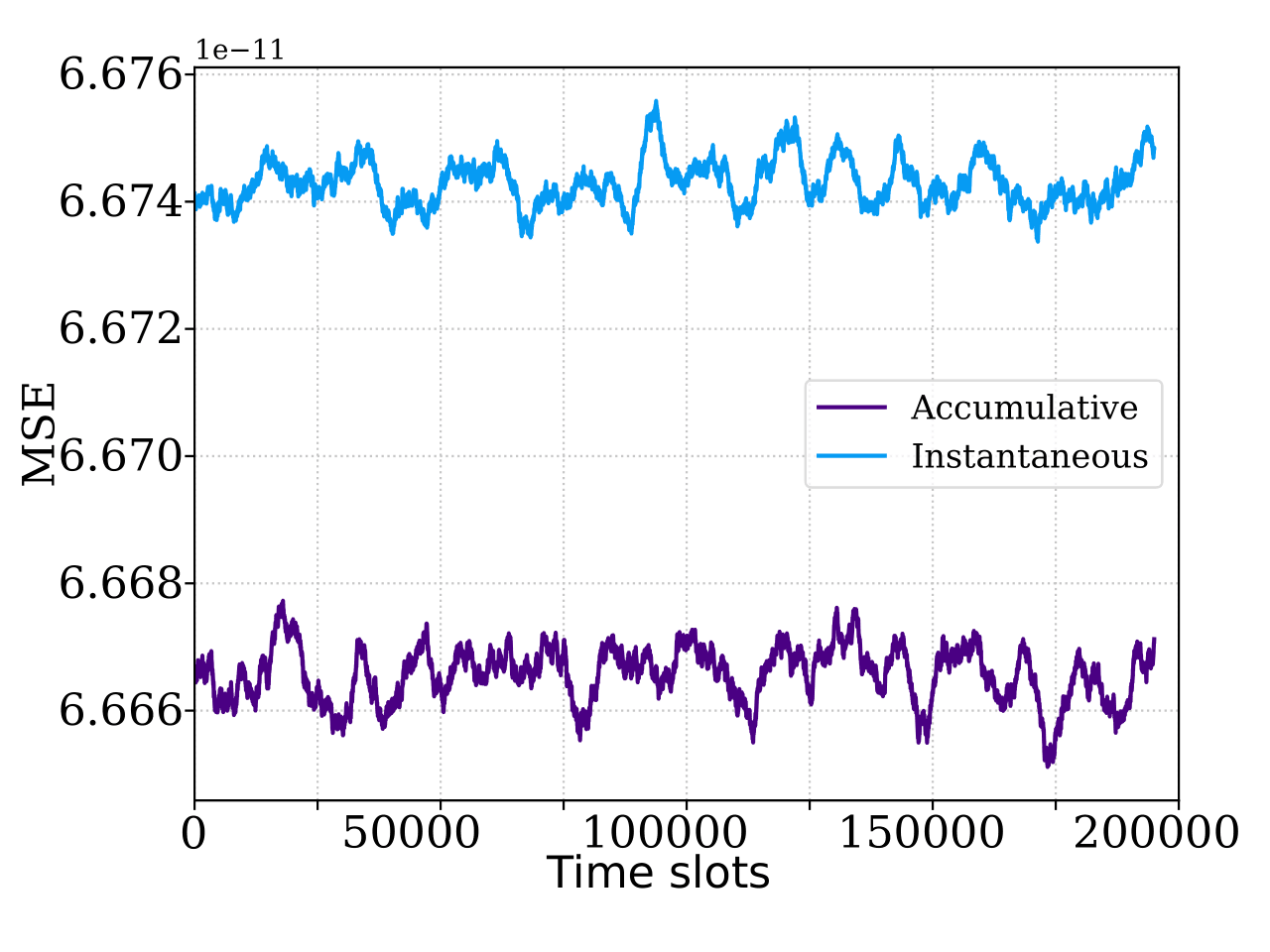}
	\label{fig:extreme_timeslot_mse}}\vspace{-0.05in}
\caption{Evaluation of DPI on $200,000$ time slots on synthetic datasets (a) KL divergence and (b) MSE:  domain size $10,000$.}\vspace{-0.2in}
\label{fig:extremetimeslot}
\end{figure}
\vspace{-0.1in}
Second, DPI outputs data as a series of PDFs instead of count histograms, suitable for most downstream analyses. Nevertheless, counts can be estimated from these PDFs by discretizing the domain into bins, calculating each bin's probability mass, and scaling these probabilities by the total data size. For instance, a 0.02 probability in a bin for ages 15-20 in a million-record dataset implies an estimated 20,000 individuals in that age group. To privately disclose counts in certain applications, DPI only needs to assume the release of the non-private total data size at time slot $t$.

\section{Related Work}
\label{sec:related}
\vspace{-0.1in}
\noindent\textbf{Data Streaming with DP.} Since the early studies on the private data streaming \cite{10.1145/882082.882086,chandrasekaran2002streaming}, 
DP models were proposed to protect the streaming data. Ebadi et al. \cite{10.1145/2676726.2677005} proposed a personalized DP for dynamically adding records to the database. Meanwhile, Liu et al. \cite{10.1145/2792838.2800191} also proposed personalized DP with the weighted posterior sampling to reduce the extra Gaussian noise to the parameter space. To further apply DP to streaming data, many works \cite{bolot2013private, 10dwork2010differential,mir2011pan,chan2012differentially,chan,pegasus,wang2016rescuedp} focus on continuously publishing statistics computed over events. Chan et al. \cite{chan} propose several methods that can handle binary streams of different users over potential infinite streams. Several works \cite{pegasus, perrier2018private, wang} follow a similar idea of applying partition algorithms to handle numeric values in the streaming data. However, event-level privacy is not strong enough to prevent sensitive data. Later, Kellaris et al. \cite{10.14778/2732977.2732989} proposed the $w$-event privacy to protect event sequence occurrence in $w$ timestamps. This notion converges to user-level privacy when $w$ is set to infinite, but the noise would be unbounded. 
Aiming to achieve user-level DP, FAST \cite{fan2013adaptive} was proposed to release real-time statistics without full DP infinite data stream. Dong et al. \cite{10179466} propose continual observation mechanisms for various fundamental functions under user-level differential privacy without requiring a priori restrictions on the data.

\subhead{Query Boosting with DP} Boosting has been widely used for improving the accuracy of learning algorithms \cite{schapire1998improved}. There are also many algorithms that boost the DP results. In \cite{practicalprivacy}, they combined private learning with DP, using learning theory to comprehend database probabilities, potentially enhancing DP's distortion reduction. For general counting queries, Dwork et al. proposed a base synopsis generator \cite{dwork2009complexity}, and also constructs an appropriate base synopsis generator for any set of low-sensitivity queries (not just counting queries) \cite{dwork2010boosting}. These well-designed synopsis generators can significantly boost the utility of DP results. Recall that our synopsis generator significantly differs from them to support infinite data streaming. 

\subhead{Data Streaming with LDP} LDP \cite{ErlingssonPK14} involves users perturbing their data before sending it to an untrusted server for aggregation, posing challenges in privacy accumulation for streaming data with limited utility (e.g., a sequence of locations \cite{WangH0QH22}). Joseph et al. \cite{joseph2018local} proposed a framework for continuous data sharing under LDP, consuming privacy budget based on distribution changes rather than collection periods to reduce the overall privacy expenditure. 
Wang et al. \cite{wang} proposed an algorithm using the Exponential mechanism with a quality function to publish a stream of real-valued data under both centralized and LDP. Li et al. \cite{ldp-memory} proposed an LDP approach to heavy hitter detection on data streams with bounded memory. 
These works can provide event-level privacy protection. Bao et al. \cite{bao2021cgm} 
proposed a novel correlated Gaussian mechanism for $(\epsilon, \delta)$-LDP on streaming data aggregation. However, CGM needs to update the privacy budget periodically. Compared to these, Ren et al. \cite{ren2022ldp} then proposed a population division framework that not only avoids the high sensitivity of LDP noise to the budget division but also requires less communication. However, it can only provide the $w$-event privacy.

\subhead{Applications} Differential privacy (DP) has been the de facto rigorous privacy solution for learning algorithms \cite{238162,10.1145/2976749.2978318,VaidyaSBH13,wu2020value, mohammady2020r2dp,wei2020federated,friedman2010data,WangSFSH22}. McSherry et al. \cite{10.1145/1557019.1557090} applied DP to provide personal recommendations. Chen et al. \cite{10.1145/3485447.3512192} proposed to publish the rating matrix of the source domain with DP guarantee. Feng et al. \cite{7543858} proposed a topic privacy-relevance parameter method for top-$k$ recommendations. Okada et al. \cite{okada2015differentially} used three queries on statistical aggregation on outliers to detect the occurrence of anomalous situations with differential privacy guarantee.

\section{Conclusion}
\label{sec:conclusion}
\vspace{-0.1in}
This paper addresses the significant challenges in an open problem: differential privacy for infinite data streams, which is crucial for various real-time monitoring and analytics applications. DP has been widely used to protect streaming data, but it has key limitations in terms of unbounded privacy leakage and sensitivity for ensuring user-level protection. We propose a novel solution, DPI, to effectively bound privacy leakage and enhance accuracy in real-time analysis on infinite data streams. We have conducted extensive theoretical studies to prove the convergence of privacy bound and the bounded errors (utility loss) for DPI over infinite data streams. We have also conducted comprehensive experiments to validate DPI on various real streaming applications and datasets.
\vspace{-0.1in}
\section*{Acknowledgments}
\vspace{-0.1in}
\noindent The authors sincerely thank the anonymous shepherd and all the reviewers for their constructive comments and suggestions. This work is supported in part by the National Science Foundation (NSF) under Grants No. CNS-2308730, CNS-2302689, CNS-2319277, and CMMI-2326341. It is also partially supported by the Cisco Research Award, the Synchrony Fellowship, the National Key Research and Development Program of China under Grant 2021YFB3100300, and the National Natural Science Foundation of China under Grants U20A20178 and 62072395.

{\footnotesize \bibliographystyle{IEEEtran}
\bibliography{bib}}

\appendices
\section{DP with AdaBoosting \cite{dwork2010boosting}}
\label{app:ada}

\begin{figure}[ht]
	\centering		\includegraphics[width=1.02\linewidth]{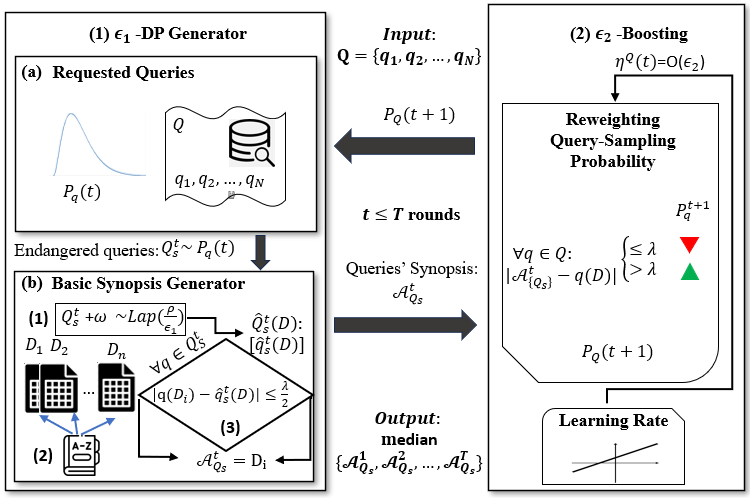}
  \vspace{-0.1in}
	\caption[Optional caption for list of figures]
	{The architecture of DP AdaBoosting~\cite{dwork2010boosting}}\vspace{-0.1in}
	\label{fig:dwork}
\end{figure}

\section{Proofs}

\subsection{\textbf{Proof of Theorem}~\ref{thm:utilitydpi}}
\label{thm:3}
As stated in Claim 4.3 in~\cite{dwork2010boosting}, the probability of inaccuracy produced by a DP boosting algorithm can be described in terms of $(\eta, \beta)$-base learners. In DPI, we can analogously consider $(\eta^{\mathcal{A}(q,t)}_i, \eta^{Q}_i, \beta)$-base learners at each iteration $i$.

After $t$ rounds of DP boosting, the probability of inaccuracy in the DPI framework is upper-bounded by $\prod_{i=1}^{t} \sqrt{\left[(1-4 (\eta^{Q}_i)^2)\left(1-4 (\eta^{\mathcal{A}}i)^2\right)\right]}$. By taking the natural log of this bound, a version similar to the privacy bound (the sum of natural logs) is obtained. Specifically,
\begin{eqnarray*}
& \log \left[\mathbb{P}(\text{bad outcome})\right]\leq \\
& \frac{1}{2} \sum{i=1}^{N \rightarrow \infty} \log \left(1-2 \eta_i\right) \left(1+2 \eta_i\right)
\end{eqnarray*}

Thus, this completes the proof. \qed

\subsection{\textbf{Proof of Remark~\ref{rem:utility}}}
\label{rem1:utility}
Our construction involves a decreasing sequence of privacy budgets, thus proving that $\eta_i$ must also decrease. We assume $\eta(x)<0.5$ is a converging function of at least an exponential order (\textit{exp}($^{-\zeta \cdot x}$)). Then, we have
\small
\begin{align}
  \eta^{\prime}(x) \geq -\zeta \eta(x), \quad \forall x\in \mathbb{R}^{+}
\end{align}
where $\zeta=\mid\inf_{x\in [1,\infty)} \frac{d\eta(x)}{dx} \mid$. By replacing $2\eta(x)$ with $u$, and for constants $m=\min\limits_{x\geq 1}(\eta(x))$ and $M=\max\limits_{x \geq 1}(\eta(x))$, we obtain
\begin{align*}
\small
    \mid L \mid \geq  -\frac{1}{4} \int_{m}^{M} \frac{\log [(1+u)(1-u)]}{\zeta u}du
\end{align*}

but given that the Spence's function~\cite{bloch2000higher} is defined as $Li_2(z)=-\int_{0}^{z} \frac{\log [(1-u)]}{u}du$, for $|z| \leq 1$, we have 
\begin{align*}
\small
    \mid L \mid \geq \frac{1}{4\zeta}\left(Li_2(M^2) - Li_2(m^2)+ 4Li_2(m)- 4Li_2(M)\right) 
\end{align*}
Here, we have utilized the following two results:
\begin{eqnarray}
\label{eq:facts}
    &\int_{0}^{z} \frac{\log [(1+u)]}{u}du=-Li_2(-z)\\
    &Li_2(z)+Li_2(-z)=\frac{1}{2}Li_2(z^2) \nonumber
\end{eqnarray}

Thus, this completes the proof. \qed

\subsection{\textbf{Proof of Theorem~\ref{rem1:privacy}}} 
\label{rem:privacy}
According to Equation~\ref{eq:tradeoff}, the continuous format of DPI utilizes the following privacy budget.
\begin{align*}
\small
    \epsilon=\frac{4}{\mu} \int_0^{\infty} \log\left[\frac{1+2 \eta(x)}{1-2 \eta(x)} \right] dx,
\end{align*}
By following a similar process as that described in the Remark~\ref{rem:utility}, we will have 
\begin{align*}
\small
    \epsilon\geq \frac{4}{\zeta \mu} \int_m^{M} \cfrac{\log\left[\frac{1+u}{1-u} \right]}{u}  dx.
\end{align*}
Finally, by utilizing the two results in Equation~\ref{eq:facts}, we have 
\begin{align*}
\small
    \epsilon \geq \frac{2}{\mu \zeta}\left(Li_2(M^2) - Li_2(m^2)\right). 
\end{align*}
Thus, this completes the proof. \qed

\subsection{\textbf{Proof of Theorem}~\ref{thm:final}}
\label{thm:4}
By establishing lower bounds for both privacy and accuracy loss, we can determine the optimal series. Thus, our optimization problem can be succinctly stated as follows.
\begin{eqnarray}
\small
\label{eq:tradeoff1}
& \min \limits_{0\leq m<M\leq 1} \frac{1}{8\zeta}\left(Li_2(M^2) - Li_2(m^2)+ 4Li_2(m)- 4Li_2(M)\right) \nonumber \\
&\text{w.r.t. }  \frac{2}{\mu \zeta}\left(Li_2(M^2) - Li_2(m^2)\right)=\epsilon,
\end{eqnarray}
but Spence's function has the following infinite series form: 
\begin{align*}
\small
    Li(z)=\sum_{t=1}^{\infty} \frac{z^t}{t^2}
\end{align*}
Therefore, we have $
    Li_2(M)-Li_2(m)=\sum_{t=1}^{\infty} \frac{M^t-m^t}{t^2}$. Due to
\begin{align}
\small
\label{epsseries}
    Li_2(M^2)-Li(m^2)=\sum_{t=1}^{\infty} \frac{M^{2t}-m^{2t}}{t^2}= O(\epsilon),
\end{align}
However, we note that minimizing \textit{probability loss} is equivalent to if we maximize $m^{t}+M^{t}$. This is true because our objective function has two opposite sign terms $A=[Li_2(M^2) - Li_2(m^2)]>0$ and $B= 4Li_2(m)- 4Li_2(M)<0$, where $B$ appears in $A$, and minimization is translated into maximizing $(m^t+M^t)$ (using the sequence of \textit{mean theroems}). Since Spence's function is strictly increasing, maximizing terms $m^t+M^t$, $\forall c$ means that $M=1$.

Thus, when $M=1\Rightarrow C\left(Li(1)-Li\left(m^2\right)\right)=\epsilon$, we have 
\small
\begin{align*}
& \Rightarrow Li\left(m^2\right)=Li(1)-\frac{\varepsilon}{C}=\frac{\pi^2}{6}-\frac{\epsilon}{C},  C=\frac{2}{\mu \zeta}\\
& \Rightarrow m^2=Li_2^{-1}(\frac{\pi^2}{6}-\frac{\epsilon}{C})\end{align*}
\small
\begin{align}
\label{eq:m}
    &\Rightarrow m=\sqrt{Li_{2}^{-1}\left(\frac{\pi^2}{6}-\frac{\epsilon}{C}\right)} 
\end{align}
Since each term in the series representation of the overall $\epsilon$ (Equation~\ref{epsseries}) has to be equal to the budget spent over the same index's disclosure, we have

\small
\begin{align*}
&\frac{2}{\mu \zeta} \cdot \sum_{t=1}^{\infty} \frac{M^{2t}-m^{2 t}}{t^2 } =\frac{4}{\mu} \log \left(\frac{1+2\eta_t}{1-2\eta_t}\right) \\
&\Rightarrow e^{ \cfrac{1-m^{2 t}}{2\zeta t^2 }} =\frac{1+2\eta_t}{1-2\eta_t} \\
&\Rightarrow \eta_t=\cfrac{\left[e^\frac{1-(Li_2^{-1}(\frac{\pi^2}{6}-\frac{\epsilon}{C}))^t}{t^2\zeta}-1\right]} {2\cdot\left[e^\frac{1-(Li_2^{-1}(\frac{\pi^2}{6}-\frac{\epsilon}{C}))^t}{t^2\zeta}+1\right]}
\end{align*}

Thus, this completes the proof. \qed

\section{Approximation of the Universe of PDFs}

\label{subsec:pool}

\begin{algorithm}[!tbh]
\footnotesize
\SetKwInOut{Input}{Input}\SetKwInOut{Output}{Output}

\KwIn{streaming dataset $D$, size of dataset $n$, sample size $N$}
 
\KwOut{the universe of synopses $\mathcal{A}$}

Initialize the quantization precision $p \in (0,1)$ to make the corresponding precision level $n_p=\frac{1}{p}+1$ equal to datset size $n$

$k$ $\leftarrow$ distinct value number of data $D$

Initialize the synopses pool $\mathcal{A}$ as empty

\For{$i \in N$}  
{Set equal sampling probability for each distinct value in the domain $Prob=\frac{1}{k}$

\tcc{Sampling n times from $k$ distinct values with $Prob=\frac{1}{k}$ and have a number of times each distinct value show among n times}
$t$ $\leftarrow$ Multinomial($n$, $k$, $Prob$)

\tcc{Sum of $P_i$ is $1$}
$P_i \leftarrow t*p$

$\mathcal{A}$.append($P_i$)

}
\Return the universe of synopses $\mathcal{A}$
\caption{$0$-DP Synopsis Generator}
\label{algm:zero-syno}
\end{algorithm}

When $n$ is sufficiently large, we can employ an asymptotic approximation to analyze the behavior. This approximation corresponds to the normal distribution limit of the Binomial distribution.

For large $n$, we have: $H \approx k - \frac{1}{2} \ln(2\pi ne) + \frac{1}{2} \sum_{i=1}^{k} \ln p_i$, where the natural logarithm is used for entropy calculation. Note that this approximation is not intended to provide an upper or lower bound, but rather an approximate form for large $n$. In the special case where the $p_i$ are uniform, the expression simplifies to:

\[ H \approx k - \frac{1}{2} \ln(2\pi ne) - \frac{k}{2} \ln k \]

While this is not an exact expression, it yields a highly accurate estimate of $H$ when $n$ is large.

\begin{lemma}
Let $H(\mathcal{A})$ denote the entropy of the constructed pool of synopses $\mathcal{A}$, and let $H_{\text{ideal}}$ denote the entropy of the ideal multinomial PDF. As the number of samples drawn from the multinomial distribution, denoted as $N$, approaches infinity, the ratio $\frac{H(\mathcal{A})}{H_{\text{ideal}}}$ converges to 1.
\end{lemma}

\begin{proof} This proof is based on the concept that as the number of samples approaches infinity, the constructed pool of synopses becomes more representative of the ideal multinomial PDF. By constructing $\mathcal{A}$ from the multinomial distribution, which is designed to distribute the available probability units uniformly, the synopses generated cover a comprehensive range of possible outcomes. As a result, the entropy of $\mathcal{A}$ approaches the entropy of the ideal multinomial PDF. 
Therefore, the ratio $\frac{H(\mathcal{A})}{H_{\text{ideal}}}$ converges to 1.
\end{proof}

\begin{theorem} Let $H(\mathcal{A}; n_p, k)$ denote the entropy of the constructed pool of synopses $\mathcal{A}$, given a precision level $n_p$ (equivalent to $1/p$, where $p$ is the quantization precision) and a domain size $k$. For sufficiently large values of $N$ (number of samples drawn), $H(\mathcal{A}; n_p, k)$ approximates the entropy expression for the ideal multinomial PDF.
\label{them:mul1}
\end{theorem}

\begin{proof} The proof of Theorem \ref{them:mul1} involves comparing the entropy expression for the ideal multinomial PDF, denoted as $H_{\text{ideal}}$, with $H(\mathcal{A}; n, k)$. For a large value of $n$, the entropy of a heterogeneous multinomial PDF can be approximated as:

\begin{equation}
H_{\text{approx}}(N, k) \approx \frac{k-1}{2}\ln(2\pi Ne) - \frac{k}{2}\ln(k),
\end{equation}

where $N$ represents the number of samples drawn from the multinomial distribution \cite{10.1137/1123020}.

By examining the entropy expression for the ideal multinomial PDF, which depends on $k$, and comparing it with $H(\mathcal{A}; n_p, k)$ for different values of $n$ and $k$ in the constructed pool of synopses, we can establish a relationship between the two. Specifically, for sufficiently large values of $N$, $H(\mathcal{A}; n_p, k)$ approaches $H_{\text{approx}}(N, k)$, indicating that the constructed pool of synopses accurately represents the ideal multinomial PDF. Figure \ref{fig:multinomial} shows the information theory of $0$-DP Synopsis Generation.
\end{proof}

\vspace{-0.2in}

\begin{figure}[!h]
    \centering
	\subfigure[Expressiveness vs $p$]{
		\includegraphics[angle=0, width=0.48\linewidth]{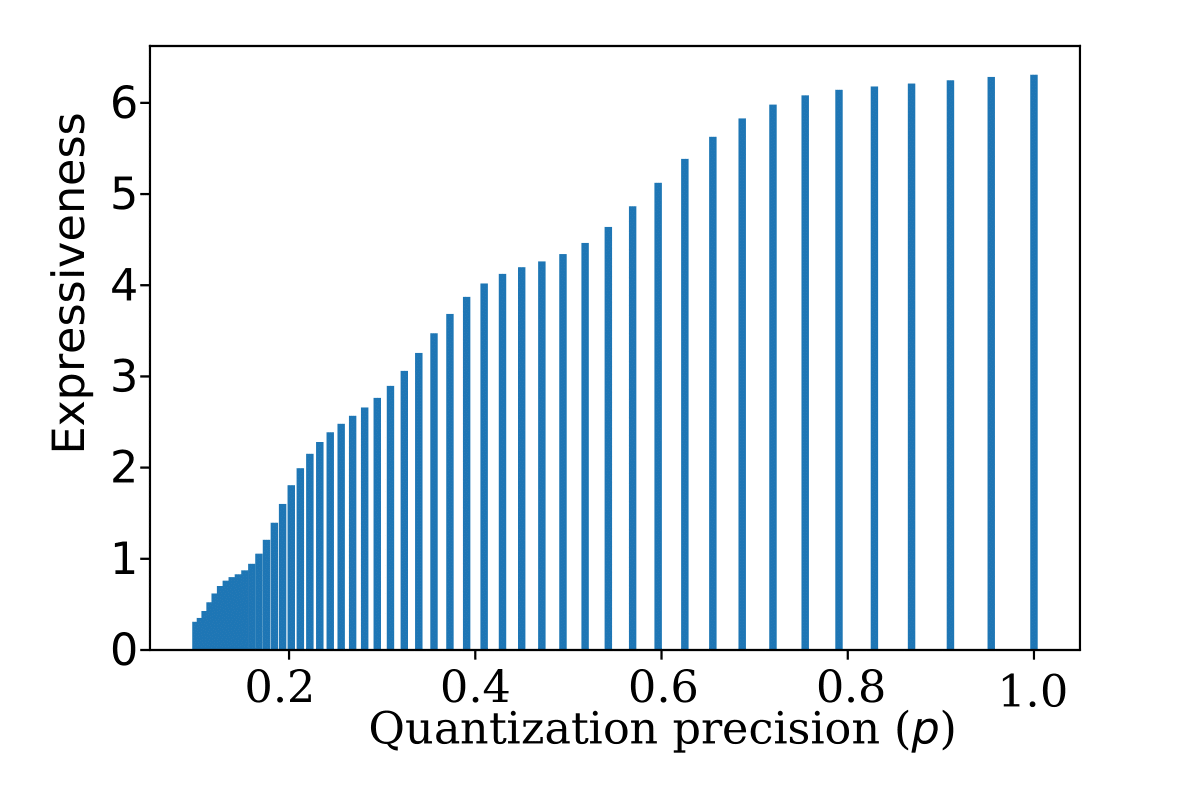}
		\label{fig:if_precision}}
	\hspace{-0.18in}
	\subfigure[Expressiveness vs $n_p$]{
	\includegraphics[angle=0, width=0.48\linewidth]{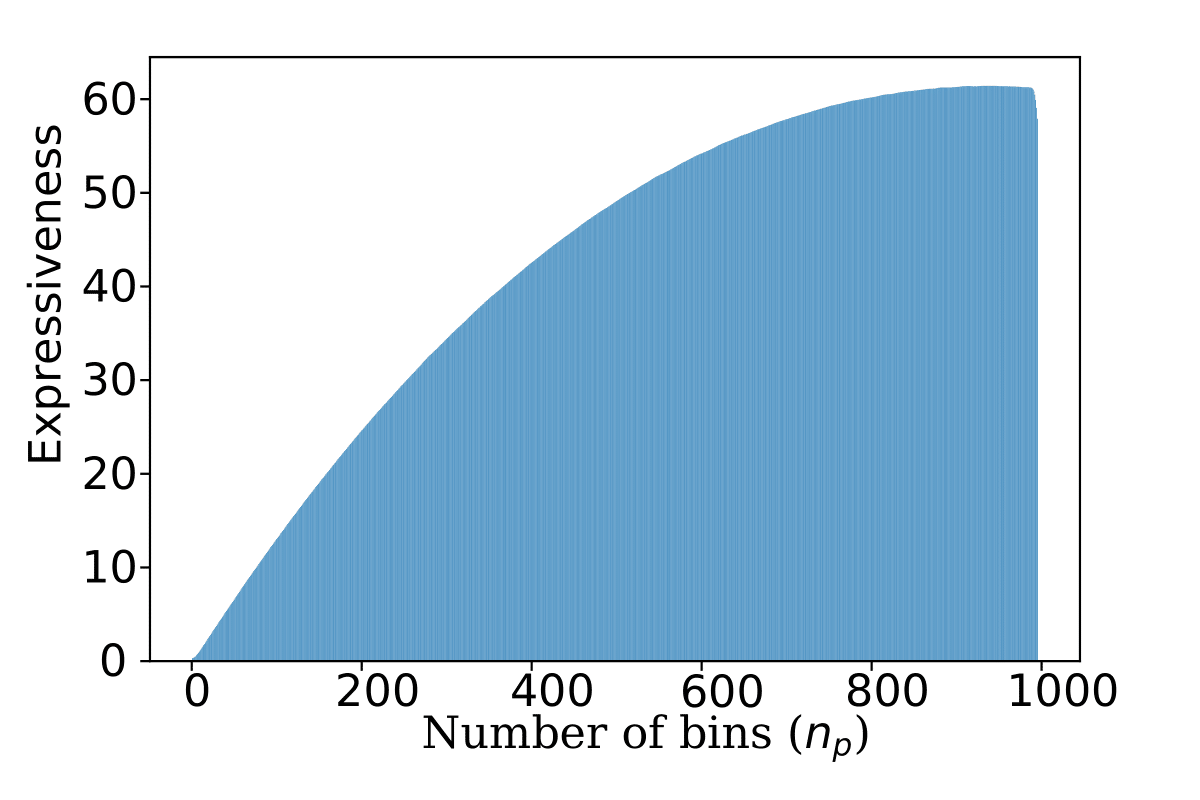}
	\label{fig:if_n}}
\caption{Information theory of $0$-DP synopsis generation.}\vspace{-0.05in}
\label{fig:multinomial}
\end{figure}

\begin{table*}[!h]
\footnotesize
\centering
\caption{Expanded sensitivity and privacy loss (actual DP) for the median query of packets in each time slot (0.5 seconds) on a network traffic dataset. $\epsilon=0.5$ is used to generate noise for existing methods ($\epsilon$ is always bounded by $0.5$ in DPI). 
}
\label{tab:privacy-results}
\begin{tabular}{|c|c|c|c|c|c|c|c|c|}
\hline
&Sensitivity after $5$ slots & Sensitivity after $50$ slots & Actual DP after $5$ slots & Actual DP after $50$ slots \\
\hline
Chan et al. \cite{chan} & 7 & 67 &  0.590 &  1.771 \\
\hline
Chen et al. \cite{pegasus} & 7 &  67 & 1.841  &  2.166 \\
\hline
Perrier et al. \cite{perrier2018private} & 7 &  67 &  1.734 &  2.248 \\
\hline
Wang et al. \cite{wang} & 7 & 67 &  1.674 &  1.972 \\
\hline
DPI (ours) & 2 & 2 & 0.001 &  0.005\\
\hline
\end{tabular}
\end{table*}

\section{Random Budget Allocation (Alternative)}
\label{sec:RBA2}

Random budget allocation with a non-depleting total budget output can also be achieved by dividing the total privacy budget $\epsilon$ into several ranges like \emph{high}, \emph{medium}, and \emph{low}. In each time slot $t$, $\epsilon_t$ is randomly selected from one of these ranges. Refer to Algorithm~\ref{algm:RBA2} for details of this alternative approach. In essence, the random allocation of the privacy budget allows us to cater to diverse and changing data streams, ensuring the high utility of DPI while preserving privacy.

\begin{algorithm}[!h]
\label{algm:RBA1}
\footnotesize
\SetAlgoLined
\caption{Random Budget Allocation (RBA)}
\KwData{available privacy budget space $\varepsilon_t = \{\epsilon_1, \epsilon_2, \ldots\}$}
\KwData{tuning parameter $\Lambda = 10^8$}
\tcc{$\Lambda$ is a large number tunable to scale of remaining epsilon values in $\varepsilon_t$}
\tcc{sample from exponential distribution}
$S \sim \text{Exp}(\Lambda)$\;

\tcc{find the closest remaining privacy budget to $S$}
$\epsilon_t \leftarrow \argmin_{\epsilon \in \varepsilon_t} |\epsilon - S|$\;

\tcc{update available privacy budget space }
$\varepsilon_{t+1} \leftarrow \varepsilon_t \setminus \{\epsilon_t\}$\;

\tcc{return the selected privacy budget}
\textbf{return} $\epsilon_t$\;
\end{algorithm}

\begin{algorithm}
\footnotesize
\caption{RBA: Alternative Approach}
\label{algm:RBA2}
initialize queues $smallNumbers$, $mediumNumbers$, $largeNumbers$

define the ranges: $smallRange$, $mediumRange$, $largeRange$

\For{each $i$ in $smallRange$, $mediumRange$, and $largeRange$}{
    correspondingQueue.enqueue($i$)
}

\SetKwFunction{FMain}{getUniqueRandomNumber}
    \While{true}{
        $randomQueue \leftarrow$ randomly choose from ($smallNumbers$, $mediumNumbers$, $largeNumbers$)
        
        \If{$randomQueue$ is not empty}{
             $number \leftarrow randomQueue$.dequeue()
             
            \Return $number$
        }
    }
\end{algorithm}

\section{Additional Results}
\subsection{Privacy Loss of DPI vs Existing Methods}
\label{sec:sotaloss}

We additionally conduct experiments to compare the privacy loss of DPI with representative existing methods  \cite{pegasus,chan,wang,perrier2018private}. Existing methods for differentially private data stream disclosures have several limitations. Most critically, they cannot preserve overall privacy loss in a bounded or converging manner. To prove this, we measure the expansion of sensitivity across the first 5 time slots and 50 time slots, respectively (in terms of the user-level DP protection) using the median query of packets in each time slot over a network traffic dataset in Table~\ref{tab:privacy-results}. In addition, we derive the privacy loss via the notion of R\'enyi differential privacy (e.g., deriving the R\'enyi differential privacy guarantees after injecting the noise and then converting the R\'enyi differential privacy to $\epsilon$-DP \cite{mironov2017renyi}). Figure \ref{fig:motivation} further demonstrates that existing methods cannot bound the privacy over time. Thus, we do not benchmark DPI with them in the infinite stream settings. 

\begin{figure}[!h]
    \centering
	\subfigure[0.5 seconds each time slot]{
		\includegraphics[angle=0, width=0.48\linewidth]{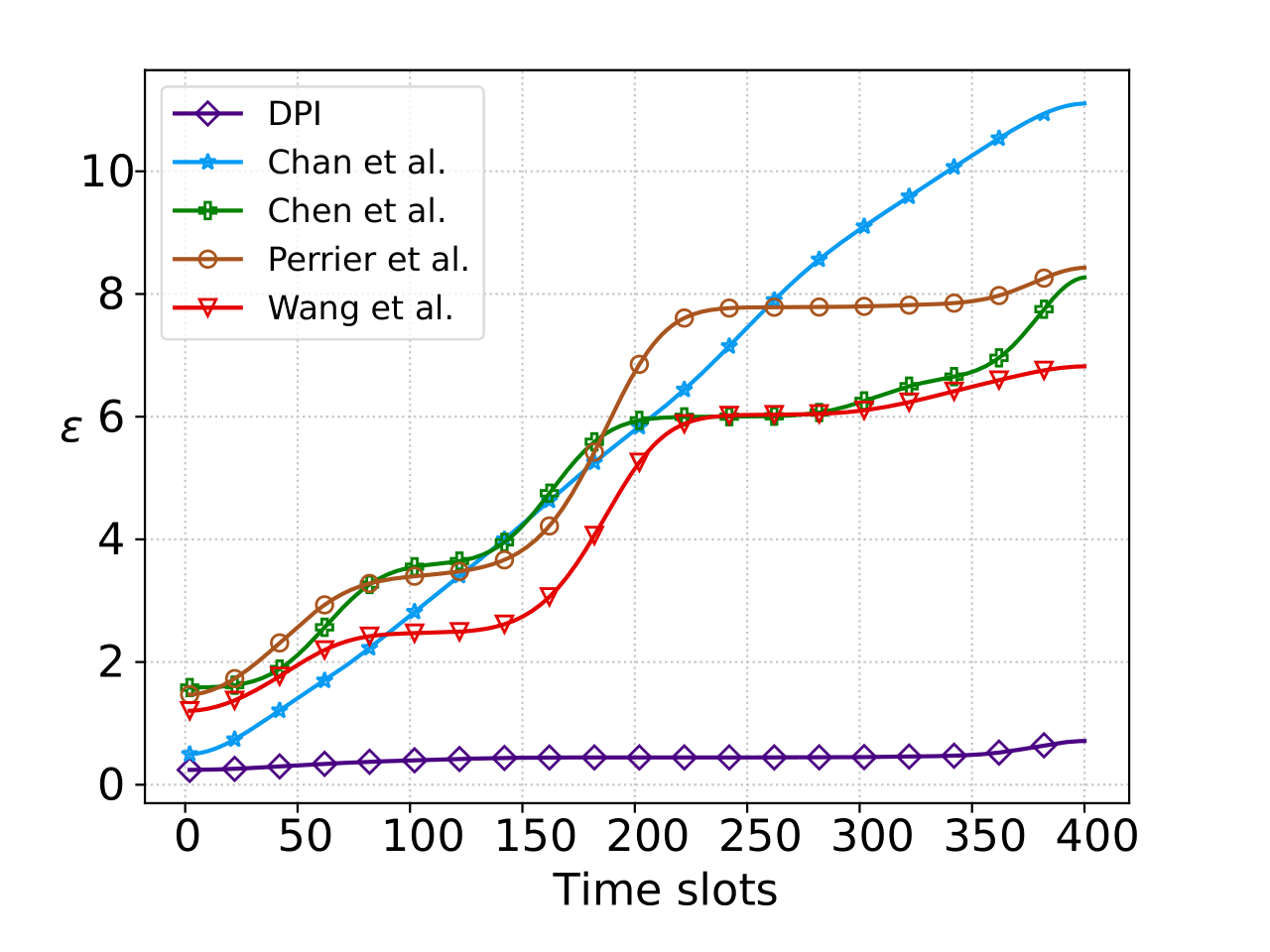}
		\label{fig:originl1}}
	\hspace{-0.18in}
	\subfigure[5 seconds each time slot]{
	\includegraphics[angle=0, width=0.48\linewidth]{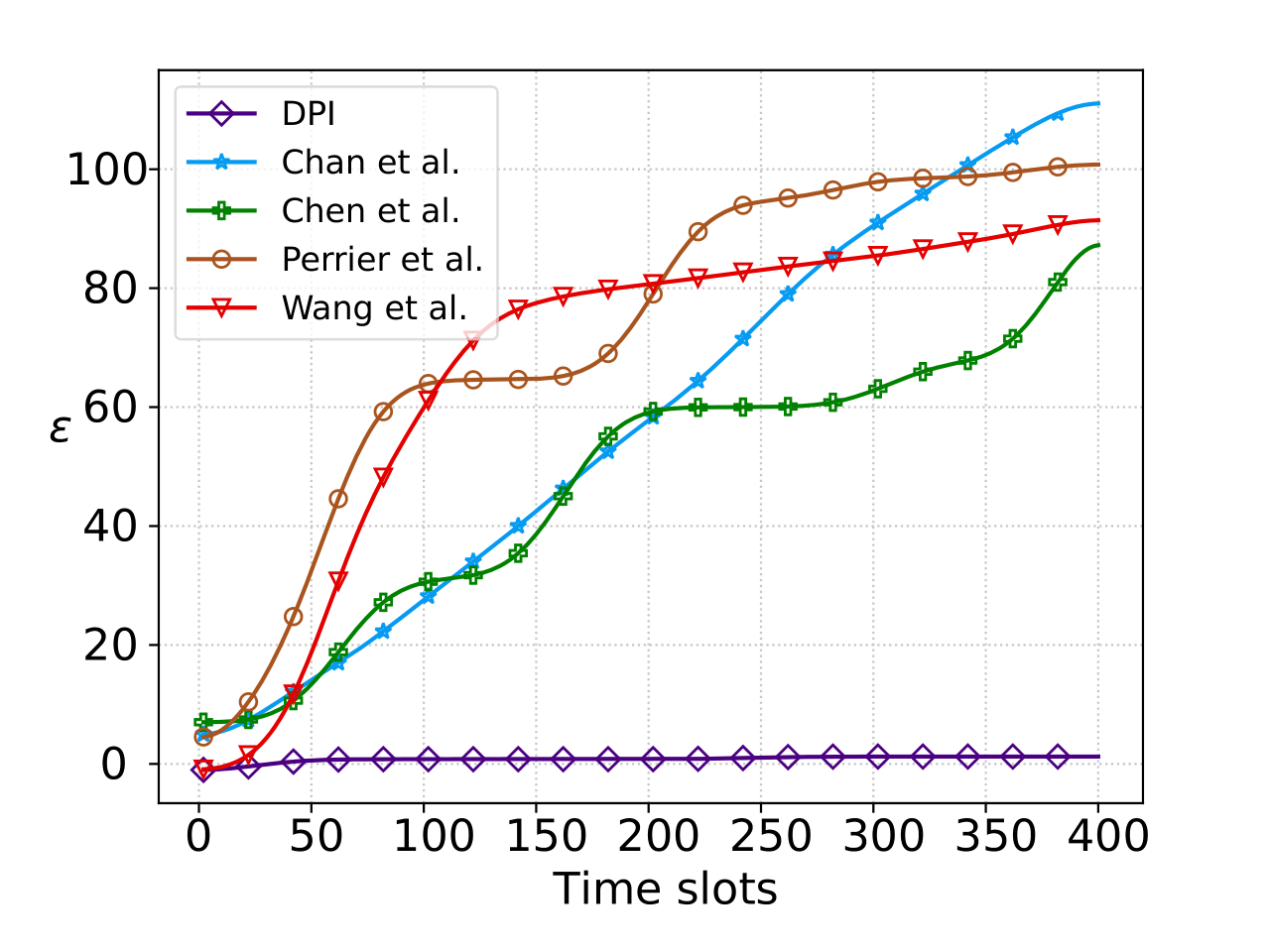}
	\label{fig:originl2}}
\caption{Total privacy loss of DPI and existing methods.}\vspace{-0.1in}
\label{fig:motivation}
\end{figure}

\subsection{Additional Results on Synthetic Data}
\label{sec:synthetic}

To further show that DPI can produce accurate output results on datasets with diverse characteristics, we conducted additional experiments on more synthetic data. 
Figure \ref{fig:box_plot_synthethic} shows the utility of DPI (w.l.o.g., using the accumulative query on the data distribution as an example) evaluated on additional $20$ synthetic datasets and different privacy budgets $\epsilon$ varying from $0.5$ to $10$. Out of $594,000$ query results ($2,700$ time slots, 11 different $\epsilon$, 10 datasets, 2 different domain sizes), DPI can still ensure low and stable MSE and KL divergence in all these different settings. These results demonstrate high utility on datasets with diverse characteristics.

\vspace{-0.1in}

\begin{figure}[!h]
    \centering
	\subfigure[MSE vs $t$]{
	\includegraphics[angle=0, width=0.48\linewidth]{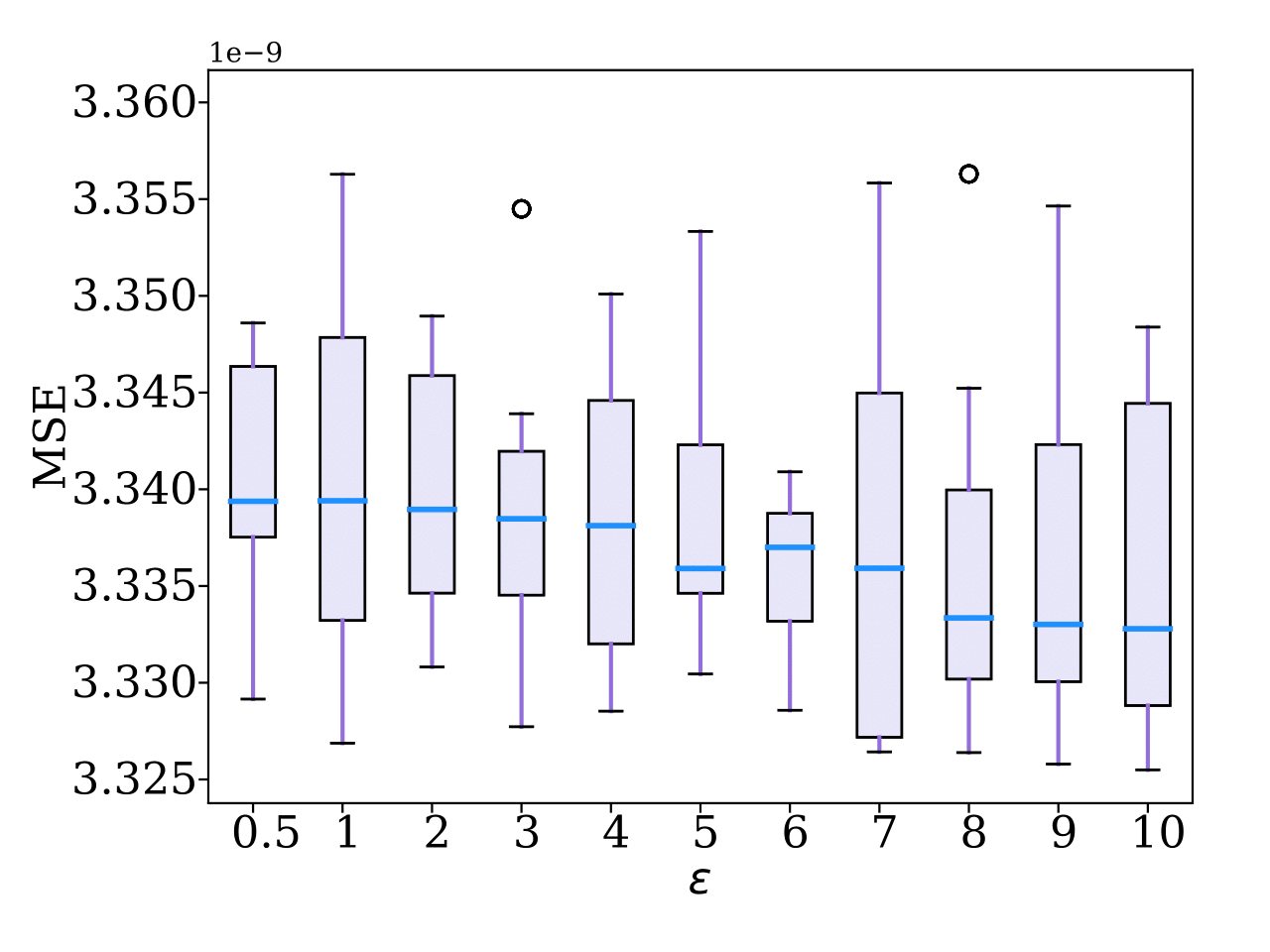}
	\label{fig:mse_box_rebuttal_1000}}
	\hspace{-0.18in}
\subfigure[KL divergence vs $t$]{
		\includegraphics[angle=0, width=0.48\linewidth]{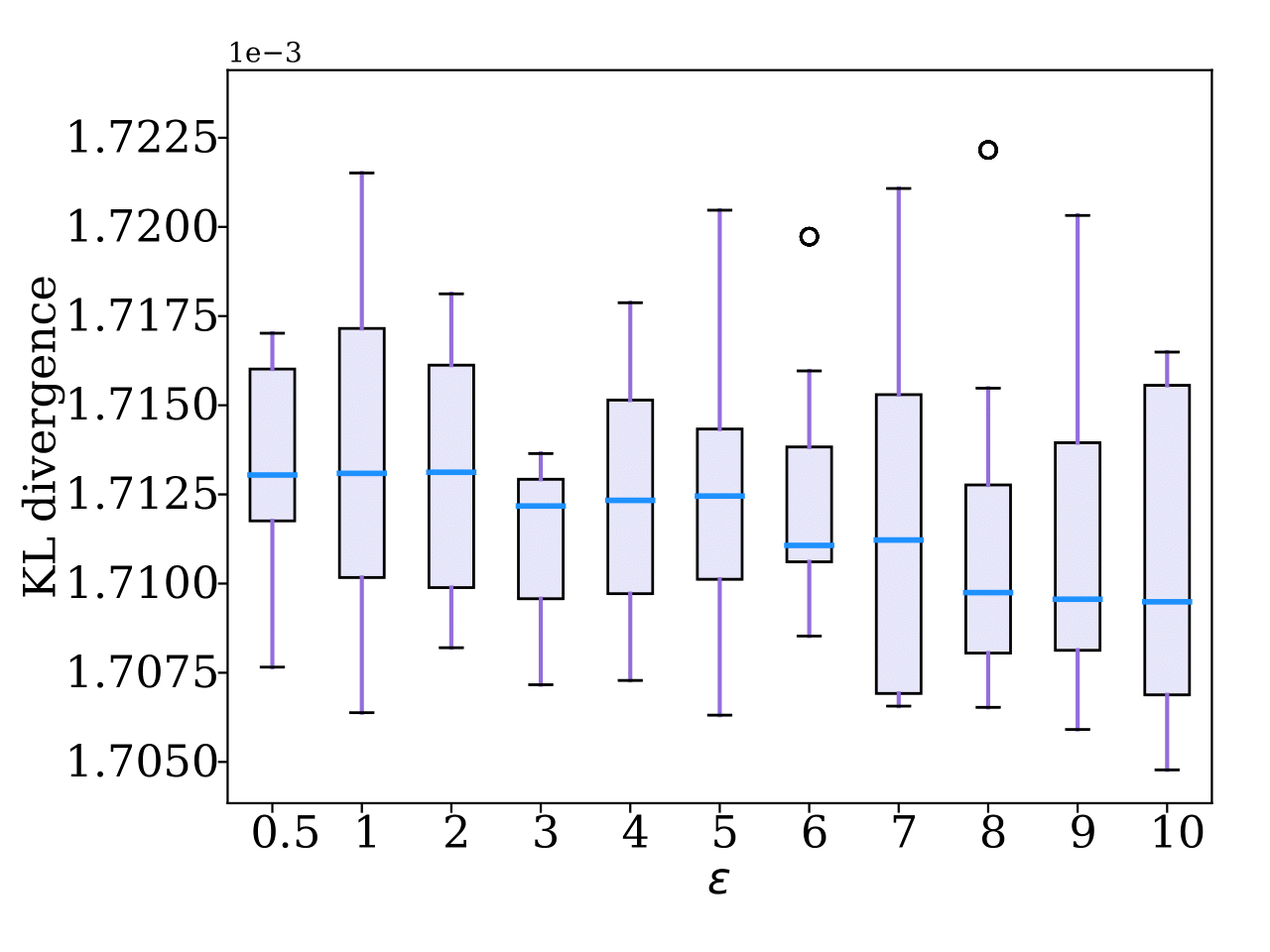}
		\label{fig:kl_box_rebuttal_1000}}
 \hspace{-0.18in}
\subfigure[MSE vs $t$]{
		\includegraphics[angle=0, width=0.48\linewidth]{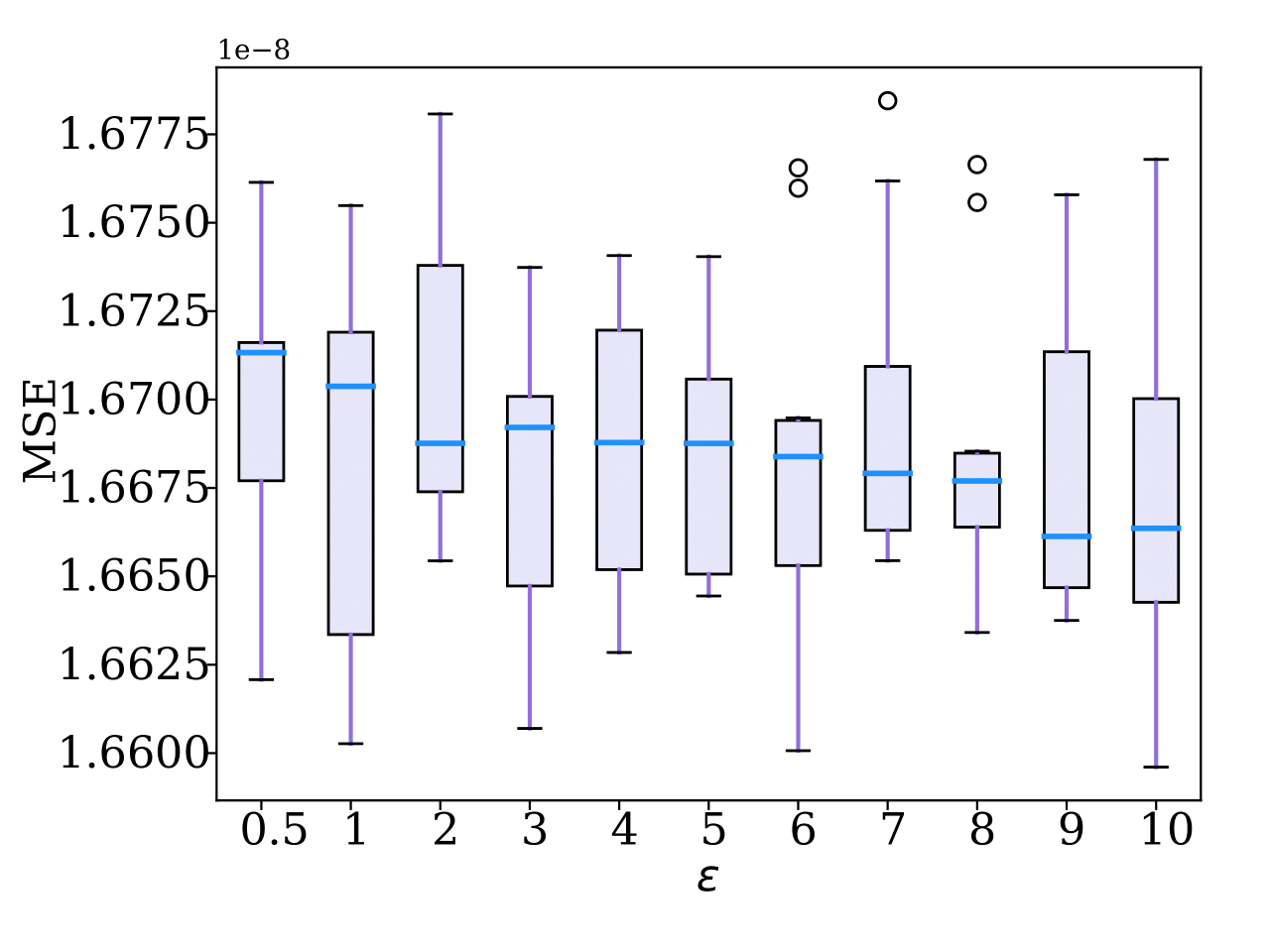}
		\label{fig:mse_box_rebuttal_100}}
  \hspace{-0.18in}
   \subfigure[KL divergence vs $t$]{
		\includegraphics[angle=0, width=0.48\linewidth]{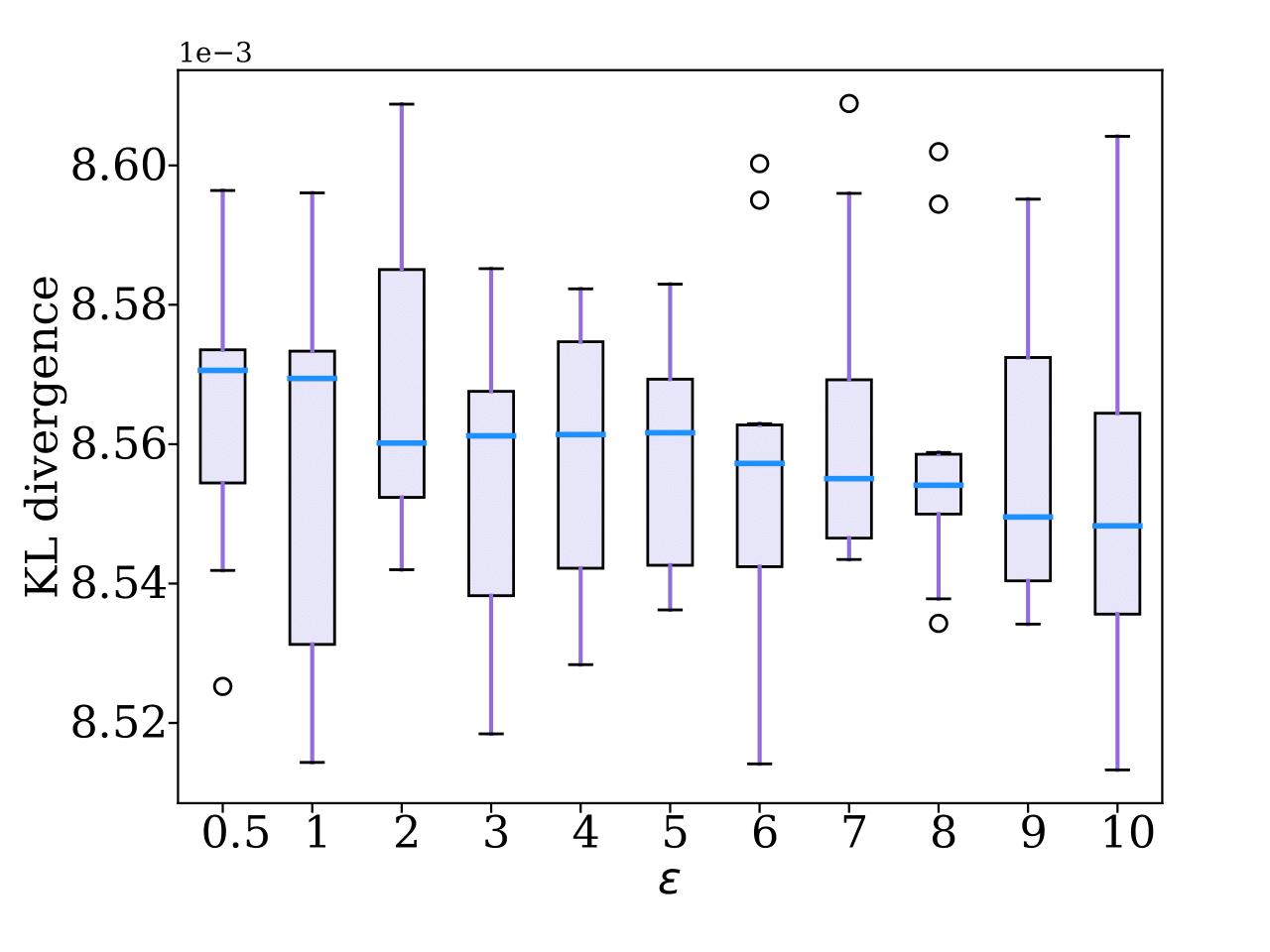}
		\label{fig:kl_box_rebuttal_100}}
\caption{Evaluation of DPI on $10$ synthetic datasets (each box includes $10$ results for a specific privacy bound $\epsilon$: the average MSE/KL of $2,700$ time slots in each of the 10 datasets). (a) (b) MSE and KL divergence: domain size $1,000$. (c) (d) MSE and KL divergence: domain size $10,000$. Each experiment is repeated for 10 times, and the average results are plotted in the figures.}
\label{fig:box_plot_synthethic}
\end{figure}

\newpage

\section{Meta-Review}

The following meta-review was prepared by the program committee for the 2024
IEEE Symposium on Security and Privacy (S\&P) as part of the review process as
detailed in the call for papers.

\subsection{Summary}

DPI is a framework for handling infinite data streams with differential privacy and bounded privacy leakage.

\subsection{Scientific Contributions}

\begin{itemize}
\item Provides a Valuable Step Forward in an Established Field

\end{itemize}

\subsection{Reasons for Acceptance}

The paper provides a valuable step forward in an established field. Handling streaming data with differential privacy, particularly infinite streams, is a known open research problem. The paper proposes the combination of several novel techniques to bound privacy loss, including: data independent synopses, DP boosting, and budget allocation from a converging infinite series.

\label{sec:appendix}

\end{document}